\documentclass[journal]{IEEEtran}%
\usepackage{graphicx}
\usepackage{epsfig}
\usepackage{times}
\usepackage{amsmath}
\usepackage{thmtools,thm-restate}
\usepackage{amssymb}
\usepackage[section]{placeins}
\usepackage{placeins}
\usepackage{amsmath}
\usepackage{amsthm}
\usepackage{multirow}
\usepackage{graphicx}
\usepackage[normalem]{ulem}
\usepackage{subcaption}
\usepackage{algorithm,tabularx}
\usepackage{algpseudocode}
\usepackage{amsfonts}
\usepackage{cases}
\usepackage{romannum}
\usepackage{textcomp}
\usepackage{xcolor}%
\providecommand{\U}[1]{\protect\rule{.1in}{.1in}}
\providecommand{\U}[1]{\protect\rule{.1in}{.1in}}

\newtheorem{assumption}{Assumption}
\newtheorem{theorem}{Theorem}
\newtheorem{corollary}{Corollary}
\newtheorem{lemma}{Lemma}
\newtheorem{proposition}{Proposition}
\newtheorem{remark}{Remark}
\newtheorem{definition}{Definition}

\useunder{\uline}{\ul}{}
\makeatletter
\newcommand{\multiline}[1]{  \begin{tabularx}{\dimexpr\linewidth-\ALG@thistlm}[t]{@{}X@{}}
#1
\end{tabularx}
}
\makeatother
\usepackage{cite}

\usepackage{picins}
\usepackage{centernot}
\let\labelindent\relax
\usepackage{enumitem}
\setlist[itemize]{leftmargin=*}
\usepackage{makecell}
\usepackage[normalem]{ulem}
\usepackage{cuted}
\usepackage{hyperref}
\usepackage{stfloats}
\usepackage{lscape}

\newcommand{\R}{\mathbb{R}}
\newcommand{\N}{\mathbb{N}}

\newcommand{\T}{\top}
\newcommand{\I}{\mathbf{I}}
\newcommand{\0}{\mathbf{0}}
\newcommand{\E}{\mathcal{E}}
\newcommand{\F}{\mathcal{F}}
\newcommand{\C}{\mathcal{C}}
\newcommand{\diag}{\text{diag}}
\newcommand{\tsup}[1]{\textsuperscript{#1}}
\newcommand{\mb}[1]{\mathbf{#1}}
\renewcommand{\H}{\mathcal{H}}
\newcommand{\tr}[1]{\mbox{tr}\left(#1\right)}
\newcommand{\bm}[1]{\begin{bmatrix}#1\end{bmatrix}}

\algdef{SE}[SUBALG]{Indent}{EndIndent}{}{\algorithmicend\ }%
\algtext*{Indent}
\algtext*{EndIndent}

\title{\LARGE \bf
A Decentralized Analysis and Control Synthesis Approach for Networked Systems with Arbitrary Interconnections
}

\author{Shirantha Welikala, Hai Lin and Panos Antsaklis 
\thanks{The support of the National Science Foundation (Grant No. IIS-1724070, CNS-1830335, IIS-2007949) is gratefully acknowledged.}
\thanks{The authors are with the Department of Electrical Engineering, College of Engineering, University of Notre Dame, IN 46556, \texttt{{\small \{wwelikal,hlin1,pantsakl\}@nd.edu}}.}}

\begin{document}

\maketitle
\thispagestyle{empty}
\pagestyle{empty}


\begin{abstract}
This paper considers the problem of decentralized analysis and control synthesis to verify and ensure properties like stability and dissipativity of a large-scale networked system comprised of linear subsystems interconnected in an arbitrary topology. In particular, we design systematic networked system analysis and control synthesis processes that can be executed in a decentralized manner at the subsystem level with minimal information sharing among the subsystems. Compared to our most recent work on the same topic, we consider a substantially more generalized problem setup in this paper and develop decentralized processes to verify and ensure a broader range of networked system properties. We show that for such decentralized processes: (1) optimizing the used subsystem indexing scheme can substantially reduce the required inter-subsystem information-sharing sessions, and (2) in some network topologies, information sharing among only neighboring subsystems is sufficient (hence, distributed!). Moreover, the proposed networked system analysis and control synthesis processes are compositional/resilient to subsystem removals, which enable them to conveniently and efficiently handle situations where new subsystems are being added/removed to/from an existing network. We also provide significant insights into our decentralized approach so that it can be quickly adopted to verify and ensure properties beyond the stability and dissipativity of networked systems. En route to developing such decentralized techniques, we have also derived new centralized solutions for dissipative observer and dynamic output feedback controller design problems. Subsequently, we also specialize all the derived results for discrete-time networked systems. We conclude this paper by providing several simulation results demonstrating the proposed novel decentralized analysis and control synthesis processes and dissipativity-based results.
\end{abstract}

\section{Introduction}

Analysis and control synthesis of large-scale networked systems comprised of dynamically coupled subsystems has gained a renewed attention due to various emerging applications in infrastructure networks \cite{Agarwal2021,Agarwal2019}. A prime example of this is found in vehicular networks where often a group of autonomous vehicles that co-ordinate with each other to maintain a particular formation can lead to saving energy as well as reducing congestion and improving safety in the transportation infrastructure \cite{Jovanovic2013,Karafyllis2021}. Another example application of large-scale networked systems is the power grid, which needs to be constantly analyzed and controlled in a distributed manner while accounting for the penetration of highly varying renewable energy sources and smart/unknown loads \cite{Welikala2016b,Shahid2017,Samad2017}.

Numerous distributed control solutions have already been proposed in the literature concerning such large-scale networked dynamical systems to enforce stability while also optimizing various performance objectives of interest \cite{Antonelli2013,Xue2013,Davison1990,Siljak1991,Siljak1978}. These solutions synthesize local controllers (i.e., at the subsystem level) that only require the state information of a subset of other subsystems in the network to operate. However, many of such distributed control solutions assume the existence of a central entity with the knowledge of the entire networked system to execute the control synthesis process in a centralized manner. Inherently, such a centralized setup has several disadvantages: (1) feasibility, privacy and security concerns related to collecting all subsystem information at a centralized entity, (2) having to re-evaluate the entire analysis and/or control synthesis when new subsystems are added (or removed) to/from the network, and (3) scalability concerns arising due to the complexity of the problem that needs to be solved at the centralized entity.

Several decentralized control synthesis approaches have also been proposed to address such limitations in the literature. In such methods, controllers are derived locally at subsystems without the explicit knowledge of the dynamics of the other subsystems. As pointed out in \cite{Agarwal2021}, such decentralized control synthesis approaches can be categorized into three groups: 
(i) Approaches that induce and exploit weak coupling between subsystems \cite{Michel1983,Bakule1988,Sezer1986};
(ii) Hierarchical approaches that compute (in a centralized manner) and enforce additional conditions on local control synthesis \cite{Ishizaki2021,Zheng2017}; 
(iii) Approaches that decompose a centralized control synthesis process using numerical techniques such as methods of multipliers and Sylvester's criterion \cite{DAndrea2003,Massioni2009,Riverso2014,Agarwal2019,Agarwal2021}. %
In many networked systems of interest, assuming or enforcing weak coupling among subsystems is not practical \cite{Agarwal2019}. Moreover, existing hierarchical approaches are computationally intensive and still involve a considerable centralized component. Taking these limitations into account, the recent work in \cite{Agarwal2019,Agarwal2021} has developed a decentralized analysis and control synthesis framework for networked dynamical systems inspired by Sylvester's criterion.

In particular, both \cite{Agarwal2019,Agarwal2021} assume each subsystem dynamics to be linear and coupled with a subset of neighboring subsystems (determined by the network topology) through their state values. In \cite{Agarwal2019}, the network topology is assumed to be cascaded bi-directional, and distributed analysis and control synthesis techniques have been developed to verify and ensure the passivity of the networked system. This approach is then further extended in \cite{Agarwal2021} considering arbitrary bi-directional network topologies to verify and ensure a general quadratic dissipativity property (known as $(Q,S,R)$-dissipativity \cite{Willems1972a}) over the networked system in a decentralized manner.   

Compared to \cite{Agarwal2021}, the recent work in \cite{Arcak2022,Arcak2016} assume each subsystem dynamics to be non-linear and coupled with a set of neighboring subsystems through their output (not the state) values. Then, dissipativity properties of the subsystems are exploited to derive centralized stability and performance verification techniques in a compositional manner. Therefore, unlike in \cite{Agarwal2019,Agarwal2021}, the techniques proposed in \cite{Arcak2022,Arcak2016} are only applicable for centralized analysis of networked systems.

On the other hand, compared to \cite{Agarwal2021}, a few application-specific (structure) decentralized control synthesis approaches can also be found in the literature. For example, \cite{Stilwell2002,Wang2020,DAndrea1998,Lu2015} and \cite{Bharadwaj2021} respectively address the decentralized control synthesis problems for vehicular platoons, switched networked systems, distributed parameter systems, micro-grids and unmanned aerial systems (see also the review in \cite{Bakule2012} and references therein). However, these approaches either exploit special structural properties specific to the considered problem setup or belong to one of the earlier mentioned three groups. 

It should also be noted that several works in the literature refer to centrally executed distributed (or fully-local) control synthesis as ``decentralized control synthesis'' (deviating from our nomenclature). For instance, \cite{Lessard2016,Naghnaeian2018} and \cite{Attia2010,Attia2014} respectively propose Youla parametrization and orthogonal functions based approaches for centrally synthesize distributed controllers - where as \cite{Agarwal2021} propose Sylvester's criterion based approach for decentrally synthesize distributed controllers. 

\paragraph{\textbf{Contributions}} 
Taking these concerns into account, in this paper, we further generalize the decentralized analysis and control synthesis approach proposed in \cite{Agarwal2021,Agarwal2022}. Compared to \cite{Agarwal2022}, our main contributions can be outlined as follows: 
\begin{enumerate}
\item We provide a comprehensive collection of linear matrix inequality (LMI) based solutions for several standard control problems associated with continuous-time linear time-invariant systems (CT-LTI). 
\item To the best of the authors' knowledge, this collection includes novel (only known to date) LMI-based solutions for $(Q,S,R)$-dissipative observer synthesis and $(Q,S,R)$-dissipative dynamic output feedback controller synthesis problems for CT-LTI systems.
\item We consider a fully-coupled continuous-time networked system (CTNS) model (in terms of subsystem states, inputs, and disturbances) and only introduce decoupling assumptions critical to addressing each interested CTNS control problem in a decentralized manner.
\item We have relaxed the bi-directionality assumption made regarding the network topology. 
\item The distributed observer synthesis problem has been considered, and a decentralized solution has been proposed. 
\item The distributed dynamic output feedback controller synthesis problem has been considered, and a decentralized solution has been proposed. 
\item We propose decentralized analysis and control synthesis approaches to verify and enforce: (1) exponential stability and (2) optimal $\H_2/\H_\infty$ gains of a CTNS; 
\item We also study the effect of subsystem indexing (i.e., the order in which the proposed decentralized processes are executed) on the required total communications among the subsystems and derive a cost function that can be optimized to avoid some costly inter-subsystem communications;
\item We provide significant insights into our approach so that it may be quickly adopted to address similar LMI-based control problems associated with CTNSs.
\item All the derived results for CTNSs have also been specialized for discrete-time networked systems (DTNSs).
\end{enumerate}


\paragraph{\textbf{Organization}} 
This paper is organized as follows. 
In Section \ref{Sec:BasicsOfCTLTISystems}, we summarize a comprehensive collection of linear matrix inequality (LMI) conditions that arise in the analysis and control synthesis of continuous-time linear time-invariant (CT-LTI) systems. 
The details of the construction of the considered continuous-time networked system (CTNS) and the problem formulation are discussed in Section \ref{Sec:CTNetworkedSystem}. 
Next, in Section \ref{Sec:DecentralizedAnalysis}, we define a class of matrices specific to a given network topology (called ``network matrices'') and discuss several related theoretical results along with a decentralized algorithm to analyze the positive-definiteness of such network matrices. 
Subsequently, in Section \ref{Sec:DistributedTechniquesForCTNS}, we provide the details of the proposed decentralized analysis and local control synthesis processes for the CTNS. 
The discrete-time versions of the concepts and results provided in Sections \ref{Sec:BasicsOfCTLTISystems}, \ref{Sec:CTNetworkedSystem} and \ref{Sec:DistributedTechniquesForCTNS} are summarized in subsequent Sections  \ref{Sec:BasicsOfDTLTISystems}, \ref{Sec:DTNetworkedSystem} and \ref{Sec:DistributedTechniquesForDTNS}, respectively. 
Finally, Section \ref{Sec:SimulationResults} discusses several simulation results before concluding the paper in Section \ref{Sec:Conclusion}.

\paragraph{\textbf{Notation}}
The sets of real and natural numbers are denoted by $\R$ and $\N$, respectively. An $n$-dimensional real vector is denoted by $\R^n$. We define $\N_N\triangleq\{1,2,\ldots,N\}$ where $N\in\N$. 
An $n\times m$ block matrix $A$ can be represented as $A=[A_{ij}]_{i\in\N_n, j\in\N_m}$ (or simply as $[A_{ij}]$) where $A_{ij}$ is the $(i,j)$\tsup{th} block of $A$. 
Similarly, $[A_{ij}]_{j\in \N_m}$ represents a block row matrix, $\diag(A_{ii}:i\in\N_n)$ represents a block diagonal matrix and unless defined otherwise, $A_i \triangleq \{A_{ii}\}\cup\{A_{ij},j\in\N_{i-1}\}\cup\{A_{ji}:j\in\N_i\}$.
The transpose of a matrix $A$ is denoted by $A^\T$ and $(A^\T)^{-1} = A^{-\T}$, The zero matrix is denoted by $\0$ and the standard identity matrix is denoted by $\I$ (dimensions will be clear form the context). A symmetric positive definite (semi-definite) matrix $A\in\R^{n\times n}$ is represented as $A=A^\T>0$ ($A=A^\T \geq 0$). Unless stated otherwise, $A>0 \iff A=A^\T>0$ (i.e., symmetry is implied by the positive definiteness). The symbol $\star$ is used to represent redundant conjugate block matrices (e.g., 
$\bm{A & B \\ \star & C} = \bm{A & B \\ B^\T & C}$ and 
$\bm{A_{ij} & B_{ij} \\ \star & C_{ij}} = \bm{A_{ij} & B_{ij} \\ B_{ji}^\T & C_{ij}}$). 
The symmetric part of a matrix $A$ is denoted by $\H_s(A) \triangleq A+A^\T$ and $\H_s(A_{ij}) \triangleq A_{ij}+A_{ji}^\T$.
$\mathcal{L}_{2e}$ is the extended $\mathcal{L}_2$ space (i.e., the space of signals with finite $\mathcal{L}_2$ norms). 
Given sets $A$ and $B$, $A\backslash B$ indicates the set subtraction operation that results in the set of elements in $A$ that are not in $B$. The notation $\mb{1}_{\{\cdot\}}$ is used to represent the indicator function and $e_{ij} \triangleq \I \cdot \mb{1}_{\{i=j\}}$.

\section{Preliminaries: Continuous-Time Linear Time Invariant (CT-LTI) Systems}
\label{Sec:BasicsOfCTLTISystems}

In this section, we present a comprehensive collection of linear matrix inequality (LMI) conditions that arise when analyzing or synthesizing controllers for continuous-time linear time-invariant (CT-LTI) systems. 
We particularly focus on LMIs because: (1) we will subsequently propose a systematic approach to decentralize such LMIs (in Sec. \ref{Sec:DecentralizedAnalysis}) and (2) LMIs can be efficiently and conveniently solved using standard convex optimization techniques \cite{Boyd1994}. We start with stating two well-known lemmas that will be useful when deriving LMIs. 

\begin{lemma} (Schur's complement \cite{Bernstein2009}) \label{Lm:TwoByTwoBlockMatrixPDF}
Let $W=\scriptsize \begin{bmatrix}\Theta & \Phi \\ \Phi^\T & \Gamma \end{bmatrix}$ be a symmetric $2\times2$ block matrix. Then if: \\
(i) $\Theta$ is invertible, $W>0 \iff \Theta > 0,\ \Gamma-\Phi^\T \Theta^{-1} \Phi > 0$,\\
(ii) $\Gamma$ is invertible, $W>0 \iff \Gamma>0,\ \Theta-\Phi \Gamma^{-1} \Phi^\T>0$. \end{lemma}

\begin{lemma} (Congruence principle \cite{Bernstein2009}) \label{Lm:PreAndPostMultiplication}
A matrix $W > 0$ if and only if $P^\T WP > 0$ where $P$ is a full-rank matrix.
\end{lemma}


Consider the CT-LTI system given by 
\begin{equation}\label{Eq:CTLTISystem}
\begin{aligned}
    \dot{x}(t) = Ax(t) + Bu(t),\\
    y(t) = Cx(t) + Du(t),
\end{aligned}
\end{equation}
where $x(t)\in\R^n,u(t)\in\R^p$ and $y(t)\in\R^m$ respectively represents the state, input and output at time $t\in\R_{\geq0}$.

\subsection{Analysis of CT-LTI Systems}

\subsubsection{\textbf{Stability}}
A well-known necessary and sufficient LMI condition for the \emph{stability} of \eqref{Eq:CTLTISystem} (under $u(t)=\0$) is given in the following proposition.

\begin{proposition}\cite{Antsaklis2006} \label{Pr:CTLTIStability}
The CT-LTI system \eqref{Eq:CTLTISystem} under $u(t)=\0$ is globally exponentially stable iff $\exists P > 0$ such that 
\begin{equation}\label{Eq:Pr:CTLTIStability}
    -A^\T P - PA > 0.
\end{equation}
\end{proposition}

In the interest of brevity, we omit providing standard definitions of uniform and exponential stability, which can be found in \cite{Antsaklis2006,Khalil1996}. In the remainder of this paper, by `stability,' we simply refer to the global exponential stability.

\subsubsection{\textbf{$(Q,S,R)$-Dissipativity}} 
Similar to the stability, the \emph{dissipativity} property introduced in the seminal paper \cite{Willems1972a} is an important property of dynamical systems. In this paper, we particularly consider the quadratic dissipativity property called $(Q,S,R)$-dissipativity \cite{Kottenstette2014} defined below. 

\begin{definition} \cite{Kottenstette2014} \label{Def:CTLTIQSRDissipativity}
The CT-LTI system \eqref{Eq:CTLTISystem} is $(Q,S,R)$-dissipative (from $u(t)$ to $y(t)$) if there exists a positive definite function $V(x):\R^n\rightarrow\R_{\geq0}$ called the storage function such that for all $t_1 \geq t_0 \geq 0, x(t_0) \in \R^n$ and $u(t)\in\R^p$, the inequality
\begin{equation*}
    V(x(t_1))-V(x(t_0)) \leq \int_{t_0}^{t_1} 
    \bm{
    y(t)\\u(t)
    }^\T
    \bm{
    Q & S \\ S^\T & R
    }
    \bm{
    y(t)\\u(t)
    }dt
\end{equation*}
holds, where $Q\in\R^{m \times m}, S\in \R^{m \times p}$, $R\in\R^{p\times p}$ are given.
\end{definition}

Analyzing and enforcing this $(Q,S,R)$-dissipativity property on large-scale networked systems is a central objective of this paper as through appropriate choices of $Q, S$ and $R$ matrices, it can capture a wide range of dynamical properties of interest as summarized in the following remark.  

\begin{remark}\hspace{-2mm}\cite{Kottenstette2014} \label{Rm:QSRDissipativityVariants}
The dynamical system \eqref{Eq:CTLTISystem} satisfying Def.\ref{Def:CTLTIQSRDissipativity}:
\begin{enumerate}
    \item is \emph{passive} iff 
    $Q=0, S=\frac{1}{2}\I, R=0$;
    \item is \emph{strictly input passive} iff 
    $Q=0, S=\frac{1}{2}\I, R=-\nu \I$ where $\nu>0$ ($\nu$ is an input feedforward passivity index);
    \item is \emph{strictly output passive} iff 
    $Q=-\rho \I, S=\frac{1}{2}\I, R=0$ where $\rho>0$ ($\rho$ is an output feedback passivity index);
    \item is \emph{strictly passive} iff 
    $Q=-\rho \I, S=\frac{1}{2}\I, R=-\nu \I$ where $\rho,\nu>0$;
    \item is \emph{$\mathcal{L}_2$-stable} iff 
    $Q=-\frac{1}{\gamma}\I, S=0, R=-\gamma\I$ where $\gamma \geq 0$ ($\gamma$ is an \emph{$\mathcal{L}_2$-gain} of the system);
    \item \emph{conic} iff 
    $Q=-\I, S=c\I, R=(r^2-c^2)\I$ where $c\in\R$ and $r>0$ 
    ($c$ and $r$ are conic parameters); and
    \item is \emph{sector bounded} iff 
    $Q=-\I, S=(a+b)\I, R=-ab\I$ where $a,b\in\R$ ($a,b$ are sector bound parameters).
\end{enumerate}
\end{remark}

A necessary and sufficient LMI condition for the $(Q,S,R)$-dissipativity of \eqref{Eq:CTLTISystem} is established in the following proposition. 
 
\begin{proposition}
\label{Pr:CTLTIQSRDissipativity}
The CT-LTI system \eqref{Eq:CTLTISystem} is $(Q,S,R)$-dissipative (with $-Q>0,R=R^\T$) from $u(t)$ to $y(t)$ iff $\exists P>0$ such that 
\begin{equation}\label{Eq:Pr:CTLTIQSRDissipativity}
    \begin{bmatrix}
    -A^\T P - P A     &  -PB + C^\T S            & C^\T \\
    \star          & D^\T S + S^\T D + R    & D^\T \\ 
    \star & \star & -Q^{-1}
    \end{bmatrix} \geq 0.
\end{equation}
\end{proposition}

\begin{proof}
Using Lm. \ref{Lm:TwoByTwoBlockMatrixPDF} in \eqref{Eq:Pr:CTLTIQSRDissipativity}, we get:
\begin{equation}\label{Eq:Pr:CTLTIQSRDissipativityStep1}
    \begin{bmatrix}
    -A^\T P - P A + \hat{Q} & -PB + \hat{S}\\
    \star         & \hat{R}
    \end{bmatrix} \geq 0,
\end{equation}
with
$ 
\hat{Q} = C^\T Q C,\ \ 
\hat{S} = C^\T S + C^\T Q D \mbox{ and }\ \ 
\hat{R} = D^\T Q D + D^\T S + S^\T D + R
$. Hence the proof is complete \cite[Lm. 2]{Kottenstette2014} (note also that \eqref{Eq:Pr:CTLTIQSRDissipativityStep1} can be used instead of \eqref{Eq:Pr:CTLTIQSRDissipativity} if $Q\geq 0$). 
\end{proof}

\subsubsection{\textbf{$\H_2$-Norm}}
Let $\mathcal{G}:\mathcal{L}_{2e} \rightarrow \mathcal{L}_{2e}$ represent the transfer matrix of the CT-LTI system \eqref{Eq:CTLTISystem} from $u(t)$ to $y(t)$. If $A$ in \eqref{Eq:CTLTISystem} is Hurwitz, $\mathcal{G}(s)=C(s\I-A)^{-1}B+D$ and its $\H_2$-norm is 
\begin{equation}
    \Vert \mathcal{G} \Vert_{\H_2}^2 \triangleq \sup_{h>0}\frac{1}{2\pi}\int_{-\infty}^{\infty}\tr{\mathcal{G}^\T(j\omega+h)\mathcal{G}(j\omega+h)}d\omega.  
\end{equation}
As shown in \cite{Caverly2019}, under $D=\0$, the $\H_2$-norm of $\mathcal{G}$ is 
\begin{equation}
    \Vert \mathcal{G} \Vert_{\H_2}^2 = \tr{B^\T M B^\T}  = \tr{C N C^\T}, 
\end{equation}
where $M,N>0$ with $MA + A^\T M + C^\T C = 0$ and $AN + NA^\T + BB^\T=0$ . Consider the following proposition.

\begin{proposition} \hspace{-2mm} \cite[pp. 58]{Caverly2019}\label{Pr:H2Norm}
The $\H_2$-Norm of the transfer matrix of the CT-LTI system \eqref{Eq:CTLTISystem} (i.e., $\Vert \mathcal{G} \Vert_{\H_2}$) is $\Vert \mathcal{G} \Vert_{\H_2} < \gamma < \infty$  iff  $D=\0$ and $\exists P,Q>0$ and $\gamma\in\R_{>0}$ such that 
\begin{equation}\label{Eq:Pr:H2Norm1}
\bm{-AP-PA^\T & -PC^\T \\ \star & \gamma \I}>0,\ 
\bm{P & B \\ \star & Q}>0,\ 
\tr{Q}<\gamma,
\end{equation}
or
\begin{equation}\label{Eq:Pr:H2Norm2}
\bm{-A^\T P - PA & -PB \\ \star & \gamma \I}>0,\ 
\bm{P & C^\T \\ \star & Q}>0,\ 
\tr{Q}<\gamma.
\end{equation}
\end{proposition}

\subsubsection{\textbf{$H_\infty$-Norm}}

The $\H_\infty$-norm of $\mathcal{G}$ is formally defined as 
\begin{equation}
    \Vert \mathcal{G} \Vert_{\H_\infty}^2 \triangleq \frac{1}{2\pi} \int_{-\infty}^{\infty}\tr{G^\T(j\omega)G(j\omega)}d\omega 
\end{equation}
($\Vert \mathcal{G} \Vert_{\H_\infty}=\Vert \mathcal{G} \Vert_{\mathcal{L}_2}$). As shown in \cite[pp. 210]{Khalil1996} and \cite[pp.49]{Caverly2019},
\begin{equation}
    \Vert \mathcal{G} \Vert_{\H_\infty}^2 
    =\sup_{u\in\mathcal{L}_2,u\neq0} \frac{\Vert y \Vert_{\mathcal{L}_2}^2}{\Vert u \Vert_{\mathcal{L}_2}^2}
    =\sup_{\omega \in \R} \Vert G(j\omega)\Vert_2^2.
\end{equation}
Now, consider the following proposition.

\begin{proposition}\cite[pp.50]{Caverly2019}\label{Pr:HInfNorm}
The $\H_\infty$-Norm of the transfer matrix of the CT-LTI system \eqref{Eq:CTLTISystem} (i.e., $\Vert \mathcal{G} \Vert_{\H_\infty}$) is $\Vert \mathcal{G} \Vert_{\H_\infty} < \gamma$  iff  $\exists P>0$ and $\gamma\in\R_{>0}$ such that 
\begin{equation}\label{Eq:Pr:HInfNorm}
\bm{-A^\T P-PA & -PB & -C^\T \\ \star & \gamma \I & -D^\T \\ \star & \star & \gamma \I}>0.
\end{equation}
\end{proposition}

\subsubsection{\textbf{Controllability}}
Regarding the \emph{controllability} \cite{Antsaklis2006} of the CT-LTI system \eqref{Eq:CTLTISystem}, consider the following proposition. 
\begin{proposition}\cite[pp.76]{Turner2007} \cite[pp. 86]{Caverly2019} \label{Pr:Stabilizability}
A CT-LTI system \eqref{Eq:CTLTISystem} is: (1) stable and controllable iff $\exists P>0$ such that 
\begin{equation}
    -AP-PA^\T-BB^\T = 0,
\end{equation}
and (2) stabilizable iff $\exists P>$ such that
\begin{equation}\label{Eq:Pr:Stabilizability}
    -AP-PA^\T+BB^\T > 0 
\end{equation}
(here, under $K=-\frac{1}{2} B^\T P^{-1}$, $A+BK$ is Hurwitz).
\end{proposition}

\subsubsection{\textbf{Observability}}
Regarding the \emph{observability} \cite{Antsaklis2006} of the CT-LTI system \eqref{Eq:CTLTISystem}, consider the following proposition. 
\begin{proposition}\cite[pp.76]{Turner2007}\cite[pp. 87]{Caverly2019}\label{Pr:Detectability}
A CT-LTI system \eqref{Eq:CTLTISystem} is: (1) stable and observable iff $\exists P>0$ such that 
\begin{equation}
    -A^\T P -PA - C^\T C = 0,
\end{equation}
and (2) detectable iff $\exists P>$ such that
\begin{equation}\label{Eq:Pr:Detectability}
    -A^\T P-PA + C^\T C > 0 
\end{equation}
(here, under $L=\frac{1}{2} P^{-1} C^\T$, $A-LC$ is Hurwitz).
\end{proposition}

\subsection{Full-State Feedback (FSF) Controller Synthesis for CT-LTI Systems}\label{SubSec:FSFControllerSynthesis}

Consider the CT-LTI system  \eqref{Eq:CTLTISystem} with noise $w(t)\in\R^q$:
\begin{equation}\label{Eq:NoisyCTLTISystem}
\begin{aligned}
    \dot{x}(t) = Ax(t) + Bu(t) + Ew(t),\\
    y(t) = Cx(t) + Du(t) + Fw(t).
\end{aligned}
\end{equation}
Under full-state feedback (FSF) control $u(t)=Kx(t)$, the closed-loop CT-LTI system takes the form
\begin{equation}\label{Eq:CTLTIUnderFSF}
    \begin{aligned}
        \dot{x}(t) = (A+BK)x(t)+Ew(t),\\
        y(t) = (C+DK)x(t)+Fw(t).
    \end{aligned}
\end{equation}


\subsubsection{\textbf{Stabilization}}
The following proposition gives an LMI condition that leads to synthesizing a FSF controller $K$ such that the closed-loop system  \eqref{Eq:CTLTIUnderFSF} is stabilized.

\begin{proposition}\label{Pr:StabilizationUnderFSF}
Under $D=w(t)=\0$, the closed-loop CT-LTI system \eqref{Eq:CTLTIUnderFSF} is stable iff $\exists M>0$ and $L$ such that 
\begin{equation}\label{Eq:Pr:StabilizationUnderFSF}
    -MA^\T-AM-L^\T B^\T-BL>0
\end{equation}
and $K=LM^{-1}$.
\end{proposition}
\begin{proof}
The proof is complete by following the steps: (1) apply Prop. \ref{Pr:CTLTIStability} for \eqref{Eq:CTLTIUnderFSF} to get an LMI in $P$ and $K$, (2) transform the obtained LMI using Lm. \ref{Lm:PreAndPostMultiplication} with $P^{-1}$ (pre- and post-multiply by $P^{-1}$), and (3) change the LMI variables using $M \triangleq P^{-1}$ and $L \triangleq KP^{-1}$. Details are omitted for brevity.
\end{proof}

\subsubsection{\textbf{$(Q,S,R)$-Dissipativation}}
The following proposition provides an LMI condition that leads to synthesize a FSF controller $K$ such that the closed-loop system \eqref{Eq:CTLTIUnderFSF} is $(Q,S,R)$-dissipative from $w(t)$ to $y(t)$. 

\begin{proposition}\label{Pr:DissipativationUnderFSF}
Under $D=\0$, the closed-loop CT-LTI system \eqref{Eq:CTLTIUnderFSF} is $(Q,S,R)$-dissipative (with $-Q>0$, $R=R^\T$) from $w(t)$ to $y(t)$ iff $\exists M>0$ and $L$ such that 
\begin{equation}\label{Eq:Pr:DissipativationUnderFSF}
    \bm{
    -\H_s(AM+BL) & -E+MC^\T S & MC^\T \\ 
    \star & \H_s(F^\T S)+R & F^\T \\
    \star & \star & -Q^{-1} 
    }>0 
\end{equation}
and $K=LM^{-1}$.
\end{proposition}
\begin{proof}
The proof is complete by following the steps: 
(1) apply Prop. \ref{Pr:CTLTIQSRDissipativity} for \eqref{Eq:CTLTIUnderFSF} to get an LMI in $P$ and $K$, 
(2) transform the LMI using Lm. \ref{Lm:PreAndPostMultiplication} with $\diag(P^{-1},\I,\I)$, and (3) change the LMI variables using $M \triangleq P^{-1}$ and $L \triangleq KP^{-1}$.  
\end{proof}

\subsubsection{\textbf{$\H_2$-Optimal Control}}
Under FSF control $u(t)=Kx(t)$, the goal of $\H_2$-optimal control is to synthesize a controller $K$ that minimizes the $\H_2$-norm of the closed-loop system \eqref{Eq:CTLTIUnderFSF} (from $w(t)$ to $y(t)$). For this purpose, the following proposition provides an LMI based approach.

\begin{proposition}\label{Pr:H2ControlUnderFSF}
Under $F=\0$, the $\H_2$-optimal FSF controller $K$ for the closed-loop system \eqref{Eq:CTLTIUnderFSF} is found by solving the LMI:
\begin{equation}\label{Eq:Pr:H2ControlUnderFSF}
\begin{aligned}
    \min_{M,Q,L,\gamma}&\ \gamma\\
    \mbox{sub. to:}&\ M>0,\ Q>0,\ \gamma>0\\
    &\bm{-\H_s(AM+BL) & -MC^\T - L^\T D^\T \\ \star & \gamma \I }>0,\\
    &\bm{M & E \\ \star & Q}>0,\ \tr{Q}<\gamma, 
\end{aligned}
\end{equation}
and $K=LM^{-1}$.
\end{proposition}
\begin{proof}
The proof is complete by applying Prop. \ref{Pr:H2Norm} to \eqref{Eq:CTLTIUnderFSF} and executing the change of variables: $M=P$ and $L=KP$.
\end{proof}

\subsubsection{\textbf{$\H_\infty$-Optimal Control}}
Similarly, the goal of $\H_\infty$-optimal control is to synthesize a controller $K$ that minimizes the $\H_\infty$-norm of the closed-loop system \eqref{Eq:CTLTIUnderFSF} (from $w(t)$ to $y(t)$). 

\begin{proposition}\label{Pr:HInfControlUnderFSF}
The $\H_\infty$-optimal FSF controller $K$ for the closed-loop system \eqref{Eq:CTLTIUnderFSF} is found by solving the LMI:
\begin{equation}\label{Eq:Pr:HInfControlUnderFSF}
\begin{aligned}
    \min_{M,L,\gamma}&\ \gamma\\
    \mbox{sub. to:}&\ M>0, \gamma>0,\\
    &\bm{-\H_s(AM+BL) & -E & - MC^\T - L^\T D^\T \\ 
    \star & \gamma\I & -F^\T \\ 
    \star & \star & \gamma \I }>0,
\end{aligned}
\end{equation}
and $K=LM^{-1}$.
\end{proposition}
\begin{proof}
The proof is complete by following the steps: (1) apply Prop. \ref{Pr:HInfNorm} to \eqref{Eq:CTLTIUnderFSF}, (2) transform the main LMI using Lm. \ref{Lm:PreAndPostMultiplication} with $\diag(P^{-1},\I,\I)$, and (3) change the LMI variables using $M\triangleq P^{-1}$ and $L=KP^{-1}$. 
\end{proof}

\subsection{Observer Design for CT-LTI Systems}

For state feedback control $u(t)=Kx(t)$ (unlike output feedback control $u(t)=Ky(t)$), the controller requires the state information $x(t)$ of the CT-LTI system \eqref{Eq:NoisyCTLTISystem}. However, typically, the state $x(t)$ (unlike the output $y(t)$) is not available to the controller. Therefore, an \emph{observer} is required to keep an \emph{estimate} of the state $x(t)$ as $\hat{x}(t)$.

\subsubsection{\textbf{Luenberger Observer}}
For the CT-LTI system \eqref{Eq:NoisyCTLTISystem}, consider a \emph{Luenberger observer} implemented at the controller:
\begin{equation}\label{Eq:LuenbergerObserver}
    \dot{\hat{x}} = \hat{A}\hat{x}(t) + \hat{B}u(t) + Ly(t),
\end{equation}
that has the estimation error ($e(t)\triangleq x(t)-\hat{x}(t)$) dynamics:
\begin{equation}\label{Eq:LuenbergerObserverErrorDynamics}
\begin{aligned}
    \dot{e}(t) = &\hat{A}e(t) + (A-\hat{A}-LC)x(t) \
    \\ &+ (B-\hat{B}-LD)u(t) + (E-LF)w(t).
\end{aligned}
\end{equation}
The Luenberger observer parameters: $\hat{A}, \hat{B}$ and $L$ can be selected according to the following proposition. 

\begin{proposition}\label{Pr:Observer}
Under $w(t)=0$ and Luenberger observer \eqref{Eq:LuenbergerObserver} parameters: $\hat{A} \triangleq A-LC$ and $\hat{B} \triangleq B-LD$, the estimation error dynamics \eqref{Eq:LuenbergerObserverErrorDynamics} are stable iff $\exists P > 0$ and $K$ such that 
\begin{equation}\label{Eq:Pr:Observer}
    -A^\T P - PA + C^\T K^\T + KC > 0 
\end{equation}
and $L=P^{-1}K$.
\end{proposition}
\begin{proof}
Under $w(t)=0$ and the given observer parameter choices, the estimation error dynamics \eqref{Eq:LuenbergerObserverErrorDynamics} reduces to: 
\begin{equation}\label{Eq:Pr:ObserverStep1}
\dot{e}(t) = (A-LC)e(t).    
\end{equation}
The proof is complete by applying Prop. \ref{Pr:CTLTIStability} to \eqref{Eq:Pr:ObserverStep1} and then executing a change of variables using $K=PL$. 
\end{proof}

The Luenberger observer design proposed above assumes the noise-less case of \eqref{Eq:NoisyCTLTISystem} (i.e., \eqref{Eq:CTLTISystem}). This assumption is relaxed in the $(Q,S,R)$-dissipative and $\H_2/\H_\infty$-optimal observer designs described in subsequent subsections. However, before getting into those details, first, consider the Luenberger observer \eqref{Eq:LuenbergerObserver} parameters: 
\begin{equation}
    \label{Eq:LuenbergerObserverParameters}
    \hat{A} \triangleq A-LC \ \ \ \ \mbox{ and } \ \ \ \ 
    \hat{B} \triangleq B-LD,
\end{equation}
under which the estimation error dynamics \eqref{Eq:LuenbergerObserverErrorDynamics} take the form 
\begin{equation}\label{Eq:LuenbergerObserverWithPerformance}
    \begin{aligned}
        \dot{e}(t) =&\ (A-LC)e(t)+(E-LF)w(t),\\
        z(t) =&\ Ge(t)+Jw(t).
    \end{aligned}
\end{equation}
where $z(t)$ is a pre-defined performance metric.  

\subsubsection{\textbf{$(Q,S,R)$-Dissipative Observer}}
The $(Q,S,R)$-dissipative observer synthesizes the Luenberger observer gain $L$ such that \eqref{Eq:LuenbergerObserverWithPerformance} is $(Q,S,R)$-dissipative from $w(t)$ to $z(t)$. For this purpose, the following proposition can be used.

\begin{proposition}\label{Pr:DissipativeObserver}
The estimation error dynamics \eqref{Eq:LuenbergerObserverWithPerformance} is $(Q,S,R)$-dissipative (with $-Q>0$, $R=R^\T$) from $w(t)$ to $z(t)$ iff $\exists P>0$ and $K$ such that 
\begin{equation}\label{Eq:Pr:DissipativeObserver}
    \bm{
    -\H_s(PA-KC) & -PE+KF+G^\T S & G^\T \\ \star & \H_s(J^\T S)+R & J^\T \\ \star & \star & -Q^{-1}
    }>0
\end{equation}
and $L=KP^{-1}$.
\end{proposition}
\begin{proof}
The proof is complete by applying Prop. \ref{Pr:CTLTIQSRDissipativity} to \eqref{Eq:LuenbergerObserverWithPerformance} and executing a change of variables using $K=PL$. 
\end{proof}

\subsubsection{\textbf{$\H_2$-Optimal Observer}}
The goal of the $\H_2$-optimal observer is to synthesize the Luenberger observer gain $L$ that minimizes the $\H_2$-norm of \eqref{Eq:LuenbergerObserverWithPerformance} (from $e(t)$ to $z(t)$). 

\begin{proposition}\label{Pr:H2Observer}
Under $J=0$, the $\H_2$-optimal observer gain $L$ for \eqref{Eq:LuenbergerObserverWithPerformance} is found by solving the LMI:
\begin{equation}\label{Eq:Pr:H2Observer}
\begin{aligned}
    \min_{P,Q,K,\gamma}&\ \gamma\\
    \mbox{sub. to:}&\ P>0,\ Q>0,\ \gamma>0,\\
    &\bm{-\H_s(PA-KC) & -PE + KF \\ \star & \gamma \I }>0,\\
    &\bm{P & G^\T \\ \star & Q}>0,\ \tr{Q}<\gamma,
\end{aligned}
\end{equation}
and $L=P^{-1}K$.
\end{proposition}
\begin{proof}
The proof is complete by applying Prop. \ref{Pr:H2Norm} \eqref{Eq:Pr:H2Norm2} to \eqref{Eq:LuenbergerObserverWithPerformance} and executing a change of variables using $K=PL$. 
\end{proof}

\subsubsection{\textbf{$\H_\infty$-Optimal Observer}}
Similarly, the goal of the $\H_\infty$-optimal observer is to synthesize the Luenberger observer gain $L$ that minimizes the $\H_\infty$-norm of \eqref{Eq:LuenbergerObserverWithPerformance} (from $e(t)$ to $z(t)$).

\begin{proposition}\label{Pr:HInfObserver}
The $\H_\infty$-optimal observer gain $L$ for \eqref{Eq:LuenbergerObserverWithPerformance} is found by solving the LMI:
\begin{equation}\label{Eq:Pr:HInfObserver}
\begin{aligned}
    \min_{P,K,\gamma}&\ \gamma\\
    \mbox{sub. to:}&\ P>0,\ \gamma>0,\\
    &\bm{-\H_s(PA-KC) & -PE + KF & - G^\T \\ 
    \star & \gamma\I & -J^\T \\ 
    \star & \star & \gamma \I }>0,
\end{aligned}
\end{equation}
and $L=P^{-1}K$.
\end{proposition}
\begin{proof}
The proof is complete by applying Prop. \ref{Pr:HInfNorm} to \eqref{Eq:LuenbergerObserverWithPerformance} and executing a change of variables using $K=PL$. 
\end{proof}

\subsection{Dynamic Output Feedback (DOF) Controller Synthesis for CT-LTI Systems}\label{SubSec:DOFControllerSynthesis}

Consider the CT-LTI system \eqref{Eq:CTLTISystem} with noise $w(t)\in\R^{q}$ and performance $z(t)\in\R^{l}$:
\begin{equation}\label{Eq:NoisyCTLTISystemWithPerf}
    \begin{aligned}
        \dot{x}(t) = Ax(t) + Bu(t) + Ew(t),\\
        y(t) = Cx(t) + Du(t) + Fw(t),\\
        z(t) = Gx(t) + Hu(t) + Jw(t).
    \end{aligned}
\end{equation}
Under $D=\0$ and dynamic output feedback (DOF) control (from $y(t)$ to $u(t)$):
\begin{equation}\label{Eq:CTDOFController}
    \begin{aligned}
        \dot{\zeta}(t) = A_c \zeta(t) + B_c y(t),\\
        u(t) = C_c \zeta(t) + D_c y(t),
    \end{aligned}
\end{equation}
where $\zeta(t)\in\R^{r}$, the closed-loop CT-LTI system \eqref{Eq:NoisyCTLTISystemWithPerf} takes the form:
\begin{equation}\label{Eq:CTLTIUnderDOF}
    \begin{aligned}
        \dot{\theta}(t) = \bar{A}\theta(t) + \bar{B}w(t),\\
        z(t) = \bar{C}\theta(t) + \bar{D}w(t),
    \end{aligned}
\end{equation}
with $\theta(t)=\bm{x^\T(t) & \zeta^\T(t)}^\T$ and
\begin{equation*}
\begin{aligned}
    \bar{A} \triangleq \bm{A+BD_cC & BC_c\\B_cC & A_c},\ 
    \bar{B} \triangleq \bm{E+BD_cF\\B_cF},\\ 
    \bar{C} \triangleq \bm{G+HD_cC & HC_c},\ 
    \bar{D} \triangleq \bm{J+HD_cF}.
\end{aligned}
\end{equation*}

In parallel to Sec. \ref{SubSec:FSFControllerSynthesis}, in the subsequent subsections, we provide LMI conditions for DOF controller synthesis (i.e., to design $A_c,B_c,C_c,D_c$ in \eqref{Eq:CTDOFController}) so as to stabilize, $(Q,S,R)$-dissipativate or optimize $\H_2$/$\H_\infty$-norm of the closed-loop system \eqref{Eq:CTLTIUnderDOF}. Before getting into those details, first, consider the unique change of variables (CoVs) process given below.

\subsubsection{\textbf{Change of Variables (CoVs)}}\label{SubSubSec:ChangeOfVariables}

Inspired by \cite{Scherer1997,Caverly2019}, let 
\begin{equation}\label{Eq:CoVPAndPiDef}
\begin{aligned}
    &P \triangleq \bm{X & M\\ M^\T & \ast},\ 
    &P^{-1} \triangleq \bm{Y & N\\ N^\T & \ast},\\
    &\Pi_x \triangleq \bm{\I & X \\ \0 & M^\T},\  
    &\Pi_y \triangleq \bm{Y & \I \\ N^\T & \0},
\end{aligned}
\end{equation}
where $X,Y,M,N$ are some matrices with appropriate dimensions ($\ast$ represents irrelevant matrices). Based on this definition, it is easy to establish the following three properties: 
\begin{enumerate}
    \item Matrices $X,Y,M,N$ satisfy:
    \begin{equation}\label{Eq:CoVsXYToMN}
        XY+MN^\T =\I;
    \end{equation}
    \item Matrices $P,P^{-1},\Pi_x,\Pi_y$ satisfy: 
    \begin{equation}\label{Eq:CoVPixPiy}
        P \Pi_y = \Pi_x, \ \ \Pi_y = P^{-1}\Pi_x, \ \ \Pi_y^\T = \Pi_x^\T P^{-1};
    \end{equation}
    \item If $M$ is full-rank, 
    \begin{equation}\label{Eq:CoVPisPositiveDefinite}
        P>0 \iff \exists X,Y>0 \mbox{ such that }\bm{Y & \I \\ \I & X}>0.
\end{equation}
\end{enumerate}
We point out that the last property given above can be proven by subsequently using Lm. \ref{Lm:PreAndPostMultiplication} with $\Pi_x^\T P^{-1}$, \eqref{Eq:CoVPixPiy} and \eqref{Eq:CoVPAndPiDef} as:  
$$P>0 \iff \Pi_x^\T P^{-1} \Pi_x= \Pi_y^\T\Pi_x = \bm{Y & \I \\ \I & X}>0.$$

Note also that, given any two matrices $X,Y$ such that $\I-XY$ is non-singular, two unique full-rank matrices $M,N$ can always be found such that \eqref{Eq:CoVsXYToMN} using LU-decomposition \cite{Bernstein2009}. 

Further, under such given $X,Y,M,N$ matrices, a set of matrices $\{A_n,B_n,C_n,D_n\}$ can be uniquely transformed respectively to and from another set of matrices $\{A_c,B_c,C_c,D_c\}$ using   
\begin{equation}\label{Eq:COVsAuxToDOF}
    \begin{aligned}
        D_c \triangleq&\ D_n,\\
        C_c \triangleq&\ (C_n-D_nCY)N^{-\T},\\
        B_c \triangleq&\ M^{-1}(B_n-XBD_n),\\
        A_c \triangleq&\ M^{-1}(A_n-B_nCY-XBC_n-X(A-BD_nC)Y)N^{-\T}
    \end{aligned}
\end{equation}
and  
\begin{equation}\label{Eq:COVsDOFToAux}
    \begin{aligned}
        D_n \triangleq&\ D_c,\\
        C_n \triangleq&\ D_cCY + C_cN^\T,\\
        B_n \triangleq&\ XBD_c + MB_c,\\
        A_n \triangleq&\ MA_cN^\T + MB_cCY + XBC_cN^\T + X(A+BD_cC)Y.
    \end{aligned}
\end{equation}

As we will see in the sequel, we formulate each DOF controller synthesis problem as an LMI problem in matrix parameters $X,Y,A_n,B_n,C_n,D_n$ using: (1) intermediate variables in \eqref{Eq:CoVPAndPiDef}, (2) Lm. \ref{Lm:PreAndPostMultiplication}, and (3) relationships in \eqref{Eq:CTLTIUnderDOF},\eqref{Eq:CoVPixPiy}, \eqref{Eq:CoVPisPositiveDefinite} and \eqref{Eq:COVsDOFToAux}. More explicitly, the following relationships (that can be proven by direct substitution) will be pivotal in this task:
\begin{equation}\label{Eq:COVsKeyResults}
\begin{aligned}
    \Pi_x^\T \bar{A} \Pi_y =&\ \bm{AY+BC_n & A+BD_nC \\ A_n & XA+B_nC},\\
    \Pi_x^\T \bar{B} =&\ \bm{E+BD_nF \\ XE+B_nF}, \\
    \bar{C}\Pi_y =&\ \bm{ GY+HC_n & G+HD_nC}.
\end{aligned}
\end{equation}
Upon solving each such formulated LMI problem (in $X,Y,A_n,B_n,C_n,D_n$), the relationships in \eqref{Eq:CoVsXYToMN} and \eqref{Eq:COVsAuxToDOF} can be used to obtain the matrices $M,N,A_c,B_c,C_c,D_c$ (i.e., the DOF controller \eqref{Eq:CTDOFController}).

\subsubsection{\textbf{Stabilization}}
The following proposition gives an LMI condition that leads to synthesizing a DOF controller \eqref{Eq:CTDOFController} such that the closed-loop system  \eqref{Eq:CTLTIUnderDOF} is stabilized.

\begin{proposition}\label{Pr:StabilizationUnderDOF}
Under $D=w(t)=\0$, the closed-loop CT-LTI system \eqref{Eq:CTLTIUnderDOF} is stable iff $\exists X,Y>0$ and $A_n,B_n,C_n,D_n$ such that 
\begin{align}
\label{Eq:Pr:StabilizationUnderDOF1}
\bm{Y & \I \\ \I & X}>0,\\ 
\label{Eq:Pr:StabilizationUnderDOF2}
\bm{-\H_s(AY+BC_n) & -A - BD_nC - A_n^\T \\ 
\star & -\H_s(XA + B_nC)}>0,
\end{align}
and $A_c,B_c,C_c,D_c$ \eqref{Eq:CTDOFController} are found by CoVs \eqref{Eq:CoVsXYToMN} and \eqref{Eq:COVsAuxToDOF}. 
\end{proposition}

\begin{proof}
Applying Prop. \ref{Pr:CTLTIStability} to \eqref{Eq:CTLTIUnderDOF} give the LMI conditions necessary and sufficient for the stabilization of \eqref{Eq:CTLTIUnderDOF} as: $\exists P>0$ such that 
\begin{equation}\label{Eq:Pr:StabilizationUnderDOFStep1}
    -\bar{A}^\T P -P\bar{A} >0.  
\end{equation}
Using the CoVs in \eqref{Eq:CoVPAndPiDef} and \eqref{Eq:CoVPisPositiveDefinite}, the LMI $P>0$ can be transformed to \eqref{Eq:Pr:StabilizationUnderDOF1}. Finally, \eqref{Eq:Pr:StabilizationUnderDOFStep1} can be transformed to \eqref{Eq:Pr:StabilizationUnderDOF2} by applying Lm. \ref{Lm:PreAndPostMultiplication} with $\Pi_x^\T P^{-1}$ and substituting from \eqref{Eq:CoVPixPiy}:
\begin{align}
    -\bar{A}^\T P -P\bar{A} > 0 
    \iff -\Pi_x^\T P^{-1}\bar{A}^\T \Pi_x -  \Pi_x^\T\bar{A}P^{-1}\Pi_x >0, \nonumber \\
    \iff -\Pi_y^\T \bar{A}^\T \Pi_x - \Pi_x^\T \bar{A} \Pi_y > 0 \iff \mbox{\eqref{Eq:Pr:StabilizationUnderDOF2}}.\nonumber
\end{align} 
Note that the last step above results from \eqref{Eq:COVsKeyResults}.
\end{proof}

\begin{figure*}[!b]
    \centering
    \hrulefill
    \begin{equation}\label{Eq:Pr:DissipativationUsingDOF2}
    \bm{
    -\H_s(AY+BC_n) & -A-BD_nC-A_n^\T & -E-BD_nF+(YG^\T+C_n^\T H^\T)S & YG^\T+C_n^\T H^\T\\
    \star & -\H_s(XA+B_nC) & -XE-B_nF+(G^\T+C^\T D_n^\T H^\T)S & G^\T+C^\T D_n^\T H^\T\\
    \star & \star & \H_s((J^\T+F^\T D_n^\T H^\T)S)+R & J^\T+F^\T D_n^\T H^\T\\
    \star & \star & \star & -Q^{-1}
    }>0
    \end{equation}
    \begin{equation}\label{Eq:Pr:H2ControlUnderDOF2}
    \bm{
    -\H_s(AY+BC_n) & -A-BD_nC-A_n^\T & -E-BD_nF \\
    \star & -\H_s(XA+B_nC) & -XE-B_nF\\
    \star & \star & \I
    }>0 \ \ \mbox{ and } \ \ 
    \bm{
    Y & \I & YG^\T+C_n^\T H^\T\\ 
    \star & X & G^\T+C^\T D_n^\T H^\T\\
    \star & \star & Q
    }>0
    \end{equation}
    \begin{equation}\label{Eq:Pr:HInfControlUnderDOF2}
        \bm{
        -\H_s(AY+BC_n) & -A-BD_nC-A_n^\T & -E-BD_nF & -YG^\T-C_n^\T H^\T\\
        \star & -\H_s(XA+B_nC) & -XE-B_nF & -G^\T-C^\T D_n^\T H^\T\\
        \star & \star & \gamma \I & -J^\T-F^\T D_n^\T H^\T\\
        \star & \star & \star & \gamma I
        }>0
    \end{equation}
\end{figure*}

\subsubsection{\textbf{$(Q,S,R)$-Dissipativation}}
The following proposition provides an LMI condition that leads to synthesize a DOF controller \eqref{Eq:CTDOFController} such that the closed-loop system \eqref{Eq:CTLTIUnderDOF} is $(Q,S,R)$-dissipative from $w(t)$ to $z(t)$. 

\begin{proposition}\label{Pr:DissipativationUsingDOF}
Under $D=\0$, the closed-loop CT-LTI system \eqref{Eq:CTLTIUnderDOF} is $(Q,S,R)$-dissipative (with $-Q>0$, $R=R^\T$) from $w(t)$ to $z(t)$ iff $\exists X,Y>0$ and $A_n,B_n,C_n,D_n$ such that 
\begin{equation}\label{Eq:Pr:DissipativationUsingDOF1}
    \bm{Y & \I \\ \I & X}>0 \mbox{ and  \eqref{Eq:Pr:DissipativationUsingDOF2}},
\end{equation}
and $A_c,B_c,C_c,D_c$ \eqref{Eq:CTDOFController} are found by CoVs \eqref{Eq:CoVsXYToMN} and \eqref{Eq:COVsAuxToDOF}. 
\end{proposition}
\begin{proof}
The proof starts with applying Prop. \ref{Pr:CTLTIQSRDissipativity} to \eqref{Eq:CTLTIUnderDOF} to obtain the  LMI conditions necessary and sufficient for the $(Q,S,R)$-dissipativation of \eqref{Eq:CTLTIUnderDOF} as: 
$\exists P>0$ such that 
\begin{equation}\label{Eq:Pr:DissipativationUsingDOFStep1}
    \bm{
    -\bar{A}^\T P-P\bar{A} & -P\bar{B}+\bar{C}^\T S & \bar{C}^\T \\
    \star & \bar{D}^\T S+S^\T\bar{D}+R & \bar{D}^\T \\
    \star & \star & -Q^{-1}
    }>0.
\end{equation}
Similar to the proof of Prop. \ref{Pr:StabilizationUnderDOF}, using the CoVs in \eqref{Eq:CoVPAndPiDef} and \eqref{Eq:CoVPisPositiveDefinite}, the LMI $P>0$ can be transformed to \eqref{Eq:Pr:DissipativationUsingDOF1}. To obtain \eqref{Eq:Pr:DissipativationUsingDOF2} from \eqref{Eq:Pr:DissipativationUsingDOFStep1}, first, Lm. \ref{Lm:PreAndPostMultiplication} is applied with $\diag(\Pi_x^\T P^{-1},\I,\I\}$ and then the results are substituted using \eqref{Eq:CoVPixPiy} to obtain:
\begin{equation}
    \bm{
    -\H_s(\Pi_x^\T\bar{A}\Pi_y) & -\Pi_x^\T\bar{B}+\Pi_y^\T\bar{C}^\T S & \Pi_y^\T \bar{C}^\T\\
    \star & \H_s(\bar{D}^\T S)+R & \bar{D}^\T \\
    \star & \star & -Q^{-1}
    }>0.
\end{equation}
Finally, the above matrix inequality can be transformed to get the LMI in \eqref{Eq:Pr:DissipativationUsingDOF2} using \eqref{Eq:COVsKeyResults}.
\end{proof}

\subsubsection{\textbf{$\H_2$-Optimal Control}}
The goal of $\H_2$-optimal control here is to synthesize a DOF controller \eqref{Eq:CTDOFController} that minimizes the $\H_2$-norm of the closed-loop system \eqref{Eq:CTLTIUnderDOF} from $w(t)$ to $z(t)$. 

\begin{proposition}\label{Pr:H2ControlUnderDOF}
Under $D=\0$, the $\H_2$-optimal DOF controller \eqref{Eq:CTDOFController} for the closed-loop system \eqref{Eq:CTLTIUnderDOF} is found by solving the LMI:
\begin{equation}\label{Eq:Pr:H2ControlUnderDOF1}
\begin{aligned}
    \min_{\substack{X,Y,Q,\gamma\\A_n,B_n,C_n,D_n}}\ \ &\gamma\\
    \mbox{sub. to:}\ \ &X>0,\ Y>0,\ Q>0,\ \gamma>0,\ \mbox{\eqref{Eq:Pr:H2ControlUnderDOF2}},\\
    &J+HD_nF = 0,\ \tr{Q}<\gamma,
\end{aligned}
\end{equation}
and $A_c,B_c,C_c,D_c$ \eqref{Eq:CTDOFController} are found by CoVs \eqref{Eq:CoVsXYToMN} and \eqref{Eq:COVsAuxToDOF}.
\end{proposition}
\begin{proof}
The proof follows similar steps as that of Prop. \ref{Pr:DissipativationUsingDOF}.  
\end{proof}

\subsubsection{\textbf{$\H_\infty$-Optimal Control}}
Similarly, the goal of $\H_\infty$-optimal control is to synthesize a DOF controller \eqref{Eq:CTDOFController} that minimizes the $\H_\infty$-norm of the closed-loop system \eqref{Eq:CTLTIUnderDOF} (from $w(t)$ to $z(t)$). 

\begin{proposition}\label{Pr:HInfControlUnderDOF}
Under $D=0$, the $\H_\infty$-optimal DOF controller \eqref{Eq:CTDOFController} for the closed-loop system \eqref{Eq:CTLTIUnderDOF} is found by solving the LMI:
\begin{equation}\label{Eq:Pr:HInfControlUnderDOF1}
\begin{aligned}
    \min_{\substack{X,Y,\gamma\\A_n,B_n,C_n,D_n}} &\ \gamma\\
    \mbox{sub. to:}&\ X>0, Y>0, \gamma>0, \bm{Y & \I\\ \I & X}>0, \eqref{Eq:Pr:HInfControlUnderDOF2}.
\end{aligned}
\end{equation}
and $A_c,B_c,C_c,D_c$ \eqref{Eq:CTDOFController} are found by CoVs \eqref{Eq:CoVsXYToMN} and \eqref{Eq:COVsAuxToDOF}.
\end{proposition}
\begin{proof}
The proof follows similar steps as that of Prop. \ref{Pr:DissipativationUsingDOF}.  
\end{proof}

We conclude this section by summarizing all the established theoretical results in Tab. \ref{Tab:LTISystemsLMIResults}. 

\begin{table}[!ht]
\caption{Summary of Theoretical Results (CT-LTI Systems).}
\label{Tab:LTISystemsLMIResults}
\resizebox{\columnwidth}{!}{
\begin{tabular}{|ll|llll|}
\hline
\multicolumn{2}{|c|}{\multirow{3}{*}{\begin{tabular}[c]{@{}c@{}}Proposition \#\\ (For CT-LTI Systems) \end{tabular} }} &
  \multicolumn{4}{c|}{Concept} \\ \cline{3-6} 
\multicolumn{2}{|c|}{} &
  \multicolumn{1}{c|}{Stability} &
  \multicolumn{1}{c|}{\begin{tabular}[c]{@{}c@{}}$(Q,S,R)$-\\ Dissipativity\end{tabular}} &
  \multicolumn{1}{c|}{\begin{tabular}[c]{@{}c@{}}$\H_2$-\\ Norm\end{tabular}} &
  \multicolumn{1}{c|}{\begin{tabular}[c]{@{}c@{}}$\H_\infty$-\\ Norm\end{tabular}} \\ \hline
\multicolumn{1}{|l|}{\multirow{8}{*}{Task}} &
  \begin{tabular}[c]{@{}l@{}}CT-LTI System\\ Analysis\end{tabular} &
  \multicolumn{1}{c|}{\ref{Pr:CTLTIStability}} & 
  \multicolumn{1}{c|}{\ref{Pr:CTLTIQSRDissipativity}} &
  \multicolumn{1}{c|}{\ref{Pr:H2Norm}} & 
  \multicolumn{1}{c|}{\ref{Pr:HInfNorm}} 
   \\ \cline{2-6} 
\multicolumn{1}{|l|}{} &
  \begin{tabular}[c]{@{}l@{}}FSF Controller\\ Synthesis\end{tabular} &
  \multicolumn{1}{c|}{\ref{Pr:StabilizationUnderFSF},\ref{Pr:Stabilizability}} &
  \multicolumn{1}{c|}{\ref{Pr:DissipativationUnderFSF}} &
  \multicolumn{1}{c|}{\ref{Pr:H2ControlUnderFSF}} & 
  \multicolumn{1}{c|}{\ref{Pr:HInfControlUnderFSF}}  
   \\ \cline{2-6} 
\multicolumn{1}{|l|}{} &
  \begin{tabular}[c]{@{}l@{}}Observer\\ Design\end{tabular} &
  \multicolumn{1}{c|}{\ref{Pr:Observer},\ref{Pr:Detectability}} &
  \multicolumn{1}{c|}{\ref{Pr:DissipativeObserver}} &
  \multicolumn{1}{c|}{\ref{Pr:H2Observer}} &
  \multicolumn{1}{c|}{\ref{Pr:HInfObserver}} 
   \\ \cline{2-6} 
\multicolumn{1}{|l|}{} &
  \begin{tabular}[c]{@{}l@{}}DOF Controller\\ Synthesis\end{tabular} &
  \multicolumn{1}{c|}{\ref{Pr:StabilizationUnderDOF}} &
  \multicolumn{1}{c|}{\ref{Pr:DissipativationUsingDOF}} &
  \multicolumn{1}{c|}{\ref{Pr:H2ControlUnderDOF}} & 
  \multicolumn{1}{c|}{\ref{Pr:HInfControlUnderDOF}} 
   \\ \hline
\end{tabular}}
\end{table}

\section{The Continuous-Time Networked System (CTNS)}
\label{Sec:CTNetworkedSystem}

In this section, we provide the details of the considered continuous-time networked system (CTNS) and outline the interested research problem. 

\subsection{Subsystems of the CTNS}
\subsubsection{\textbf{Dynamics}} 
We consider a CTNS $\mathcal{G}_N$ comprised of $N$ interconnected subsystems $\{\Sigma_i:i\in\N_N\}$ (e.g., see Fig. \ref{Fig:ExampleNetworkedSystem}). The dynamics of the $i$\tsup{th} subsystem $\Sigma_i, i\in\N_N$ are given by
\begin{equation}\label{Eq:CTSubsystemDynamics}
\begin{aligned}
    \dot{x}_i(t) =& \sum_{j\in\bar{\E}_i}A_{ij}x_j(t) + \sum_{j\in\bar{\E}_i}B_{ij}u_j(t)+\sum_{j\in\bar{\E}_i}E_{ij}w_{j}(t),\\ 
    y_i(t) =& \sum_{j\in\bar{\E}_i}C_{ij}x_j(t) + \sum_{j\in\bar{\E}_i}D_{ij}u_j(t) + \sum_{j\in\bar{\E}_i}F_{ij}w_j(t),
\end{aligned}
\end{equation}
where $x_i(t) \in \R^{n_i},\ u_i(t)\in\R^{p_i},\ w_i(t)\in\R^{q_i}$ and $y_i(t)\in\R^{m_i}$ respectively represent the state, input, disturbance and output specific to the subsystems $\Sigma_i$ at time $t\in\R_{\geq 0}$. 

\begin{figure}[!tb]
    \centering
    \includegraphics[width=3in]{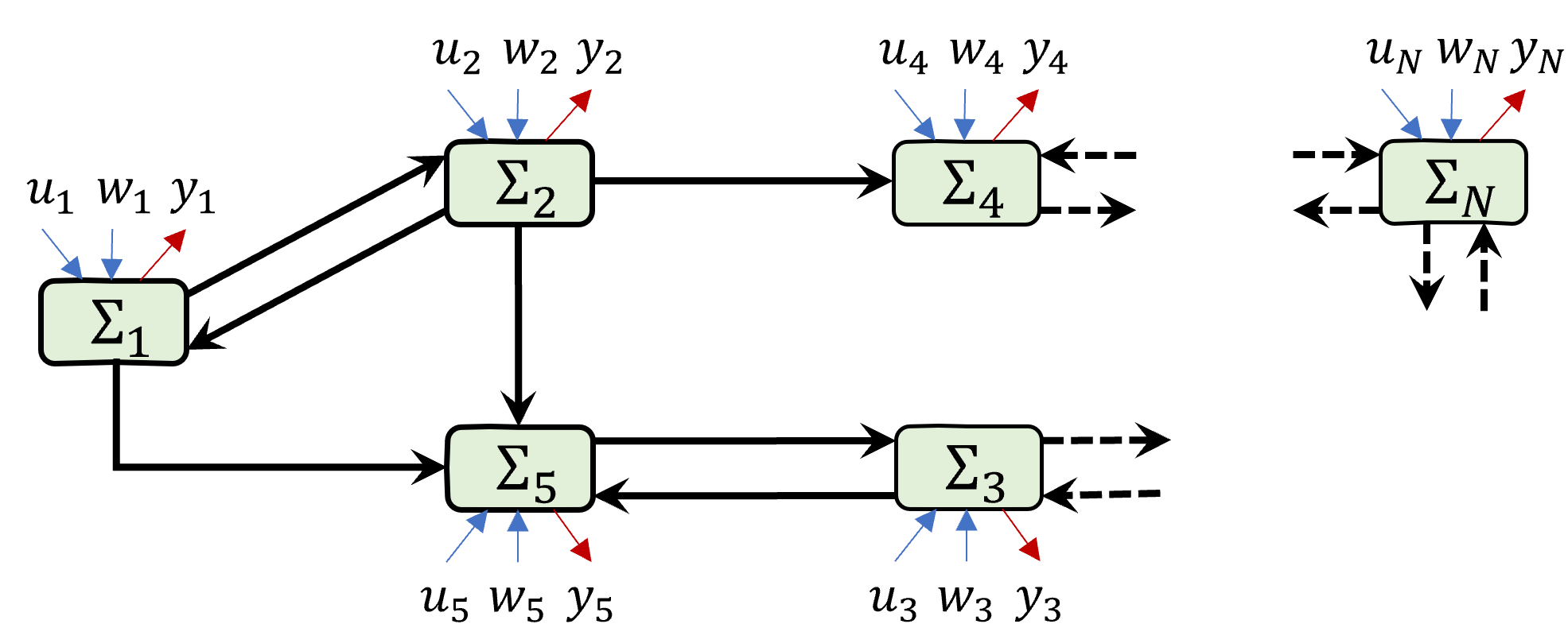}
    \caption{An example networked dynamical system $\mathcal{G}_N$.}
    \label{Fig:ExampleNetworkedSystem}
\end{figure}

\subsubsection{\textbf{Neighbors}} 
In \eqref{Eq:CTSubsystemDynamics}, we denote $\bar{\E}_i \triangleq \E_i\cup\{i\}$ where $\E_i \subset \N_N$ is the set of ``in-neighbors'' of the subsystem $\Sigma_i$. Formally, any subsystem $\Sigma_j$ is an ``in-neighbor'' of subsystem $\Sigma_i$ (i.e., $j\in\E_i$) iff $j \neq i$ and matrices $A_{ij},B_{ij},C_{ij},D_{ij},E_{ij},F_{ij}$ in \eqref{Eq:CTSubsystemDynamics} are not \textbf{all} zero matrices.

On the other hand, we also define a set of ``out-neighbors'' for the subsystem $\Sigma_i$ as $\F_i \triangleq \{j: j\in\N_N, \E_j \ni i \}$ with $\bar{\F}_i \triangleq \F_i \cup \{i\}$. Formally, any subsystem $\Sigma_j$ is an ``out-neighbor'' of subsystem $\Sigma_i$ (i.e., $j\in\F_i$) iff $j \neq i$ and matrices $A_{ji},B_{ji},C_{ji},D_{ji},E_{ji},F_{ji}$ (in \eqref{Eq:CTSubsystemDynamics} written for subsystem $\Sigma_j$) are not \textbf{all} zero matrices. Finally, for notational convenience, let us denote $\C_i \triangleq \E_i \cup \F_i$ with $\bar{\C}_i \triangleq \C \cup \{i\}$.

It is worth noting that in this paper (similar to our preliminary work \cite{WelikalaP32022}) we do not constrain ourselves to cascaded network topologies where $\E_i \equiv \F_i = \{i-1,i+1\}\cap\N_N, \forall i\in \N_N$ as in \cite{Agarwal2019} or to bi-directional network topologies where $\E_i \equiv \F_i, \forall  i\in\N_N$ as in \cite{Agarwal2021}. Moreover, as can be seen in \eqref{Eq:CTSubsystemDynamics}, here we consider a general form of a networked dynamical system as opposed to \cite{Agarwal2019} or \cite{Agarwal2021} where coupling between any two subsystems $\Sigma_i$ and $\Sigma_j$ was possible only through the state variables $x_i, x_j$ (i.e., via matrices $A_{ij}, A_{ji}$ in \eqref{Eq:CTSubsystemDynamics}).

\subsubsection{\textbf{Local Controllers and Observers}}
One of the main objectives of this paper is to design distributed (i.e., local) controllers (e.g., FSF) and observers (e.g., Luenberger) at subsystems of the considered CTNS in a decentralized manner. 

In particular, at a subsystem $\Sigma_i$: 
\begin{enumerate}
    \item a local FSF controller may take the form
        \begin{equation}\label{Eq:CTLocalFSFController}
            u_i(t) = \sum_{j\in \bar{\E}_i} K_{ij} x_j(t),
        \end{equation}
        where $\{K_{ij}:j\in\bar{\E}_i\}$ is the set of local FSF controller parameters that subsystem $\Sigma_i$ has to design;
    \item a local Luenberger observer may take the form 
        \begin{equation}\label{Eq:CTLocalObserver}
            \dot{\hat{x}}_i(t) = \sum_{j\in \bar{\E}_i} \hat{A}_{ij} \hat{x}_j(t) + \sum_{j\in \bar{\E}_i} \hat{B}_{ij} u_j + \sum_{j\in \bar{\E}_i} L_{ij}y_j(t),
        \end{equation}
        where $\{(\hat{A}_{ij},\hat{B}_{ij},L_{ij}):j\in\bar{\E}_i\}$ is the set of local Luenberger observer parameters that subsystem $\Sigma_i$ has to design;
    \item a local DOF controller may take the form
    \begin{equation}\label{Eq:CTLocalDOFController}
        \begin{aligned}
            \dot{\zeta}_i(t) = \sum_{j\in\bar{\E}_i} A_{c,ij} \zeta_j(t) + \sum_{j\in\bar{\E}_i} B_{c,ij} y_j(t),\\
            u_i(t) = \sum_{j\in\bar{\E}_i} C_{c,ij} \zeta_j(t) + \sum_{j\in\bar{\E}_i} D_{c,ij} y_j(t),
        \end{aligned}    
    \end{equation}
    where $\{(A_{c,ij},B_{c,ij},C_{c,ij},D_{c,ij}):j\in\bar{\E}_i\}$ is the set of local DOF controller parameters that subsystem $\Sigma_i$ has to design.
\end{enumerate}

\subsubsection{\textbf{Local Performance Metrics}} 
As we saw in Sec. \ref{Sec:BasicsOfCTLTISystems}, a pre-defined performance metric (often denoted as $z(t)$, e.g., see \eqref{Eq:LuenbergerObserverWithPerformance},\eqref{Eq:NoisyCTLTISystemWithPerf}) is required when designing controllers or observers in a $(Q,S,R)$-dissipative or $\H_2/\H_\infty$-optimal sense. Therefore, a pre-defined local performance metric is required at each subsystem $\Sigma_i$ when designing local controllers and observers in a $(Q,S,R)$-dissipative or $\H_2/\H_\infty$-optimal sense.   

In particular, at a subsystem $\Sigma_i$, we use the following pre-defined local performance metrics:
\begin{enumerate}
    \item For local FSF controller design: 
    \begin{equation}\label{Eq:CTLocalFSFControllerPerf}
        y_i(t) = \sum_{j\in\bar{\E}_i}C_{ij}x_j(t) + \sum_{j\in\bar{\E}_i}D_{ij}u_j(t) + \sum_{j\in\bar{\E}_i}F_{ij}w_j(t);
    \end{equation} 
    \item For local Luenberger observer design: 
    \begin{equation}\label{Eq:CTLocalObserverPerf}
        z_i(t) = \sum_{j\in\bar{\E}_i} G_{ij}(x_j(t)-\hat{x}_j(t)) + \sum_{j\in\bar{\E}_i}J_{ij}w_j(t);
    \end{equation}
    \item For local DOF controller design:
    \begin{equation}\label{Eq:CTLocalDOFControllerPerf}
        z_i(t) = \sum_{j\in\bar{\E}_i} G_{ij}x_j(t) + \sum_{j\in\bar{\E}_i} H_{ij}u_j(t) + \sum_{j\in\bar{\E}_i}J_{ij}w_j(t).
    \end{equation}
\end{enumerate}

\subsection{The CTNS}

\subsubsection{\textbf{Dynamics}}
By writing \eqref{Eq:CTSubsystemDynamics} for all $i\in\N_N$, the dynamics of the networked system $\mathcal{G}_N$ can be obtained as 
\begin{equation}\label{Eq:CTNSDynamics}
\begin{aligned}
\dot{x}(t) =&\ Ax(t) + Bu(t) + Ew(t),\\
y(t) =&\ Cx(t) + Du(t) + Fw(t),
\end{aligned}
\end{equation}
where $A=[A_{ij}]_{i,j\in\N_N}$, $B=[B_{ij}]_{i,j\in\N_N}$, $E=[E_{ij}]_{i,j\in\N_N}$, $C=[C_{ij}]_{i,j\in\N_N}$, $D=[D_{ij}]_{i,j\in\N_N}$ and $F=[F_{ij}]_{i,j\in\N_N}$ are all $N \times N$ block matrices, and $x(t)=[x_i^\T(t)]_{i\in\N_N}^\T \in \R^n$, $u(t)=[u_i^\T(t)]_{i\in\N_N}^\T \in\R^p$, $w(t)=[w_i^\T(t)]_{i\in\N_N}^\T \in\R^q$ and $y(t)=[y_i^\T(t)]_{i\in\N_N}^\T \in\R^m$ (with $n = \sum_{i\in\N_N} n_i$, $p = \sum_{i\in\N_N}p_i$, $q = \sum_{i\in\N_N} q_i$ and $m = \sum_{i\in\N_N} m_i$) are all $N\times 1$ block matrices respectively representing the networked system's state, input, disturbance and output at time $t\in\R_{\geq 0}$.  

\subsubsection{\textbf{Network Topology}} 
Note that the block structure of the matrices $A, B, C, D, F, F$ in \eqref{Eq:CTNSDynamics} determines the topology of the CTNS (i.e., how different subsystems are coupled) and vice versa. In the following Sec.  \ref{Sec:DecentralizedAnalysis}, we will revisit this topic and extensively study the properties of such ``network'' matrices. 

\subsubsection{\textbf{Controllers and Observers}}
By composing each local controller/observer forms in \eqref{Eq:CTLocalFSFController},\eqref{Eq:CTLocalObserver}, \eqref{Eq:CTLocalDOFController} for all $i\in \N_N$, we can respectively obtain the network level (i.e., global): 
\begin{enumerate}
    \item FSF controller 
        \begin{equation}\label{Eq:CTGlobalFSFController}
            u(t) = Kx(t);
        \end{equation}
    \item Luenberger observer 
        \begin{equation}\label{Eq:CTGlobalObserver}
            \dot{\hat{x}}(t) = \hat{A}x(t) + \hat{B}u(t) + Ly(t);  
        \end{equation}
    \item DOF controller 
        \begin{equation}\label{Eq:CTGlobalDOFController}
        \begin{aligned}
            \dot{\zeta}(t) = A_c \zeta(t) + B_c y(y),\\
            u(t) = C_c \zeta(t) + D_c y(y),
        \end{aligned}
        \end{equation}
\end{enumerate}
where matrices $K,\hat{A},\hat{B},L,A_c,B_c,C_c,D_c$ are all $N\times N$ block matrices comprised of the corresponding local design parameters (e.g., $K=[K_{ij}]_{i,j\in\N_N}$).

\subsubsection{\textbf{Performance Metrics}}
Similarly, by composing each pre-defined local controller/observer performance metric forms in \eqref{Eq:CTLocalFSFControllerPerf},\eqref{Eq:CTLocalObserverPerf},\eqref{Eq:CTLocalDOFControllerPerf} for all $i\in\N_N$, we can respectively obtain the global performance metrics considered for:
\begin{enumerate}
    \item FSF controller design as: 
    \begin{equation}\label{Eq:CTGlobalFSFControllerPerf}
        y(t) = Cx(t) + Dy(t) + Fw(t);
    \end{equation}
    \item Luenberger observer design as: 
    \begin{equation}\label{Eq:CTGlobalObserverPerf}
        z(t) = G(x(t)-\hat{x}(t))+Jw(t);
    \end{equation}
    \item DOF controller design as:
    \begin{equation}\label{Eq:CTGlobalDOFControllerPerf}
        z(t) = Gx(t)+Hu(t)+Jw(t).
    \end{equation}
\end{enumerate}
Here also matrices $C,D,F,G,H,J$ are all $N\times N$ block matrices comprised of the corresponding pre-defined local performance metric parameters (e.g., $C=[C_{ij}]_{i,j\in\N_N}$).

\subsection{The Research Problem}
Note that the forms of the CTNS \eqref{Eq:CTNSDynamics}, global controllers/observers \eqref{Eq:CTGlobalFSFController}-\eqref{Eq:CTGlobalDOFController} and global performance metrics \eqref{Eq:CTGlobalFSFControllerPerf}-\eqref{Eq:CTGlobalDOFControllerPerf} are respectively identical to the general CT-LTI system (e.g., \eqref{Eq:NoisyCTLTISystem}), controllers/observers (e.g., \eqref{Eq:LuenbergerObserver},\eqref{Eq:CTDOFController}) and performance metrics (e.g., \eqref{Eq:NoisyCTLTISystemWithPerf},\eqref{Eq:LuenbergerObserverWithPerformance}) considered in Sec. \ref{Sec:BasicsOfCTLTISystems}. 
Therefore, all the LMI-based control solutions (Prop. \ref{Pr:CTLTIStability}-\ref{Pr:HInfControlUnderDOF}) discussed in Sec. \ref{Sec:BasicsOfCTLTISystems} are directly applicable for the CTNS \eqref{Eq:CTNSDynamics}. 

For example, based on Prop. \ref{Pr:CTLTIStability}, the CTNS \eqref{Eq:CTNSDynamics} is globally exponentially stable if there exists a matrix $P=P^\T>0$ such that $-A^\T P - P A > 0$, where now $A = [A_{ij}]_{i,j\in\N_N}$ is a block matrix with a particular structure determined by the network topology (e.g., $A_{ij}=\0$ if $j\not\in \E_i$).   

Intuitively, verifying/enforcing such LMI conditions requires the knowledge of the entire networked system and thus calls for a centralized entity. Moreover, the complete verification/enforcement process may have to be repeated whenever new subsystems are introduced into the networked system. To address these challenges, we make the objective of this paper to design a systematic decentralized and compositional approach to verify/enforce different LMI conditions of interest (corresponding to different properties of interest, e.g., see Prop. \ref{Pr:CTLTIStability}-\ref{Pr:HInfControlUnderDOF})  regarding the networked system.

Towards this goal, a critical feature that we will exploit is that LMI conditions of interest now involve matrices of a particular structure determined by the network topology (as also pointed out earlier). The next section of the paper focuses on such network-related matrices (which we will define as ``network matrices'') and derives a decentralized and compositional test criterion to evaluate their positive definiteness. As we will see in the subsequent section (Sec. \ref{Sec:DecentralizedAnalysis}), such a test criterion can effortlessly be adopted to verify/enforce LMI conditions for networked systems in a decentralized and compositional manner.

\section{Decentralized Analysis of Networked Systems}
\label{Sec:DecentralizedAnalysis}

As mentioned above, this section is dedicated to establishing several theoretical and algorithmic results regarding evaluating the positive definiteness of a particular class of matrices.

\subsection{Preliminary Concepts}

\subsubsection{\textbf{Network Matrices}} 
We start by defining a class of matrices we named ``network matrices'' \cite{WelikalaP32022}, that corresponds to a given networked system topology (e.g., $\mathcal{G}_N$ in Fig. \ref{Fig:ExampleNetworkedSystem}).

\begin{definition}\label{Def:NetworkMatrices}
Given a networked dynamical system $\mathcal{G}_n, n\in\N$, any $n\times n$ block matrix $\Theta = [\Theta_{ij}]_{i,j\in\N_n}$ is a \emph{network matrix} if: 
(1) any information specific to the subsystem $i$ is embedded only in its $i$\tsup{th} block row or block column, and 
(2) $j\not\in \C_i (\triangleq \E_i \cup \F_i) \implies \Theta_{ij}=\Theta_{ji}=\0$ for all $i,j\in\N_n$.
\end{definition}

Based on this definition, note that all the $N \times N$ block matrices: (1) $A,B,C,D,E,F$ in \eqref{Eq:CTNSDynamics}, $K,\hat{A},\hat{B},L,A_c,B_c,C_c,D_c$ in \eqref{Eq:CTGlobalFSFController}-\eqref{Eq:CTGlobalDOFController} and (3) $C,D,F,G,H,J$ in \eqref{Eq:CTGlobalFSFControllerPerf}-\eqref{Eq:CTLocalDOFControllerPerf} are network matrices of the considered networked dynamical system $\mathcal{G}_N$. Note also that any $n \times n$ block diagonal matrix $\Theta=\diag(\Theta_{ii}:i\in \N_n)$ will be a network matrix of any arbitrary network with $n\in\N$ subsystems if $\Theta_{ii}$ is specific only to the subsystem $i$. The following lemma summarizes several interesting and useful properties of such network matrices.

\begin{lemma}\label{Lm:NetworkMatrixProperties}
Given a networked dynamical system $\mathcal{G}_n, n\in\N$, a few corresponding block network matrices $\Theta,\Phi,\{\Psi^{kl}:k,l\in\N_m\}$ and some arbitrary block matrix 
\footnote{Note that $\Psi$ is a block matrix of block network matrices.} $\Psi=[\Psi^{kl}]_{k,l \in \N_m}$ (with appropriate block structures):
\begin{enumerate}
    \item $\Theta^\T$, \ $\alpha \Theta + \beta \Phi$ are network matrices for any $\alpha,\beta \in \R$.
    \item $\Phi \Theta$, $\Theta\Phi$ are network matrices whenever $\Phi$ is a block diagonal network matrix.
    \item $[[\Psi^{kl}_{ij}]_{k,l\in\N_m}]_{i,j\in\N_n}$ is a network matrix.
    \item $P \Theta P^\T$ is a network matrix that corresponds to a re-indexed version of the original networked system $\mathcal{G}_n$ whenever $P$ is a block symmetric permutation matrix.
\end{enumerate}
\end{lemma}
\begin{proof}
According to Def. \ref{Def:NetworkMatrices}, any block matrix $W=[W_{ij}]_{i,j\in\N_N}$ is a network matrix if $W_{ij}$ is specific only to the subsystems $i$ and $j$. This view of Def. \ref{Def:NetworkMatrices} can be used to prove the Cases 1-3 respectively as: 
\begin{enumerate}
    \item $W=\Theta^\T \implies W_{ij}=\Theta_{ji}^\T$ and \\ 
    $W=\alpha \Theta + \beta \Phi \implies W_{ij}=\alpha \Theta_{ij} + \beta \Phi_{ij}$;
    \item $W=\Phi\Theta \implies W_{ij}=\Phi_{ii}\Theta_{ij}$ and \\ 
    $W=\Theta \Phi \implies W_{ij}=\Theta_{ij}\Phi_{jj}$; 
    \item $W=[[\Psi^{kl}_{ij}]_{k,l\in\N_m}]_{i,j\in\N_n} \implies W_{ij}=[\Psi^{kl}_{ij}]_{k,l\in\N_m}$;
\end{enumerate}
(where in each case, $W_{ij}$ is specific only to the subsystems $i$ and $j$). The proof of Case 4 is complete by noting that $P \Theta P$ executes identical block row and block column operations on $\Theta$ and $P \Theta P^\T = P \Theta P$. 
\end{proof}

The above lemma allows us to analyze custom block matrices by claiming them to be ``network matrices'' under some additional conditions. For example, if $A,P$ are block network matrices and $P$ is block diagonal, then: (1) $W = -A^\T P - P A$ (see Prop. \ref{Pr:CTLTIStability}) is a network matrix, and (2) if $\scriptsize \Psi=\bm{P & A^\T P\\PA & P}$ is some block matrix (see Prop. \ref{Pr:DTLTIStability}), its \emph{block element-wise} (BEW) form 
$$ W \triangleq \text{BEW}(\Psi) \triangleq \bm{\bm{P_{ii}e_{ij} & A_{ji}^\T P_{jj} \\ P_{ii}A_{ij} & P_{ii}e_{ij}}}_{i,j\in\N_N}$$ is a network matrix\footnote{Recall the notation: $e_{ij}\triangleq \I \cdot \mb{1}_{\{i=j\}}$.}.

\subsubsection{\textbf{Positive Definiteness}}
We next provide several useful lemmas on the positive definiteness property of matrices (in addition to Lm. \ref{Lm:TwoByTwoBlockMatrixPDF} and Lm. \ref{Lm:PreAndPostMultiplication}).

\begin{lemma}(Sylvester's criterion \cite{Antsaklis2006}) \label{Lm:SylvestersCriterion}
A symmetric matrix $W$ is positive definite if and only if its determinants of the leading principal minors are positive. 
\end{lemma}

\begin{lemma}(Cholesky decomposition \cite{Bernstein2009}) \label{Lm:CholeskyDecomposition}
A symmetric matrix $W$ is positive definite (or positive semi-definite) if and only if there exists a lower-triangular matrix $L$ with positive (or non-negative) diagonal entries such that $W=LL^\T$. 
\end{lemma}

The following lemma is parallel to Lm. \ref{Lm:NetworkMatrixProperties}-Case 3.

\begin{lemma}\label{Lm:ColumnandRowPermutations}
Let $\Psi = [\Psi^{kl}]_{k,l\in\N_m}$ be an $m\times m$ block matrix where each of the constituent matrices $\{\Psi^{kl}:k,l\in\N_m\}$ is also an $n \times n$ block matrix (with appropriate dimensions). Then, $\Psi > 0$ iff $W=\text{BEW}(\Psi) \triangleq [[\Psi^{kl}_{ij}]_{k,l\in\N_m}]_{i,j\in\N_n}>0$.
\end{lemma}

\begin{proof}
Notice that $W$ is the block element-wise form of $\Psi$ ($W=\text{BEW}(\Psi)$). Therefore, $W$ can be constructed from $\Psi$ by executing a series of simultaneous row swap and column swap operations on $\Psi$. In other words, we can find a permutation matrix $P$ so that $W=P\Psi P$. Note that this permutation matrix $P$ will be symmetric as all the required row/column operations are simple independent row/column swap operations. Therefore, $W=P \Psi P^\T$. Since all permutation matrices are full-rank, we can apply Lm. \ref{Lm:PreAndPostMultiplication} to arrive at result: $\Psi>0 \iff W = P \Psi P^\T > 0$.
\end{proof}

\subsection{The Main Theoretical Result}

We are now ready to establish our main theoretical result as a lemma, which will be exploited throughout the remainder of this paper. We also acknowledge that different versions of this lemma have already appeared in \cite{Agarwal2021,Agarwal2022} and \cite{Welikala2019P3}, but without rigorous proofs. Therefore, here we provide a concise version of it along with a complete proof. 

\begin{lemma}\label{Lm:MainLemmaShort}
A symmetric $N \times N$ block matrix $W = [W_{ij}]_{i,j\in\N_N} > 0$ iff  $\tilde{W}_{ii} >0, \forall i\in\N_N$ where 
\begin{equation}\label{Eq:Lm:MainLemmaShort}
    \tilde{W}_{ij} = W_{ij} - \sum_{k\in\N_{j-1}} \tilde{W}_{ik}\tilde{W}_{kk}^{-1}\tilde{W}_{jk}^\T, \ \ \ \ \forall j\in\N_{i}.
\end{equation}
\end{lemma}

\begin{proof}
To apply Lm. \ref{Lm:CholeskyDecomposition}, we first determine a lower-triangular block matrix $L=[L_{ij}]_{i,j\in\N_N}$ such that $W=LL^\T$. Note that 
\begin{equation}\label{Eq:Lm:MainLemmaShortStep0}
    W = LL^\T \iff W_{ij} = \sum_{k\in \N_N} L_{ik}L_{jk}^\T = \sum_{k\in \N_{\min\{i,j\}}} L_{ik}L_{jk}^\T,
\end{equation}
for any $i,j\in\N_N$ (the last step is due to $L_{ik}=\0, \forall k>i$ and $L_{jk}=\0, \forall k>j$). In \eqref{Eq:Lm:MainLemmaShortStep0}, to make the term $L_{ij}$ the subject, we consider the case $i \geq j$:  
\begin{equation}
    W_{ij} = \sum_{k\in \N_j} L_{ik}L_{jk}^\T = L_{ij}L_{jj}^\T + \sum_{k\in \N_{j-1}} L_{ik}L_{jk}^\T,
\end{equation}
which gives (also using the fact that $L_{ij} =\0$, for $i<j$)
\begin{equation}\label{Eq:Lm:MainLemmaShortStep1}
    L_{ij} = (W_{ij}-\sum_{k\in\N_{j-1}}L_{ik}L_{jk}^\T)L_{jj}^{-\T}\mb{1}_{\{i \geq j\}} \ \ \ \ \forall i,j\in\N_N.
\end{equation}
With $L=[L_{ij}]_{i,j\in\N_N}$ derived in \eqref{Eq:Lm:MainLemmaShortStep1} we get $W=LL^\T$. Therefore, according to Lm. \ref{Lm:CholeskyDecomposition}, $W>0$ if and only if the diagonal elements of $L$ are positive, i.e., if and only if the diagonal elements of lower-triangular matrices  $\{L_{ii}:i\in\N_N\}$ are positive. Re-using Lm. \ref{Lm:CholeskyDecomposition}, it is easy to see that the latter will occur if and only if $L_{ii}L_{ii}^\T>0, \forall i \in \N_N$. In all, 
\begin{equation}\label{Eq:Lm:MainLemmaShortStep2}
W>0 \iff L_{ii}L_{ii}^\T>0, \ \ \ \ \forall i \in \N_N.
\end{equation}

Now, for the case $j \leq i$ (i.e., $j\in\N_i$), we can state \eqref{Eq:Lm:MainLemmaShortStep1} as 
\begin{equation}\label{Eq:Lm:MainLemmaShortStep3}
    L_{ij} = \tilde{W}_{ij}L_{jj}^{-\T}, 
\end{equation}
where we define $\tilde{W}_{ij}$ as 
\begin{equation}\label{Eq:Lm:MainLemmaShortStep4}
    \tilde{W}_{ij} \triangleq W_{ij}-\sum_{k\in\N_{j-1}}L_{ik}L_{jk}^\T,  \ \ \ \ j \in \N_i.
\end{equation}
From \eqref{Eq:Lm:MainLemmaShortStep3} (or \eqref{Eq:Lm:MainLemmaShortStep1}), when $j=i$, we get the relationship
\begin{equation} \label{Eq:Lm:MainLemmaShortStep5}
    L_{ii}L_{ii}^\T = \tilde{W}_{ii}.
\end{equation}
From \eqref{Eq:Lm:MainLemmaShortStep5} and \eqref{Eq:Lm:MainLemmaShortStep2}, it is clear that $W>0 \iff \tilde{W}_{ii} >0, \forall i\in\N_N$. Therefore, we now only need to prove that \eqref{Eq:Lm:MainLemmaShortStep4} $\iff$ \eqref{Eq:Lm:MainLemmaShort}. For this purpose, we first simplify the $L_{ik}L_{jk}^\T$ term in \eqref{Eq:Lm:MainLemmaShortStep4} using \eqref{Eq:Lm:MainLemmaShortStep3} as 
\begin{align}
    L_{ik}L_{jk}^\T =&\ \tilde{W}_{ik}L_{kk}^{-\T}(\tilde{W}_{jk}L_{kk}^{-\T})^\T = \tilde{W}_{ik} L_{kk}^{-\T}L_{kk}^{-1}\tilde{W}_{jk}^\T \nonumber \\ \label{Eq:Lm:MainLemmaShortStep6}
    =&\ \tilde{W}_{ik}(L_{kk}L_{kk}^\T)^{-1}\tilde{W}_{jk}^\T = \tilde{W}_{ik}\tilde{W}_{kk}^{-1}\tilde{W}_{jk}^\T.  
\end{align}
Finally, applying \eqref{Eq:Lm:MainLemmaShortStep6} in \eqref{Eq:Lm:MainLemmaShortStep4}, we can obtain \eqref{Eq:Lm:MainLemmaShort}. 
\end{proof}

According to the above lemma, testing positive definiteness of an $N\times N$ block matrix $W=[W_{ij}]_{i,j\in\N_N}$ can be broken down to $N$ separate smaller tests (iterations). In particular, at the $i$\tsup{th} iteration, we now only need to test whether $\tilde{W}_{ii}>0$. Since $\tilde{W}_{ii} = W_{ii}-\sum_{k\in\N_{i-1}}\tilde{W}_{ik}\tilde{W}_{kk}^{-1}\tilde{W}_{ik}^\T$ \eqref{Eq:Lm:MainLemmaShort}, computing $\tilde{W}_{ii}$ only requires the following sets of matrices: 
\begin{enumerate}
    \item $\{W_{ij}:j \in \N_i\}$ (extracted from $W$ in iteration $i$);
    \item $\{\tilde{W}_{ij}:j\in \N_{i-1}\}$ (computed using \eqref{Eq:Lm:MainLemmaShort} in iteration $i$);
    \item $\{\{\tilde{W}_{jk}:k\in\N_j\}:j\in\N_{i-1}\}$ (computed in previous $(i-1)$ iterations).
\end{enumerate}

The following corollary of Lm. \ref{Lm:MainLemmaShort} provide more insights on how the information computed in previous $(i-1)$ iterations are used when testing $\tilde{W}_{ii}>0$ at the $i$\tsup{th} iteration.

\begin{corollary}\label{Co:MainLemmaMatrix}
A symmetric $N \times N$ block matrix $W = [W_{ij}]_{i,j\in\N_N} > 0$ iff $\tilde{W}_{ii} > 0, \forall i\in\N_N$ where 
\begin{equation}\label{Eq:Co:MainLemmaMatrix}
    \begin{aligned}
        \tilde{W}_{ii} \triangleq&\ W_{ii} - \tilde{W}_i \mathcal{D}_i \tilde{W}_i^\T,\\
        \tilde{W}_i \triangleq&\ [\tilde{W}_{i1}, \tilde{W}_{i2}, \ldots, \tilde{W}_{i,i-1}] \triangleq W_i(\mathcal{D}_i\mathcal{A}_i^\T)^{-1},\\
        W_i \triangleq&\  [W_{i1},W_{i2},\ldots, W_{i,i-1}], \\
        \mathcal{D}_i \triangleq&\ \diag([\tilde{W}_{11}^{-1},\tilde{W}_{22}^{-1},\ldots,\tilde{W}_{i-1,i-1}^{-1}]),\\
        \mathcal{A}_i \triangleq&\ 
        \bm{
        \tilde{W}_{11} & \0 & \cdots & \0 \\
        \tilde{W}_{21} & \tilde{W}_{22} & \cdots & \0\\
        \vdots & \vdots & \vdots & \vdots \\
        \tilde{W}_{i-1,1} & \tilde{W}_{i-1,2} & \cdots & \tilde{W}_{i-1,i-1}
        }.
    \end{aligned}
\end{equation}
\end{corollary}

\begin{proof}
By simplifying \eqref{Eq:Lm:MainLemmaShort} for each $j \in \N_{i-1}$ and re-arranging its terms to make the term $W_{ij}$ the subject, we can obtain a system of $(i-1)$ equations that can be jointly represented by $W_i = \tilde{W}_i \mathcal{D}_i \mathcal{A}_i^\T$ (using the block matrices defined in \eqref{Eq:Co:MainLemmaMatrix}). Therefore, $\tilde{W}_i = W_i(\mathcal{D}_i\mathcal{A}_i^\T)^{-1}$. 
Finally, by writing \eqref{Eq:Lm:MainLemmaShort} for $j=i$, we get 
$
\tilde{W}_{ii} 
= W_{ii} - \sum_{k\in\N_{i-1}} \tilde{W}_{ik}\tilde{W}_{kk}^{-1}\tilde{W}_{ik}^\T 
= W_{ii} - \tilde{W}_{i}\mathcal{D}_i\tilde{W}_{i}^\T,
$ 
and hence the proof is complete via Lm. \ref{Lm:MainLemmaShort}.
\end{proof}

\subsection{Application to Networked Systems Analysis}

Let the $N \times N$ block matrix $W$ considered in Co. \ref{Co:MainLemmaMatrix} be a network matrix (see Def. \ref{Def:NetworkMatrices}) corresponding to some networked system  $\mathcal{G}_N$ \eqref{Eq:CTNSDynamics}. The following remarks now can be made regarding using Co. \ref{Co:MainLemmaMatrix} in a such network setting.

\subsubsection{\textbf{Decentralized and Compositional Nature}} 
The nature of Co. \ref{Co:MainLemmaMatrix} implies that testing/enforcing $W>0$ can be achieved in a \textbf{decentralized} manner over $\mathcal{G}_N$ by sequentially testing/enforcing $\tilde{W}_{ii}>0$ at each subsystem $\Sigma_i, i\in\N_N$. 

Moreover, during a such process, at a subsystem $\Sigma_i$, it only requires to execute some local computations using some information obtained from the subsystems that came before it (i.e., from $\{\Sigma_j: j\in\N_{i-1}\}$). Therefore, testing/enforcing $W>0$ can be achieved in a \textbf{compositional} manner. 

In other words, adding a new subsystem to $\mathcal{G}_N$ while ensuring the positive definiteness of some overall network matrix (now corresponding to $\mathcal{G}_{N+1}$) can be efficiently and conveniently achieved without having to re-evaluate the local tests/enforcements at the existing subsystems in $\mathcal{G}_N$.

\subsubsection{\textbf{Resilience to Subsystem Removals}}
The aforementioned compositionality property implies that the proposed decentralized approach to test/enforce the positive definiteness of a network matrix (via. Co. \ref{Co:MainLemmaMatrix}) is resilient to \emph{subsystem additions}. It turns out that the proposed approach is also resilient to \emph{subsystem removals}. To understand this, first, note that a removal of a subsystem $\Sigma_i, i\in\N_N$ will change the network matrix $W$ into a smaller $(N-1)\times(N-1)$ block network matrix $\bar{W}$. Algebraically, $\bar{W}$ can be obtained from $W$ by removing its $i$\tsup{th} block row and block column. The following Lm. \ref{Lm:SubsystemRemoval} proves that $W>0 \implies \bar{W}$ for any $i\in\N_N$, i.e., if the network matrix $W$ is positive definite, any possible residual network matrix $\bar{W}$ will also be positive definite. 

\begin{lemma}\label{Lm:SubsystemRemoval}
Given a symmetric $N\times N$ block matrix $W>0$, any $(N-1)\times(N-1)$ block matrix $\bar{W}$ obtained from $W$ by removing its $i$\tsup{th} block row and block column, will retain the positive definiteness, i.e., $\bar{W}>0$, for any $i\in\N_N$.
\end{lemma}
\begin{proof}
Note that, for any $i\in\N_N$, there exists a symmetric permutation matrix $P$ such that 
$$
PWP^\T=
\begin{bmatrix}
\bar{W} & W_i^\T\\
W_i & W_{ii}
\end{bmatrix},
$$
where $[W_i, W_{ii}]$ contains all the blocks of the $i$\tsup{th} block row of $W$. According to the Sylvester's criterion (see Lm. \ref{Lm:SylvestersCriterion}), $PWP^\T > 0 \implies \bar{W}>0$. Moreover, as $P$ is a full-rank matrix (bu definition), from Lm. \ref{Lm:PreAndPostMultiplication}, $W>0 \iff PWP^\T > 0$. Therefore, by combining these two results, we get $W>0 \implies \bar{W}>0$.  
\end{proof}

This result implies that the proposed approach to enforce/test the positive definiteness of a network matrix (via. Co. \ref{Co:MainLemmaMatrix}) is resilient to \emph{subsystem removals}. Repeated use of this result implies that if a network matrix of some network is positive definite, even if several subsystems were removed from that network, the network matrix of the residual network will remain positive definite. Therefore, this eliminates the need to re-evaluate the local tests upon such subsystem removals from a network.

\subsubsection{\textbf{The Algorithm}}

Note that the $j$\tsup{th} block row of the matrix $\mathcal{A}_i$ in \eqref{Eq:Co:MainLemmaMatrix} can be obtained from the information seen at the subsystem $\Sigma_j, j\in \N_{i-1}$ (when $\tilde{W}_{jj}>0$ was tested). In essence, this matrix $\mathcal{A}_i$ can be seen as a compilation of messages received at the subsystem $\Sigma_i$ from previous/existing $(i-1)$ subsystems in the network. On the other hand, the matrix $\mathcal{D}_i$ in \eqref{Eq:Co:MainLemmaMatrix} is fully-determined by this message matrix $\mathcal{A}_i$. Note also that some components of the matrix $W_i$ in \eqref{Eq:Co:MainLemmaMatrix} may still be unknown to the subsystem $i$ if the network is asymmetric. However, such unknown components can be obtained by requesting them from previous/existing $(i-1)$ subsystems in the network. Hence additional communications may be required to create the matrix $W_i$. Finally, note that $W_{ii}$ is known at (in fact, is intrinsic to) the subsystem $i$. Therefore, it is now clear how $\tilde{W}_{ii}$ (of which the positive definiteness needs to be tested) can be obtained using \eqref{Eq:Co:MainLemmaMatrix}. The proposed overall decentralized and compositional approach to test/enforce the positive-definiteness of a network matrix $W$ (based on Co. \ref{Co:MainLemmaMatrix}) is summarized in the following Alg. \ref{Alg:DistributedPositiveDefiniteness}

\begin{algorithm}[!h]
\caption{Testing/Enforcing $W>0$ in a Network Setting.}
\label{Alg:DistributedPositiveDefiniteness}
\begin{algorithmic}[1]
\State \textbf{Input: } $W = [W_{ij}]_{i,j\in\N_N}$
\State \textbf{At each subsystem $\Sigma_i, i \in \N_N$ execute:} 
\Indent
    \If{$i=1$}
        \State Test/Enforce: $W_{11}>0$
        \State Store: $\tilde{W}_1 \triangleq [W_{11}]$ \Comment{To be sent to others.}
    \Else
        \State \textbf{From each subsystem $\Sigma_j, j\in\N_{i-1}$:}
        \Indent
            \State Receive: $\tilde{W}_j \triangleq  [\tilde{W}_{j1},\tilde{W}_{j2},\ldots,\tilde{W}_{jj}]$
            \State Receive: Required info. to compute $W_{ij}$
        \EndIndent
        \State \textbf{End receiving}
        \State Construct: $\mathcal{A}_i, \mathcal{D}_i$ and $W_i$.  
        \Comment{Using: \eqref{Eq:Co:MainLemmaMatrix}.}
        \State Compute: $\tilde{W}_i \triangleq W_i (\mathcal{D}_i\mathcal{A}_i^\T)^{-1}$ 
        \Comment{From \eqref{Eq:Co:MainLemmaMatrix}.}
        \State Compute: $\tilde{W}_{ii} \triangleq W_{ii} - \tilde{W}_{i}\mathcal{D}_i\tilde{W}_i^\T$ 
        \Comment{From \eqref{Eq:Co:MainLemmaMatrix}.}
        \State Test/Enforce: $\tilde{W}_{ii} > 0$
        \State Store: $\tilde{W}_i \triangleq [\tilde{W}_i, \tilde{W}_{ii}]$
        \Comment{To be sent to others}
    \EndIf
\EndIndent
\State \textbf{End execution}
\end{algorithmic}
\end{algorithm}

\subsection{Inter-Subsystem Communications}

\subsubsection{\textbf{Redundant Communications}}

Even though Alg. \ref{Alg:DistributedPositiveDefiniteness} is decentralized and compositional, it is, in general, not distributed. This limitation is evident from the fact that a subsystem $\Sigma_i, i\in\N_N$ having to communicate with all the subsystems that came before it (i.e., with $\{\Sigma_j: j\in \N_{i-1}\}$) so as to construct the matrix $\mathcal{A}_i$ \eqref{Eq:Co:MainLemmaMatrix} when executing Alg. \ref{Alg:DistributedPositiveDefiniteness}.

Nevertheless, we prove the following corollary of Lm. \ref{Lm:MainLemmaShort} to show that some communication sessions between subsystems are redundant and thus can be avoided - depending on the network topology and the used subsystem indexing scheme.

\begin{corollary}\label{Lm:MainLemmaInfo}
A symmetric $N \times N$ block network matrix $W = [W_{ij}]_{i,j\in\N_N} > 0$ iff $\tilde{W}_{ii} >0, \forall i\in\N_N$ where 
\begin{equation}\label{Eq:Lm:MainLemmaInfo}
    \tilde{W}_{ij} = W_{ij} - \sum_{k\in\N_{j-1}\backslash \N_{L_{ij}-1}}\tilde{W}_{ik}\tilde{W}_{kk}^{-1}\tilde{W}_{jk}^\T, \ \ \forall j\in\N_{i}, 
\end{equation}
with $\mathcal{L}_{ij} \triangleq \max\{\min\{\bar{\C}_i\},\min\{\bar{\C}_j\}\}$.
\end{corollary}
\begin{proof}
According to \eqref{Eq:Co:MainLemmaMatrix}, $\mathcal{D}_i\mathcal{A}_i$ is upper-triangular. Therefore, $(\mathcal{D}_i\mathcal{A}_i)^{-1}$ is also upper-triangular. Consequently, due to the relationship $\tilde{W}_i = W_i (\mathcal{D}_i\mathcal{A}_i)^{-1}$ \eqref{Eq:Co:MainLemmaMatrix}, whenever the first $n\in\N_{i-1}$ blocks of $W_i$ \eqref{Eq:Co:MainLemmaMatrix} are zero blocks, the first $n\in\N_{i-1}$ blocks of $\tilde{W}_i$ \eqref{Eq:Co:MainLemmaMatrix} are also zero blocks. 

Since $W$ is a network matrix, $W_{ij}=\0$ for all $i,j\in\N_N$ such that $j \not \in \C_i$ (recall that $\C_i\triangleq \E_i \cup \F_i$, see Def. \ref{Def:NetworkMatrices}). Therefore, the first $\min\{\bar{\C}_i\}-1$ blocks of $W_i$ are zero blocks. Consequently, the first $\min\{\bar{\C}_i\}-1$ blocks of $\tilde{W}_i$ are also zero blocks. Simply, for any $i\in \N_N$, $\tilde{W}_{ij} = \0,  \forall j < \min\{\bar{\C}_i\}$. Similarly, for any $i\in\N_N$, $\tilde{W}_{ik} = \0, \forall k<\min\{\bar{\C}_i \}$ and for any $j\in\N_N$,  $\tilde{W}_{jk} = \0, \forall k < \min\{\bar{\C}_j\}$. 

Therefore, $\tilde{W}_{ik}\tilde{W}_{kk}^{-1}\tilde{W}_{jk}^\T=\0, \forall k<\mathcal{L}_{ij}$ where $\mathcal{L}_{ij}\triangleq \max\{\min\{\bar{\C}_i\},\min\{\bar{\C}_j\}\}$. By applying this result in \eqref{Eq:Lm:MainLemmaShort}, we can directly obtain \eqref{Eq:Lm:MainLemmaInfo}. 
\end{proof}

Compared to \eqref{Eq:Lm:MainLemmaShort}, in \eqref{Eq:Lm:MainLemmaInfo}, there may be less number of terms in the summation as $\mathcal{L}_{ij} \geq 1$. This implies a possible reduction in the communications required at the subsystem $\Sigma_i$ compared to what is required in \eqref{Eq:Lm:MainLemmaShort}. For example, in a  bi-directional network, if a subsystem $\Sigma_j, j \in \N_{i-1}$ is such that $\min\{\bar{\C}_j\} > j-1$ (i.e., $\Sigma_j$ is the least indexed subsystem among its neighbors, which also implies $\mathcal{L}_{ij}>j-1$), then the subsystem $\Sigma_i$ only need to get $\tilde{W}_{jj}$ value from subsystem $\Sigma_j$ as $\tilde{W}_{jk}=\0, \forall k\in\N_{j-1}$ under $\mathcal{L}_{ij}>j-1$. 

At this point, it should be clear that even though the proposed positive definiteness testing/enforcing criterion (i.e., Alg. \ref{Alg:DistributedPositiveDefiniteness}) is not distributed in general, depending on the network topology and the used subsystem indexing scheme, some communication sessions can be avoided. We next propose a communication cost function to find an optimum indexing scheme for a given network topology.

\subsubsection{\textbf{Communication Cost}}

According Co. \ref{Lm:MainLemmaInfo}, to analyze/enforce $\tilde{W}_{ii}>0$ at a subsystem $\Sigma_i, i\in\N_N$, it requires each previous subsystem $\Sigma_j, j\in\N_{i-1}$ to send the information $\{\tilde{W}_{jk}:k\in\N_j\}$ and out of these block matrices, $\mathcal{L}_{ij}-1$ (at most $j-1$) number of block matrices will be $\0$, making them redundant. Let us assume the communication cost associated with a such block matrix $\tilde{W}_{jk}$ as $\gamma_{ijk}$. On the other hand, the subsystem $\Sigma_i$ might also require additional information from each subsystem $\Sigma_j, j\in\N_{i-1}$ just to compute $\{W_{ij}:j\in\N_{i-1}\}$ - if there are unknown components in these block matrices. Note that this requirement arises only if $j \in \F_i$. Let us assume the communication cost associated with the block matrix $W_{ij}$ as $\beta_{ij}$. Taking these costs into account, we can formulate a communication cost function as: 
\begin{equation}\label{Eq:CommunicationCost}
    J(X) = \sum_{i\in\N_N} \sum_{j\in\N_{i-1}} \alpha_{ij} 
    \big(
    \beta_{ij}\mb{1}_{\{j\in\F_i\}} + \gamma_{ijj} + \sum_{k= \mathcal{L}_{ij}}^{j-1} \gamma_{ijk}  \big),
\end{equation}
associated with executing Alg. \ref{Alg:DistributedPositiveDefiniteness} over the considered network topology under the used subsystem indexing scheme $X$.  In \eqref{Eq:CommunicationCost}, $\alpha_{ij}\in\R$ represents the unit cost of communication from subsystem $j$ to $i$. Simply, we can set $\alpha_{ij} = \mb{1}_{\{j\not\in\E_i\}}$ to penalize the communications that happen over subsystems that are not neighbors. Here, the indexing scheme $X$ can be any permutation of $\N_N$ and it determines the neighbor sets (e.g., $\E_i, \F_i, \C_i$) and $\mathcal{L}_{ij}$ values for all $i,j\in \N_N$. Therefore, the objective function $J(X)$ formulated above \eqref{Eq:CommunicationCost} can be used to determine an optimal subsystem indexing scheme that minimizes costly inter-subsystem communications.

The following corollary shows that, in some networks, under some subsystem indexing schemes, we can altogether avoid communications between subsystems that are not neighbors, i.e., we can execute the proposed decentralized process (Alg. \ref{Alg:DistributedPositiveDefiniteness}) in a distributed manner. For example, Figure \ref{Fig:CascadedNetwork} shows such a network configuration. 

\begin{corollary}
Under $\alpha_{ij} \triangleq \mb{1}_{\{j\not\in\E_i\}}$ in  \eqref{Eq:CommunicationCost}, if there exists a subsystem indexing scheme $X$ that makes $\E_i \supseteq \N_{i-1}$ at each subsystem $\Sigma_i, i\in\N_N$, then, $J(X)=0$, i.e., no communications are needed between non-neighbors.  
\end{corollary}
\begin{proof}
Note that $\E_i  \supseteq \N_{i-1} \iff j \in \E_i, \forall j \in \N_{i-1}$. Moreover, since $\alpha_{ij} \triangleq \mb{1}_{\{j\not\in\E_i\}}$, $j \in \E_i \iff \alpha_{ij}=0$. Combining these two relationships, we obtain that 
$$\E_i  \supseteq \N_{i-1} \iff \alpha_{ij}=0, \forall j \in \N_{i-1}.$$
Therefore, if there exists a subsystem indexing scheme $X$ that makes $\E_i  \supseteq \N_{i-1}, \forall i\in\N_N$, it implies that $\alpha_{ij}=0, \forall j \in \N_{i-1}, \forall i\in\N_N$. Applying this in \eqref{Eq:CommunicationCost} we get the communication cost $J(X)$ as $J(X)=0$. 
\end{proof}

\begin{figure}[!h]
    \centering
    \includegraphics[width=3in]{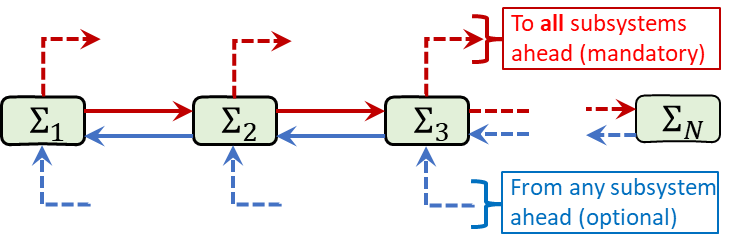}
    \caption{A network configuration where the proposed decentralized process in Alg. \ref{Alg:DistributedPositiveDefiniteness} can be executed without communications between subsystems that are not neighbors, i.e., in a  distributed manner.}
    \label{Fig:CascadedNetwork}
\end{figure}

\paragraph{\textbf{Communication Cost Optimization}}
Now, let us re-state the communication cost function $J(X)$ \eqref{Eq:CommunicationCost} as 
\begin{equation}\label{Eq:CommunicationCost2}
    J(X) = \sum_{i\in\N_N} \sum_{j\in\N_{i-1}} \bar{\alpha}_{ij},
\end{equation}
where we define 
\begin{equation}\label{Eq:ComCostElement}
 \bar{\alpha}_{ij} \triangleq \alpha_{ij}(\beta_{ij}\mb{1}_{\{i\in\E_j\}} + \gamma_{ijj} + \sum_{k= \mathcal{L}_{ij}}^{j-1} \gamma_{ijk}), 
\end{equation}
and $G\triangleq[\bar{\alpha}_{ij}]_{i,j\in\N_N}$, and make the assumption given below.

\begin{assumption}\label{As:LinearComCost}
Each $\bar{\alpha}_{ij}$ value \eqref{Eq:ComCostElement} required in \eqref{Eq:CommunicationCost2} is only dependent on the nature of the interconnection between subsystems $i$ and $j$, but not on specific $i$ and $j$ values.
\end{assumption}

For example, As. \ref{As:LinearComCost} holds if $\bar{\alpha}_{ij} = \mb{1}_{\{j\not\in\E_i\}}$ but not if $\bar{\alpha}_{ij} = j\mb{1}_{\{j\in\E_i\}}$. Pertaining to the definition given in \eqref{Eq:ComCostElement}, note that As. \ref{As:LinearComCost} holds if $\gamma_{ijk}=0, \forall k$ but not if $\gamma_{ijk}=\gamma_{ijj}, \forall k$.

Under As. \ref{As:LinearComCost}, it is easy to see that the problem of finding the optimal subsystem indexing scheme $X$ that minimizes $J(X)$ \eqref{Eq:CommunicationCost2} is identical to the problem of finding the optimal permutation matrix $P\in\R^{N \times N}$ that minimizes the sum of all upper triangular elements of $PGP^\T$ (recall that $G=[\bar{\alpha}_{ij}]_{i,j\in\N_N}$ is a known matrix). Conveniently, the latter problem belongs to a well-known class of combinatorial optimization problems called ``linear ordering problems'' \cite{Ceberio2015,Festa2001,Grotschel1984,DeCani1972,deCani1969}. 

While linear ordering problems are NP-hard \cite{Festa2001}, when $N$ is small, exact solutions can be obtained using brute-force or branch-and-bound algorithms \cite{DeCani1972}. Even if $N$ is large, local optimal solutions can be obtained using efficient and systematic heuristic algorithms \cite{Ceberio2015,Festa2001} that start with a greedy solution and execute local optimizations until a convergence is achieved. In general, linear ordering problems can be formulated as standard linear integer programs \cite{Grotschel1984}, and therefore, can also be solved using commercial solvers.

Finally, to make a remark on situations where As. \ref{As:LinearComCost} does not hold, consider the simple case where $\gamma_{ijk}=\gamma_{ijj}, \forall k$ with $\mathcal{L}_{ij}=1$ in \eqref{Eq:ComCostElement}. This transforms \eqref{Eq:ComCostElement} into the form:
\begin{equation}
    \bar{\alpha}_{ij} = \underbrace{\bar{\alpha}_{ij}\beta_{ij}\mb{1}_{\{i\in\E_j\}}}_{\triangleq\, \bar{\alpha}_{ij}^{(1)}} + j \underbrace{\bar{\alpha}_{ij}\gamma_{ijj}}_{\triangleq\,  \bar{\alpha}_{ij}^{(2)}} = \bar{\alpha}_{ij}^{(1)} + j\bar{\alpha}_{ij}^{(2)}.
\end{equation}
Consequently, the communication cost function $J(X)$ in \eqref{Eq:CommunicationCost2} is non-linear, and thus, the corresponding ``ordering problem'' is also non-linear. To the best of the authors' knowledge, no existing literature directly studies such non-linear ordering problems. However, even for such non-linear ordering problems, we still can use a brute force or a heuristic algorithm (like before) that starts with a greedy solution and executes local optimizations to obtain a locally optimal solution.

Figure \ref{Fig:ComCostNets} shows two example network configurations with the overall inter-subsystem communication cost function $J(X)$ \eqref{Eq:CommunicationCost} defined as: $\alpha_{ij}=\mb{1}_{j\not\in\E_i}$, $\beta_{ij}=2n_j^2$, $\gamma_{ijj}=n_j^2$, $\gamma_{ijk}=n_jn_k$ and $n_i=2,\forall i\in\N_N$. In there, apart from a random nominal subsystem indexing scheme, both the worst and the best possible subsystem indexing schemes obtained by optimizing $J(X)$ (via a brute force algorithm) are indicated respectively using black, red and blue colored texts. Note that, in these two examples, optimizing the subsystem indexing scheme has respectively lead to $75.0\%$ and $59.0\%$ savings in the communication cost function $J(X)$ \eqref{Eq:CommunicationCost}.

\begin{figure}[!h]
    \centering
    \begin{subfigure}{0.24\textwidth}
        \includegraphics[width=1.6in]{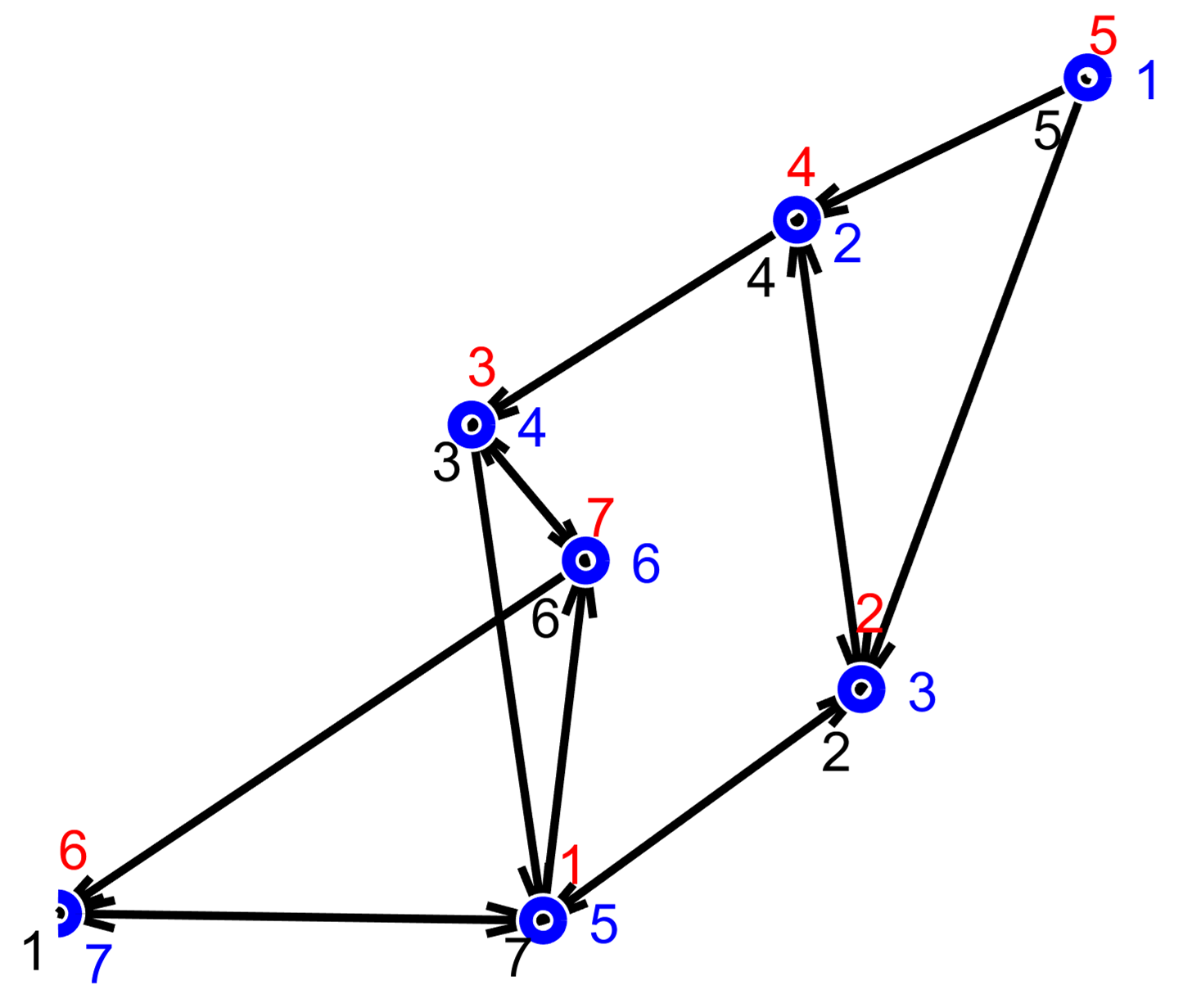}
        \vspace{1mm}
        \subcaption{Ex. 1: $J(X)=\textcolor{red}{220},192,\textcolor{blue}{48}$.}
        \label{Fig:ComCostNet1}
    \end{subfigure}
    \begin{subfigure}{0.24\textwidth}
        \includegraphics[width=1.5in]{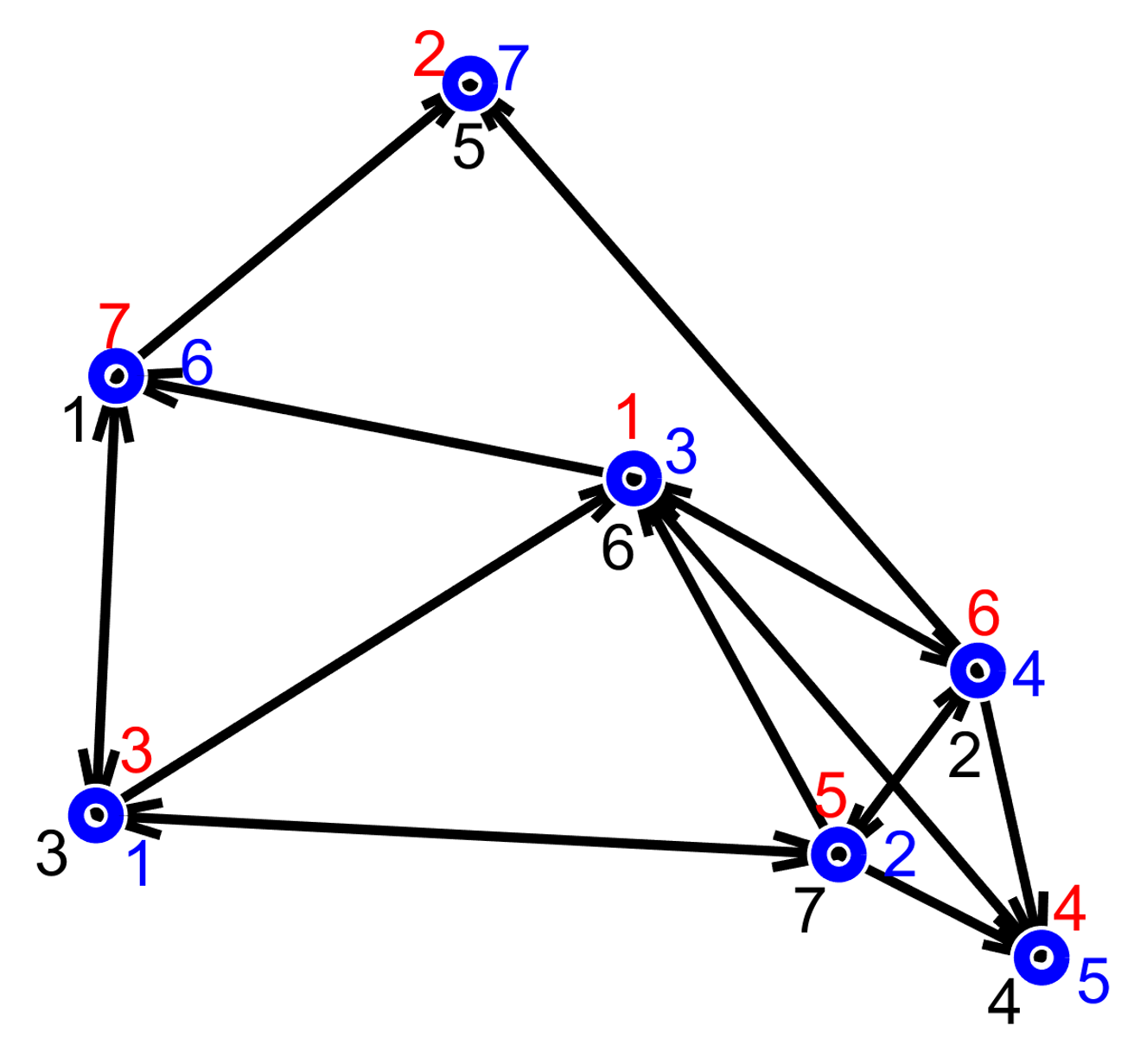}
        \subcaption{Ex. 2: $J(X)=\textcolor{red}{220},156,\textcolor{blue}{64}$.}
        \label{Fig:ComCostNet2}    
    \end{subfigure}
    \caption{Two example network configurations and their communication cost values under different subsystem indexing schemes: (1) worst (red), (2) nominal (black), (3) best (blue).}
    \label{Fig:ComCostNets}
\end{figure}

\subsection{Enforcing Matrix Equatlities}
In this section, so far, we have focused on establishing a decentralized approach for testing/enforcing matrix inequalities of the form $W>0$, i.e., Alg. \ref{Alg:DistributedPositiveDefiniteness}. In the same spirit, we conclude this section by providing a simple decentralized approach for enforcing matrix equalities of the form $V=0$, i.e., Alg. \ref{Alg:DistributedEquality}.

\begin{algorithm}[!h]
\caption{Enforcing $V=0$ in a Network Setting.}
\label{Alg:DistributedEquality}
\begin{algorithmic}[1]
\State \textbf{Input: } $V = [V_{ij}]_{i,j\in\N_N}$
\State \textbf{At each subsystem $\Sigma_i, i \in \N_N$ execute:} 
\Indent
    \If{$i=1$}
        \State Enforce: $V_{11}=0$
    \Else
        \State \textbf{From each subsystem $\Sigma_j, j\in\N_{i-1}$:}
        \Indent
            \State Receive: Required info. to compute $V_{ij},V_{ji}$
        \EndIndent
        \State \textbf{End receiving}
        \State Enforce: $V_{ij} = 0, V_{ji} = 0, \forall j\in\N_{i-1}, V_{ii} = 0$
    \EndIf
\EndIndent
\State \textbf{End execution}
\end{algorithmic}
\end{algorithm}

Consider the case where $V=0$ is a linear matrix equality (LME) in terms of a single controllable matrix variable $X$:
\begin{equation}\label{Eq:LinearMatrixEqualityCondition}
    V \equiv AXB+C = 0,
\end{equation}
with matrices $A,B,C$ and $X$ being network matrices corresponding to some network $\mathcal{G}_N$. For this particular case, the following lemma provides conditions that ensure the applicability of Alg. \ref{Alg:DistributedEquality} to enforce the LME \eqref{Eq:LinearMatrixEqualityCondition}.

\begin{lemma}\label{Lm:matrixEquality}
A unique solution for the linear matrix equality $V=0$ \eqref{Eq:LinearMatrixEqualityCondition} can be found via Alg. \ref{Alg:DistributedEquality} iff $A,B$ are block diagonal with non-singular blocks and: (1) $X$ is a strictly non-block diagonal, or (2) $C,X$ are both block diagonal.
\end{lemma}
\begin{proof}
The proof is complete by considering the $(i,j)$\tsup{th} element of both sides in \eqref{Eq:LinearMatrixEqualityCondition}, $i,j\in\N_N$, which respectively gives unique solutions for $X$ as: (1)
$$
V_{ij} \equiv A_{ii}X_{ij}B_{jj} + C_{ij} = 0 \iff 
X_{ij} = -A_{ii}^{-1}C_{ij}B_{jj}^{-1},
$$
or (2) (recall $e_{ij}\triangleq \I \cdot \mb{1}_{\{i=j\}}$)
$$
V_{ij} \equiv A_{ii}X_{ij}B_{jj} + C_{ii}e_{ij} = 0 \iff 
X_{ij} = -A_{ii}^{-1}C_{ii}e_{ij}B_{jj}^{-1}.
$$
\end{proof}

We conclude this section by pointing out that enforcing LME conditions of the form \eqref{Eq:LinearMatrixEqualityCondition} via Alg. \ref{Alg:DistributedEquality} is not only decentralized and compositional but also distributed.

\section{Distributed Analysis and Control Synthesis of CTNS}
\label{Sec:DistributedTechniquesForCTNS}

In this section, we design decentralized, compositional and possibly distributed techniques for different analysis and control synthesis tasks discussed in Sec. \ref{Sec:BasicsOfCTLTISystems} over CTNSs introduced in Sec. \ref{Sec:CTNetworkedSystem} using the algorithms proposed in Sec. \ref{Sec:DecentralizedAnalysis}.

\begin{landscape}
\begin{table}[]
\centering
\caption{Summary of Results for Decentralized Analysis and Control Synthesis Techniques for Continuous-Time Networked Systems (CTNS)}
\label{Tab:CT-LTILocalResultsSummary}
\resizebox{\columnwidth}{!}{%
\begin{tabular}{|c|l|c|l|llllc|lccl|}
\hline
\multirow{4}{*}{\textbf{Task}} &
  \multicolumn{1}{c|}{\multirow{4}{*}{\textbf{Concept}}} &
  \multirow{4}{*}{\textbf{\begin{tabular}[c]{@{}c@{}}Global\\ LMI\\ Prop. \#\end{tabular}}} &
  \multicolumn{1}{c|}{\multirow{4}{*}{\textbf{Preliminary Conditions}}} &
  \multicolumn{5}{c|}{\textbf{Assumptions on Global CTNS Parameters and Global LMI Variables}} &
  \multicolumn{4}{c|}{\textbf{Local LMI variables, Local LMIs and Local LMEs (at subsystem $\Sigma_i,i\in\N_N$)}} \\ \cline{5-13} 
 &
  \multicolumn{1}{c|}{} &
   &
  \multicolumn{1}{c|}{} &
  \multicolumn{4}{c|}{\textbf{Network   Matrices}} &
  \multirow{3}{*}{\textbf{\begin{tabular}[c]{@{}c@{}}Scaler\\ LMI\\ Var./Obj.\end{tabular}}} &
  \multicolumn{1}{c|}{\multirow{3}{*}{\textbf{LMI Variables}}} &
  \multicolumn{1}{c|}{\multirow{3}{*}{\textbf{\begin{tabular}[c]{@{}c@{}}Scaler\\ LMI\\ Var./Obj.\end{tabular}}}} &
  \multicolumn{1}{c|}{\multirow{3}{*}{\textbf{\begin{tabular}[c]{@{}c@{}}LMIs: BEW\\ Forms of \\ (via Alg. \ref{Alg:DistributedPositiveDefiniteness}):\end{tabular}}}} &
  \multicolumn{1}{c|}{\multirow{3}{*}{\textbf{\begin{tabular}[c]{@{}c@{}}LMEs: BEW\\ Form of \\ (via Alg.\ref{Alg:DistributedEquality}):\end{tabular}}}} \\ \cline{5-8}
 &
  \multicolumn{1}{c|}{} &
   &
  \multicolumn{1}{c|}{} &
  \multicolumn{2}{c|}{\textbf{CTNS Parameters}} &
  \multicolumn{2}{c|}{\textbf{LMI Variables}} &
   &
  \multicolumn{1}{c|}{} &
  \multicolumn{1}{c|}{} &
  \multicolumn{1}{c|}{} &
  \multicolumn{1}{c|}{} \\ \cline{5-8}
 &
  \multicolumn{1}{c|}{} &
   &
  \multicolumn{1}{c|}{} &
  \multicolumn{1}{c|}{\textbf{General}} &
  \multicolumn{1}{c|}{\textbf{Diagonal}} &
  \multicolumn{1}{c|}{\textbf{General}} &
  \multicolumn{1}{c|}{\textbf{Diagonal}} &
   &
  \multicolumn{1}{c|}{} &
  \multicolumn{1}{c|}{} &
  \multicolumn{1}{c|}{} &
  \multicolumn{1}{c|}{} \\ \hline
\multirow{6}{*}{\textbf{\begin{tabular}[c]{@{}c@{}}LTI   \\ System\\ Analysis\end{tabular}}} &
  \textbf{Stability} &
  \textbf{\ref{Pr:CTLTIStability}} &
  \textbf{$u=0,w=0$} &
  \multicolumn{1}{l|}{\textbf{$A$}} &
  \multicolumn{1}{l|}{\textbf{}} &
  \multicolumn{1}{l|}{\textbf{}} &
  \multicolumn{1}{l|}{\textbf{$P$}} &
  \textbf{} &
  \multicolumn{1}{l|}{\textbf{$P_{ii}$}} &
  \multicolumn{1}{c|}{\textbf{}} &
  \multicolumn{1}{c|}{\textbf{\eqref{Eq:Pr:CTLTIStability}}} &
  \textbf{} \\ \cline{2-13} 
 &
  \textbf{Dissipativity} &
  \textbf{\ref{Pr:CTLTIQSRDissipativity}} &
  \textbf{$Q<0,R=R^\T,w=0,u \rightarrow y$} &
  \multicolumn{1}{l|}{\textbf{$A,B,S,R$}} &
  \multicolumn{1}{l|}{\textbf{$C,D,Q$}} &
  \multicolumn{1}{l|}{\textbf{}} &
  \multicolumn{1}{l|}{\textbf{$P$}} &
  \textbf{} &
  \multicolumn{1}{l|}{\textbf{$P_{ii}$}} &
  \multicolumn{1}{c|}{\textbf{}} &
  \multicolumn{1}{c|}{\textbf{\eqref{Eq:Pr:CTLTIQSRDissipativity}}} &
  \textbf{} \\ \cline{2-13} 
 &
  \textbf{\color{red}$\H_2$-Norm} &
  \textbf{\ref{Pr:H2Norm}} &
  \textbf{$\textcolor{red}{D=0},w=0,u\rightarrow y$} &
  \multicolumn{1}{l|}{\textbf{$A,B,C$}} &
  \multicolumn{1}{l|}{\textbf{}} &
  \multicolumn{1}{l|}{\textbf{$Q$}} &
  \multicolumn{1}{l|}{\textbf{$P$}} &
  \textbf{$\gamma$} &
  \multicolumn{1}{l|}{\textbf{$Q_i,P_{ii}$}} &
  \multicolumn{1}{c|}{\textbf{$\gamma_i$}} &
  \multicolumn{1}{c|}{\textbf{\eqref{Eq:Pr:H2Norm1} or \eqref{Eq:Pr:H2Norm2}}} &
  \textbf{} \\ \cline{2-13} 
 &
  \textbf{$H_\infty$-Norm} &
  \textbf{\ref{Pr:HInfNorm}} &
  \textbf{$w=0,u\rightarrow y$} &
  \multicolumn{1}{l|}{\textbf{$A,B,C,D$}} &
  \multicolumn{1}{l|}{\textbf{}} &
  \multicolumn{1}{l|}{\textbf{}} &
  \multicolumn{1}{l|}{\textbf{$P$}} &
  \textbf{$\gamma$} &
  \multicolumn{1}{l|}{\textbf{$P_{ii}$}} &
  \multicolumn{1}{c|}{\textbf{$\gamma_i$}} &
  \multicolumn{1}{c|}{\textbf{\eqref{Eq:Pr:HInfNorm}}} &
  \textbf{} \\ \cline{2-13} 
 &
  \textbf{Stabilizability} &
  \textbf{\ref{Pr:Stabilizability}} &
  \textbf{} &
  \multicolumn{1}{l|}{\textbf{$A$}} &
  \multicolumn{1}{l|}{\textbf{$B$}} &
  \multicolumn{1}{l|}{\textbf{}} &
  \multicolumn{1}{l|}{\textbf{$P$}} &
  \textbf{} &
  \multicolumn{1}{l|}{\textbf{$P_{ii}$}} &
  \multicolumn{1}{c|}{\textbf{}} &
  \multicolumn{1}{c|}{\textbf{\eqref{Eq:Pr:Stabilizability}}} &
  \textbf{} \\ \cline{2-13} 
 &
  \textbf{Detectability} &
  \textbf{\ref{Pr:Detectability}} &
  \textbf{} &
  \multicolumn{1}{l|}{\textbf{$A$}} &
  \multicolumn{1}{l|}{\textbf{$C$}} &
  \multicolumn{1}{l|}{\textbf{}} &
  \multicolumn{1}{l|}{\textbf{$P$}} &
  \textbf{} &
  \multicolumn{1}{l|}{\textbf{$P_{ii}$}} &
  \multicolumn{1}{c|}{\textbf{}} &
  \multicolumn{1}{c|}{\textbf{\eqref{Eq:Pr:Detectability}}} &
  \textbf{} \\ \hline
\multirow{4}{*}{\textbf{\begin{tabular}[c]{@{}c@{}}FSF\\ Controller\\ Synthesis\end{tabular}}} &
  \textbf{Stability} &
  \textbf{\ref{Pr:StabilizationUnderFSF}} &
  \textbf{$D=0,w=0$} &
  \multicolumn{1}{l|}{\textbf{$A$}} &
  \multicolumn{1}{l|}{\textbf{$B$}} &
  \multicolumn{1}{l|}{\textbf{$L$}} &
  \multicolumn{1}{l|}{\textbf{$M$}} &
  \textbf{} &
  \multicolumn{1}{l|}{\textbf{$L_i,M_{ii}$}} &
  \multicolumn{1}{c|}{\textbf{}} &
  \multicolumn{1}{c|}{\textbf{\eqref{Eq:Pr:StabilizationUnderFSF}}} &
  \textbf{$K=LM^{-1}$} \\ \cline{2-13} 
 &
  \textbf{Dissipativity} &
  \textbf{\ref{Pr:DissipativationUnderFSF}} &
  \textbf{$D=0,Q<0,R=R^\T,w\rightarrow y$} &
  \multicolumn{1}{l|}{\textbf{$A,E,S,R$}} &
  \multicolumn{1}{l|}{\textbf{$B,C,F,Q$}} &
  \multicolumn{1}{l|}{\textbf{$L$}} &
  \multicolumn{1}{l|}{\textbf{$M$}} &
  \textbf{} &
  \multicolumn{1}{l|}{\textbf{$L_i,M_{ii}$}} &
  \multicolumn{1}{c|}{\textbf{}} &
  \multicolumn{1}{c|}{\textbf{\eqref{Eq:Pr:DissipativationUnderFSF}}} &
  \textbf{$K=LM^{-1}$} \\ \cline{2-13} 
 &
  \textbf{\color{red}$\H_2$-Norm} &
  \textbf{\ref{Pr:H2ControlUnderFSF}} &
  \textbf{$\textcolor{red}{F=0},w\rightarrow y$} &
  \multicolumn{1}{l|}{\textbf{$A,C,E$}} &
  \multicolumn{1}{l|}{\textbf{$B,D$}} &
  \multicolumn{1}{l|}{\textbf{$L,Q$}} &
  \multicolumn{1}{l|}{\textbf{$M$}} &
  \textbf{$\gamma$} &
  \multicolumn{1}{l|}{\textbf{$L_i,Q_i,M_{ii}$}} &
  \multicolumn{1}{c|}{\textbf{$\gamma_i$}} &
  \multicolumn{1}{c|}{\textbf{\eqref{Eq:Pr:H2ControlUnderFSF}}} &
  \textbf{$K=LM^{-1}$} \\ \cline{2-13} 
 &
  \textbf{$H_\infty$-Norm} &
  \textbf{\ref{Pr:HInfControlUnderFSF}} &
  \textbf{$w\rightarrow y$} &
  \multicolumn{1}{l|}{\textbf{$A,C,E,F$}} &
  \multicolumn{1}{l|}{\textbf{$B,D$}} &
  \multicolumn{1}{l|}{\textbf{$L$}} &
  \multicolumn{1}{l|}{\textbf{$M$}} &
  \textbf{$\gamma$} &
  \multicolumn{1}{l|}{\textbf{$L_i,M_{ii}$}} &
  \multicolumn{1}{c|}{\textbf{$\gamma_i$}} &
  \multicolumn{1}{c|}{\textbf{\eqref{Eq:Pr:HInfControlUnderFSF}}} &
  \textbf{$K=LM^{-1}$} \\ \hline
\multirow{4}{*}{\textbf{\begin{tabular}[c]{@{}c@{}}Observer\\ Design\end{tabular}}} &
  \textbf{Stability} &
  \textbf{\ref{Pr:Observer}} &
  \textbf{$w=0$} &
  \multicolumn{1}{l|}{\textbf{$A$}} &
  \multicolumn{1}{l|}{\textbf{$C,D$}} &
  \multicolumn{1}{l|}{\textbf{$K$}} &
  \multicolumn{1}{l|}{\textbf{$P$}} &
  \textbf{} &
  \multicolumn{1}{l|}{\textbf{$K_i,P_{ii}$}} &
  \multicolumn{1}{c|}{\textbf{}} &
  \multicolumn{1}{c|}{\textbf{\eqref{Eq:Pr:Observer}}} &
  \textbf{$L=P^{-1}K$, \eqref{Eq:LuenbergerObserverParameters}} \\ \cline{2-13} 
 &
  \textbf{Dissipativity} &
  \textbf{\ref{Pr:DissipativeObserver}} &
  \textbf{$Q<0,R=R^\T,w\rightarrow z$} &
  \multicolumn{1}{l|}{\textbf{$A,E,S,R$}} &
  \multicolumn{1}{l|}{\textbf{$C,D,F,G,J,Q$}} &
  \multicolumn{1}{l|}{\textbf{$K$}} &
  \multicolumn{1}{l|}{\textbf{$P$}} &
  \textbf{} &
  \multicolumn{1}{l|}{\textbf{$K_i,P_{ii}$}} &
  \multicolumn{1}{c|}{\textbf{}} &
  \multicolumn{1}{c|}{\textbf{\eqref{Eq:Pr:DissipativeObserver}}} &
  \textbf{$L=P^{-1}K$, \eqref{Eq:LuenbergerObserverParameters}} \\ \cline{2-13} 
 &
  \textbf{\color{red}$\H_2$-Norm} &
  \textbf{\ref{Pr:H2Observer}} &
  \textbf{$\textcolor{red}{J=0},w\rightarrow z$} &
  \multicolumn{1}{l|}{\textbf{$A,E,G$}} &
  \multicolumn{1}{l|}{\textbf{$C,D,F$}} &
  \multicolumn{1}{l|}{\textbf{$K,Q$}} &
  \multicolumn{1}{l|}{\textbf{$P$}} &
  \textbf{$\gamma$} &
  \multicolumn{1}{l|}{\textbf{$K_i,Q_i,P_{ii}$}} &
  \multicolumn{1}{c|}{\textbf{$\gamma_i$}} &
  \multicolumn{1}{c|}{\textbf{\eqref{Eq:Pr:H2Observer}}} &
  \textbf{$L=P^{-1}K$, \eqref{Eq:LuenbergerObserverParameters}} \\ \cline{2-13} 
 &
  \textbf{$H_\infty$-Norm} &
  \textbf{\ref{Pr:HInfObserver}} &
  \textbf{$w \rightarrow z$} &
  \multicolumn{1}{l|}{\textbf{$A,E,G,J$}} &
  \multicolumn{1}{l|}{\textbf{$C,D,F$}} &
  \multicolumn{1}{l|}{\textbf{$K$}} &
  \multicolumn{1}{l|}{\textbf{$P$}} &
  \textbf{$\gamma$} &
  \multicolumn{1}{l|}{\textbf{$K_i,P_{ii}$}} &
  \multicolumn{1}{c|}{\textbf{$\gamma_i$}} &
  \multicolumn{1}{c|}{\textbf{\eqref{Eq:Pr:HInfObserver}}} &
  \textbf{$L=P^{-1}K$, \eqref{Eq:LuenbergerObserverParameters}} \\ \hline
\multirow{4}{*}{\textbf{\begin{tabular}[c]{@{}c@{}}DOF\\ Controller\\ Synthesis\end{tabular}}} &
  \textbf{Stability} &
  \textbf{\ref{Pr:StabilizationUnderDOF}} &
  \textbf{$D=0,w=0$} &
  \multicolumn{1}{l|}{\textbf{$A$}} &
  \multicolumn{1}{l|}{\textbf{$B,C$}} &
  \multicolumn{1}{l|}{\textbf{$A_n,B_n,C_n,D_n$}} &
  \multicolumn{1}{l|}{\textbf{$X,Y$}} &
  \textbf{} &
  \multicolumn{1}{l|}{\textbf{$A_{n,i},B_{n,i},C_{n,i},D_{n,i},X_{ii},Y_{ii}$}} &
  \multicolumn{1}{c|}{\textbf{}} &
  \multicolumn{1}{c|}{\textbf{\eqref{Eq:Pr:StabilizationUnderDOF1},\eqref{Eq:Pr:StabilizationUnderDOF2}}} &
  \textbf{\eqref{Eq:COVsAuxToDOF}} \\ \cline{2-13} 
 &
  \textbf{Dissipativity} &
  \textbf{\ref{Pr:DissipativationUsingDOF}} &
  \textbf{$D=0,Q<0,R=R^\T,w\rightarrow z$} &
  \multicolumn{1}{l|}{\textbf{$A,E,G,J,R$}} &
  \multicolumn{1}{l|}{\textbf{$B,C,F,H,S,Q$}} &
  \multicolumn{1}{l|}{\textbf{$A_n,B_n,C_n,D_n$}} &
  \multicolumn{1}{l|}{\textbf{$X,Y$}} &
  \textbf{} &
  \multicolumn{1}{l|}{\textbf{$A_{n,i},B_{n,i},C_{n,i},D_{n,i},X_{ii},Y_{ii}$}} &
  \multicolumn{1}{c|}{\textbf{}} &
  \multicolumn{1}{c|}{\textbf{\eqref{Eq:Pr:DissipativationUsingDOF1},\eqref{Eq:Pr:DissipativationUsingDOF2}}} &
  \textbf{\eqref{Eq:COVsAuxToDOF}} \\ \cline{2-13} 
 &
  \textbf{\color{red}$\H_2$-Norm} &
  \textbf{\ref{Pr:H2ControlUnderDOF}} &
  \textbf{$D=0,w\rightarrow z$} &
  \multicolumn{1}{l|}{\textbf{$A,E,G,J$}} &
  \multicolumn{1}{l|}{\textbf{$B,C,F,H$}} &
  \multicolumn{1}{l|}{\textbf{$A_n,B_n,C_n,D_n,Q$}} &
  \multicolumn{1}{l|}{\textbf{$X,Y$}} &
  \textbf{$\gamma$} &
  \multicolumn{1}{l|}{\textbf{$A_{n,i},B_{n,i},C_{n,i},D_{n,i},Q_iX_{ii},Y_{ii}$}} &
  \multicolumn{1}{c|}{\textbf{$\gamma_i$}} &
  \multicolumn{1}{c|}{\textbf{\eqref{Eq:Pr:H2ControlUnderDOF1}, \eqref{Eq:Pr:H2ControlUnderDOF2}}} &
  \textbf{\textcolor{red}{$J+HD_nF=0$}, \eqref{Eq:COVsAuxToDOF}} \\ \cline{2-13} 
 &
  \textbf{$H_\infty$-Norm} &
  \textbf{\ref{Pr:HInfControlUnderDOF}} &
  \textbf{$D=0,w\rightarrow z$} &
  \multicolumn{1}{l|}{\textbf{$A,E,G,J$}} &
  \multicolumn{1}{l|}{\textbf{$B,C,F,H$}} &
  \multicolumn{1}{l|}{\textbf{$A_n,B_n,C_n,D_n$}} &
  \multicolumn{1}{l|}{\textbf{$X,Y$}} &
  \textbf{$\gamma$} &
  \multicolumn{1}{l|}{\textbf{$A_{n,i},B_{n,i},C_{n,i},D_{n,i},X_{ii},Y_{ii}$}} &
  \multicolumn{1}{c|}{\textbf{$\gamma_i$}} &
  \multicolumn{1}{c|}{\textbf{\eqref{Eq:Pr:HInfControlUnderDOF1}, \eqref{Eq:Pr:HInfControlUnderDOF2}}} &
  \textbf{\eqref{Eq:COVsAuxToDOF}} \\ \hline
\end{tabular}%
}
\end{table}

\begin{table}[]
\centering
\caption{Summary of Results for Decentralized Analysis and Control Synthesis Techniques for Discrete-Time Networked Systems (DTNS)}
\label{Tab:DT-LTILocalResultsSummary}
\resizebox{\columnwidth}{!}{%
\begin{tabular}{|c|l|c|l|llllc|lccl|}
\hline
\multirow{4}{*}{\textbf{Task}} &
  \multicolumn{1}{c|}{\multirow{4}{*}{\textbf{Concept}}} &
  \multirow{4}{*}{\textbf{\begin{tabular}[c]{@{}c@{}}Global\\ LMI\\ Prop. \#\end{tabular}}} &
  \multicolumn{1}{c|}{\multirow{4}{*}{\textbf{Preliminary Conditions}}} &
  \multicolumn{5}{c|}{\textbf{Assumptions on Global DTNS Parameters and Global LMI Variables}} &
  \multicolumn{4}{c|}{\textbf{Local LMI variables, Local LMIs and Local LMEs (at subsystem $\Sigma_i,i\in\N_N$)}} \\ \cline{5-13} 
 &
  \multicolumn{1}{c|}{} &
   &
  \multicolumn{1}{c|}{} &
  \multicolumn{4}{c|}{\textbf{Network   Matrices}} &
  \multirow{3}{*}{\textbf{\begin{tabular}[c]{@{}c@{}}Scaler\\ LMI\\ Var./Obj.\end{tabular}}} &
  \multicolumn{1}{c|}{\multirow{3}{*}{\textbf{LMI Variables}}} &
  \multicolumn{1}{c|}{\multirow{3}{*}{\textbf{\begin{tabular}[c]{@{}c@{}}Scaler\\ LMI\\ Var./Obj.\end{tabular}}}} &
\multicolumn{1}{c|}{\multirow{3}{*}{\textbf{\begin{tabular}[c]{@{}c@{}}LMIs: BEW\\ Forms of \\ (via Alg. \ref{Alg:DistributedPositiveDefiniteness}):\end{tabular}}}} &
  \multicolumn{1}{c|}{\multirow{3}{*}{\textbf{\begin{tabular}[c]{@{}c@{}}LMEs: BEW\\ Form of \\ (via Alg.\ref{Alg:DistributedEquality}):\end{tabular}}}} \\ \cline{5-8}
 &
  \multicolumn{1}{c|}{} &
   &
  \multicolumn{1}{c|}{} &
  \multicolumn{2}{c|}{\textbf{DTNS Parameters}} &
  \multicolumn{2}{c|}{\textbf{LMI Variables}} &
   &
  \multicolumn{1}{c|}{} &
  \multicolumn{1}{c|}{} &
  \multicolumn{1}{c|}{} &
  \multicolumn{1}{c|}{} \\ \cline{5-8}
 &
  \multicolumn{1}{c|}{} &
   &
  \multicolumn{1}{c|}{} &
  \multicolumn{1}{c|}{\textbf{General}} &
  \multicolumn{1}{c|}{\textbf{Diagonal}} &
  \multicolumn{1}{c|}{\textbf{General}} &
  \multicolumn{1}{c|}{\textbf{Diagonal}} &
   &
  \multicolumn{1}{c|}{} &
  \multicolumn{1}{c|}{} &
  \multicolumn{1}{c|}{} &
  \multicolumn{1}{c|}{} \\ \hline
\multirow{6}{*}{\textbf{\begin{tabular}[c]{@{}c@{}}LTI   \\ System\\ Analysis\end{tabular}}} &
  \textbf{Stability} &
  \textbf{\ref{Pr:DTLTIStability}} &
  \textbf{$u=0,w=0$} &
  \multicolumn{1}{l|}{\textbf{$A$}} &
  \multicolumn{1}{l|}{\textbf{}} &
  \multicolumn{1}{l|}{\textbf{}} &
  \multicolumn{1}{l|}{\textbf{$P$}} &
  \textbf{} &
  \multicolumn{1}{l|}{\textbf{$P_{ii}$}} &
  \multicolumn{1}{c|}{\textbf{}} &
  \multicolumn{1}{c|}{\textbf{\eqref{Eq:Pr:DTLTIStability}}} &
  \textbf{} \\ \cline{2-13} 
 &
  \textbf{Dissipativity} &
  \textbf{\ref{Pr:DTLTIQSRDissipativity}} &
  \textbf{$Q<0,R=R^\T,w=0,u \rightarrow y$} &
  \multicolumn{1}{l|}{\textbf{$A,B,S,R$}} &
  \multicolumn{1}{l|}{\textbf{$C,D,Q$}} &
  \multicolumn{1}{l|}{\textbf{}} &
  \multicolumn{1}{l|}{\textbf{$P$}} &
  \textbf{} &
  \multicolumn{1}{l|}{\textbf{$P_{ii}$}} &
  \multicolumn{1}{c|}{\textbf{}} &
  \multicolumn{1}{c|}{\textbf{\eqref{Eq:Pr:DTLTIQSRDissipativity}}} &
  \textbf{} \\ \cline{2-13} 
 &
  \textbf{\color{red}$\H_2$-Norm} &
  \textbf{\ref{Pr:DTH2Norm}} &
  \textbf{$w=0,u\rightarrow y$} &
  \multicolumn{1}{l|}{\textbf{$A,B,C,\textcolor{red}{D}$}} &
  \multicolumn{1}{l|}{\textbf{}} &
  \multicolumn{1}{l|}{\textbf{$Q$}} &
  \multicolumn{1}{l|}{\textbf{$P$}} &
  \textbf{$\gamma$} &
  \multicolumn{1}{l|}{\textbf{$Q_i,P_{ii}$}} &
  \multicolumn{1}{c|}{\textbf{$\gamma_i$}} &
  \multicolumn{1}{c|}{\textbf{\eqref{Eq:Pr:DTH2Norm1} or \eqref{Eq:Pr:DTH2Norm2}}} &
  \textbf{} \\ \cline{2-13} 
 &
  \textbf{$H_\infty$-Norm} &
  \textbf{\ref{Pr:DTHInfNorm}} &
  \textbf{$w=0,u\rightarrow y$} &
  \multicolumn{1}{l|}{\textbf{$A,B,C,D$}} &
  \multicolumn{1}{l|}{\textbf{}} &
  \multicolumn{1}{l|}{\textbf{}} &
  \multicolumn{1}{l|}{\textbf{$P$}} &
  \textbf{$\gamma$} &
  \multicolumn{1}{l|}{\textbf{$P_{ii}$}} &
  \multicolumn{1}{c|}{\textbf{$\gamma_i$}} &
  \multicolumn{1}{c|}{\textbf{\eqref{Eq:Pr:DTHInfNorm1} or \eqref{Eq:Pr:DTHInfNorm1}}} &
  \textbf{} \\ \cline{2-13} 
 &
  \textbf{Stabilizability} &
  \textbf{\ref{Pr:DTStabilizability}} &
  \textbf{} &
  \multicolumn{1}{l|}{\textbf{$A$}} &
  \multicolumn{1}{l|}{\textbf{$B$}} &
  \multicolumn{1}{l|}{\textbf{}} &
  \multicolumn{1}{l|}{\textbf{$P$}} &
  \textbf{} &
  \multicolumn{1}{l|}{\textbf{$P_{ii}$}} &
  \multicolumn{1}{c|}{\textbf{}} &
  \multicolumn{1}{c|}{\textbf{\eqref{Eq:Pr:DTStabilizability}}} &
  \textbf{} \\ \cline{2-13} 
 &
  \textbf{Detectability} &
  \textbf{\ref{Pr:DTDetectability}} &
  \textbf{} &
  \multicolumn{1}{l|}{\textbf{$A$}} &
  \multicolumn{1}{l|}{\textbf{$C$}} &
  \multicolumn{1}{l|}{\textbf{}} &
  \multicolumn{1}{l|}{\textbf{$P$}} &
  \textbf{} &
  \multicolumn{1}{l|}{\textbf{$P_{ii}$}} &
  \multicolumn{1}{c|}{\textbf{}} &
  \multicolumn{1}{c|}{\textbf{\eqref{Eq:Pr:DTDetectability}}} &
  \textbf{} \\ \hline
\multirow{4}{*}{\textbf{\begin{tabular}[c]{@{}c@{}}FSF\\ Controller\\ Synthesis\end{tabular}}} &
  \textbf{Stability} &
  \textbf{\ref{Pr:DTStabilizationUnderFSF}} &
  \textbf{$D=0,w=0$} &
  \multicolumn{1}{l|}{\textbf{$A$}} &
  \multicolumn{1}{l|}{\textbf{$B$}} &
  \multicolumn{1}{l|}{\textbf{$L$}} &
  \multicolumn{1}{l|}{\textbf{$M$}} &
  \textbf{} &
  \multicolumn{1}{l|}{\textbf{$L_i,M_{ii}$}} &
  \multicolumn{1}{c|}{\textbf{}} &
  \multicolumn{1}{c|}{\textbf{\eqref{Eq:Pr:DTStabilizationUnderFSF}}} &
  \textbf{$K=LM^{-1}$} \\ \cline{2-13} 
 &
  \textbf{Dissipativity} &
  \textbf{\ref{Pr:DTDissipativationUnderFSF}} &
  \textbf{$D=0,Q<0,R=R^\T,w\rightarrow y$} &
  \multicolumn{1}{l|}{\textbf{$A,E,S,R$}} &
  \multicolumn{1}{l|}{\textbf{$B,C,F,Q$}} &
  \multicolumn{1}{l|}{\textbf{$L$}} &
  \multicolumn{1}{l|}{\textbf{$M$}} &
  \textbf{} &
  \multicolumn{1}{l|}{\textbf{$L_i,M_{ii}$}} &
  \multicolumn{1}{c|}{\textbf{}} &
  \multicolumn{1}{c|}{\textbf{\eqref{Eq:Pr:DTDissipativationUnderFSF}}} &
  \textbf{$K=LM^{-1}$} \\ \cline{2-13} 
 &
  \textbf{\color{red}$\H_2$-Norm} &
  \textbf{\ref{Pr:DTH2ControlUnderFSF}} &
  \textbf{$w\rightarrow y$} &
  \multicolumn{1}{l|}{\textbf{$A,C,E,\textcolor{red}{F}$}} &
  \multicolumn{1}{l|}{\textbf{$B,D$}} &
  \multicolumn{1}{l|}{\textbf{$L,Q$}} &
  \multicolumn{1}{l|}{\textbf{$M$}} &
  \textbf{$\gamma$} &
  \multicolumn{1}{l|}{\textbf{$L_i,Q_i,M_{ii}$}} &
  \multicolumn{1}{c|}{\textbf{$\gamma_i$}} &
  \multicolumn{1}{c|}{\textbf{\eqref{Eq:Pr:DTH2ControlUnderFSF}}} &
  \textbf{$K=LM^{-1}$} \\ \cline{2-13} 
 &
  \textbf{$H_\infty$-Norm} &
  \textbf{\ref{Pr:DTHInfControlUnderFSF}} &
  \textbf{$w\rightarrow y$} &
  \multicolumn{1}{l|}{\textbf{$A,C,E,F$}} &
  \multicolumn{1}{l|}{\textbf{$B,D$}} &
  \multicolumn{1}{l|}{\textbf{$L$}} &
  \multicolumn{1}{l|}{\textbf{$M$}} &
  \textbf{$\gamma$} &
  \multicolumn{1}{l|}{\textbf{$L_i,M_{ii}$}} &
  \multicolumn{1}{c|}{\textbf{$\gamma_i$}} &
  \multicolumn{1}{c|}{\textbf{\eqref{Eq:Pr:DTHInfControlUnderFSF}}} &
  \textbf{$K=LM^{-1}$} \\ \hline
\multirow{4}{*}{\textbf{\begin{tabular}[c]{@{}c@{}}Observer\\ Design\end{tabular}}} &
  \textbf{Stability} &
  \textbf{\ref{Pr:DTObserver}} &
  \textbf{$w=0$} &
  \multicolumn{1}{l|}{\textbf{$A$}} &
  \multicolumn{1}{l|}{\textbf{$C,D$}} &
  \multicolumn{1}{l|}{\textbf{$K$}} &
  \multicolumn{1}{l|}{\textbf{$P$}} &
  \textbf{} &
  \multicolumn{1}{l|}{\textbf{$K_i,P_{ii}$}} &
  \multicolumn{1}{c|}{\textbf{}} &
  \multicolumn{1}{c|}{\textbf{\eqref{Eq:Pr:DTObserver}}} &
  \textbf{$L=P^{-1}K$, \eqref{Eq:LuenbergerObserverParameters}} \\ \cline{2-13} 
 &
  \textbf{Dissipativity} &
  \textbf{\ref{Pr:DTDissipativeObserver}} &
  \textbf{$Q<0,R=R^\T,w\rightarrow z$} &
  \multicolumn{1}{l|}{\textbf{$A,E,S,R$}} &
  \multicolumn{1}{l|}{\textbf{$C,D,F,G,J,Q$}} &
  \multicolumn{1}{l|}{\textbf{$K$}} &
  \multicolumn{1}{l|}{\textbf{$P$}} &
  \textbf{} &
  \multicolumn{1}{l|}{\textbf{$K_i,P_{ii}$}} &
  \multicolumn{1}{c|}{\textbf{}} &
  \multicolumn{1}{c|}{\textbf{\eqref{Eq:Pr:DTDissipativeObserver}}} &
  \textbf{$L=P^{-1}K$, \eqref{Eq:LuenbergerObserverParameters}} \\ \cline{2-13} 
 &
  \textbf{\color{red}$\H_2$-Norm} &
  \textbf{\ref{Pr:DTH2Observer}} &
  \textbf{$w\rightarrow z$} &
  \multicolumn{1}{l|}{\textbf{$A,E,G,\textcolor{red}{J}$}} &
  \multicolumn{1}{l|}{\textbf{$C,D,F$}} &
  \multicolumn{1}{l|}{\textbf{$K,Q$}} &
  \multicolumn{1}{l|}{\textbf{$P$}} &
  \textbf{$\gamma$} &
  \multicolumn{1}{l|}{\textbf{$K_i,Q_i,P_{ii}$}} &
  \multicolumn{1}{c|}{\textbf{$\gamma_i$}} &
  \multicolumn{1}{c|}{\textbf{\eqref{Eq:Pr:DTH2Observer}}} &
  \textbf{$L=P^{-1}K$, \eqref{Eq:LuenbergerObserverParameters}} \\ \cline{2-13} 
 &
  \textbf{$H_\infty$-Norm} &
  \textbf{\ref{Pr:DTHInfObserver}} &
  \textbf{$w \rightarrow z$} &
  \multicolumn{1}{l|}{\textbf{$A,E,G,J$}} &
  \multicolumn{1}{l|}{\textbf{$C,D,F$}} &
  \multicolumn{1}{l|}{\textbf{$K$}} &
  \multicolumn{1}{l|}{\textbf{$P$}} &
  \textbf{$\gamma$} &
  \multicolumn{1}{l|}{\textbf{$K_i,P_{ii}$}} &
  \multicolumn{1}{c|}{\textbf{$\gamma_i$}} &
  \multicolumn{1}{c|}{\textbf{\eqref{Eq:Pr:DTHInfObserver}}} &
  \textbf{$L=P^{-1}K$, \eqref{Eq:LuenbergerObserverParameters}} \\ \hline
\multirow{4}{*}{\textbf{\begin{tabular}[c]{@{}c@{}}DOF\\ Controller\\ Synthesis\end{tabular}}} &
  \textbf{Stability} &
  \textbf{\ref{Pr:DTStabilizationUnderDOF}} &
  \textbf{$D=0,w=0$} &
  \multicolumn{1}{l|}{\textbf{$A$}} &
  \multicolumn{1}{l|}{\textbf{$B,C$}} &
  \multicolumn{1}{l|}{\textbf{$A_n,B_n,C_n,D_n$}} &
  \multicolumn{1}{l|}{\textbf{$X,Y$}} &
  \textbf{} &
  \multicolumn{1}{l|}{\textbf{$A_{n,i},B_{n,i},C_{n,i},D_{n,i},X_{ii},Y_{ii}$}} &
  \multicolumn{1}{c|}{\textbf{}} &
  \multicolumn{1}{c|}{\textbf{\eqref{Eq:Pr:DTStabilizationUnderDOF1},\eqref{Eq:Pr:DTStabilizationUnderDOF2}}} &
  \textbf{\eqref{Eq:COVsAuxToDOF}} \\ \cline{2-13} 
 &
  \textbf{Dissipativity} &
  \textbf{\ref{Pr:DTDissipativationUsingDOF}} &
  \textbf{$D=0,Q<0,R=R^\T,w\rightarrow z$} &
  \multicolumn{1}{l|}{\textbf{$A,E,G,J,R$}} &
  \multicolumn{1}{l|}{\textbf{$B,C,F,H,S,Q$}} &
  \multicolumn{1}{l|}{\textbf{$A_n,B_n,C_n,D_n$}} &
  \multicolumn{1}{l|}{\textbf{$X,Y$}} &
  \textbf{} &
  \multicolumn{1}{l|}{\textbf{$A_{n,i},B_{n,i},C_{n,i},D_{n,i},X_{ii},Y_{ii}$}} &
  \multicolumn{1}{c|}{\textbf{}} &
  \multicolumn{1}{c|}{\textbf{\eqref{Eq:Pr:DTDissipativationUsingDOF1},\eqref{Eq:Pr:DTDissipativationUsingDOF2}}} &
  \textbf{\eqref{Eq:COVsAuxToDOF}} \\ \cline{2-13} 
 &
  \textbf{\color{red}$\H_2$-Norm} &
  \textbf{\ref{Pr:DTH2ControlUnderDOF}} &
  \textbf{$D=0,w\rightarrow z$} &
  \multicolumn{1}{l|}{\textbf{$A,E,G,J$}} &
  \multicolumn{1}{l|}{\textbf{$B,C,F,H$}} &
  \multicolumn{1}{l|}{\textbf{$A_n,B_n,C_n,D_n,Q$}} &
  \multicolumn{1}{l|}{\textbf{$X,Y$}} &
  \textbf{$\gamma$} &
  \multicolumn{1}{l|}{\textbf{$A_{n,i},B_{n,i},C_{n,i},D_{n,i},Q_iX_{ii},Y_{ii}$}} &
  \multicolumn{1}{c|}{\textbf{$\gamma_i$}} &
  \multicolumn{1}{c|}{\textbf{\eqref{Eq:Pr:DTH2ControlUnderDOF1}, \eqref{Eq:Pr:DTH2ControlUnderDOF2}}} &
  \textbf{\eqref{Eq:COVsAuxToDOF}} \\ \cline{2-13} 
 &
  \textbf{$H_\infty$-Norm} &
  \textbf{\ref{Pr:DTHInfControlUnderDOF}} &
  \textbf{$D=0,w\rightarrow z$} &
  \multicolumn{1}{l|}{\textbf{$A,E,G,J$}} &
  \multicolumn{1}{l|}{\textbf{$B,C,F,H$}} &
  \multicolumn{1}{l|}{\textbf{$A_n,B_n,C_n,D_n$}} &
  \multicolumn{1}{l|}{\textbf{$X,Y$}} &
  \textbf{$\gamma$} &
  \multicolumn{1}{l|}{\textbf{$A_{n,i},B_{n,i},C_{n,i},D_{n,i},X_{ii},Y_{ii}$}} &
  \multicolumn{1}{c|}{\textbf{$\gamma_i$}} &
  \multicolumn{1}{c|}{\textbf{\eqref{Eq:Pr:DTHInfControlUnderDOF1}, \eqref{Eq:Pr:DTHInfControlUnderDOF2}}} &
  \textbf{\eqref{Eq:COVsAuxToDOF}} \\ \hline
\end{tabular}%
}
\end{table}

\paragraph*{\textbf{Notation}} Double and single subscripts used in local LMI variables (Column 10) denote variable matrices as follows: (1) Double subscripts, e.g., $P_{ii}$ is the $i$\tsup{th} diagonal block of $P\triangleq \diag(P_{ii}:i\in\N_N)$; and (2) Single subscript, e.g., $Q_i$ is the collection of blocks $L_i\triangleq \{L_{ii}\}\cup\{L_{ij}:j\in\N_{i-1}\}\cup\{L_{ji}:j\in\N_{i-1}\}$ taken from $L \triangleq [L_{ij}]_{i,j\in\N_N}$. The acronym BEW (Column 12-13) stands for ``Block Element-Wise'' (see Lms. \ref{Lm:NetworkMatrixProperties} and \ref{Lm:ColumnandRowPermutations}).
\end{landscape}


\subsection{Summary of Results (Table \ref{Tab:CT-LTILocalResultsSummary})}

Recall that, in Sec. \ref{Sec:BasicsOfCTLTISystems}, we formulated different analysis and control synthesis tasks of interest for CT-LTI systems \eqref{Eq:CTLTISystem} as LMI problems in Props. \ref{Pr:CTLTIStability}-\ref{Pr:HInfControlUnderDOF}. As pointed out in Sec. \ref{Sec:CTNetworkedSystem}, each of these LMI problems are directly applicable to globally analyze and synthesize controllers for CTNSs \eqref{Eq:CTNSDynamics}. Moreover, we showed in Section \ref{Sec:DecentralizedAnalysis} that, under certain conditions, such global LMI (and LME) conditions that arise in network settings can be enforced in a decentralized, compositional and possibly distributed manner using Alg. \ref{Alg:DistributedPositiveDefiniteness} (and \ref{Alg:DistributedEquality}) based on the established Lm. \ref{Lm:MainLemmaShort} (and \ref{Lm:matrixEquality}). For  example, to decentrally analyze/enforce the global LMI condition $W>0$ via Alg. \ref{Alg:DistributedPositiveDefiniteness}, $W$ requires to be a network matrix corresponding to the considered CTNS (see Def. \ref{Def:NetworkMatrices} and Lm. \ref{Lm:NetworkMatrixProperties}). In all, under certain conditions, different analysis and control synthesis tasks of interest for CTNSs \eqref{Eq:CTNSDynamics} can be executed in a decentralized manner via solving specifically formulated local versions of the LMI problems in Props. \ref{Pr:CTLTIStability}-\ref{Pr:HInfControlUnderDOF} using Algs. \ref{Alg:DistributedPositiveDefiniteness} and \ref{Alg:DistributedEquality}.

In the interest of brevity, rather than explicitly stating all such local (decentralized) versions of the LMI problems in Props. \ref{Pr:CTLTIStability}-\ref{Pr:HInfControlUnderDOF}, we have summarized their key details in Tab. \ref{Tab:CT-LTILocalResultsSummary}. In particular, Tab. \ref{Tab:CT-LTILocalResultsSummary}-Col. 4 provides the preliminary conditions required of the CTNS to apply the interested proposition (given in Col. 2) while Cols. 5-9 provide the assumptions required regarding various CTNS parameters and LMI variables to decentralize the interested proposition. Note that, while all the CTNS parameters and LMI variables are, by definition, network matrices (except for the scalar LMI variables given in Tab. \ref{Tab:CT-LTILocalResultsSummary}-Col. 9), some of them may need to be block diagonal network matrices in order to decentralize the LMI problem in the interested proposition. 

For example, consider the stability analysis (i.e., Prop. \ref{Pr:CTLTIStability} and Tab. \ref{Tab:CT-LTILocalResultsSummary}-Row 1) which depends on finding a matrix $P>0$ such that $W\equiv -A^\T P -PA >0$. In this case, to decentralize the enforcement of the LMI $W>0$, $W$ needs to be a network matrix. As $A$, by definition, is a general (non-block diagonal) network matrix, based on Lm. \ref{Lm:NetworkMatrixProperties}, for $W$ to be a network matrix, $P$ needs to be a block diagonal network matrix.

Table \ref{Tab:CT-LTILocalResultsSummary} Cols. 10-13 summarize the details of the decentralized LMIs that needs to be solved locally. In particular, the variables involved in the local LMI problem solved at the subsystem $\Sigma_i,i\in\N_N$ are given in Tab. \ref{Tab:CT-LTILocalResultsSummary} Cols. 10-11. According to the used notation (see below Tabs. \ref{Tab:CT-LTILocalResultsSummary}-\ref{Tab:DT-LTILocalResultsSummary}), note that, such a local LMI problem may involve finding: (1) the $i$-th diagonal block of a block diagonal network matrix (denoted using double subscripts, e.g., $P_{ii}$), (2) a certain collection of blocks (matrices) of a general network matrix (denoted using single subscripts, e.g., $L_i\triangleq \{L_{ii}\} \cup \{L_{ij}:j\in\N_{i-1}\}\cup\{L_{ji}:j\in\N_{i-1}\}$). On the other hand, the local LMIs and LMEs that need to be enforced at subsystem $\Sigma_i, i\in\N_N$ are given in Tab. \ref{Tab:CT-LTILocalResultsSummary} Cols. 12-13. Note that, such local LMIs and LMEs are basically the BEW forms of global LMIs and LMEs (given in Props. \ref{Pr:CTLTIStability}-\ref{Pr:HInfControlUnderDOF}) enforced through decentralized Algs. \ref{Alg:DistributedPositiveDefiniteness} and \ref{Alg:DistributedEquality}.

For example, consider the FSF stabilization (i.e., Prop. \ref{Pr:StabilizationUnderFSF} and Tab. \ref{Tab:CT-LTILocalResultsSummary}-Row 7) which depends on finding a matrix $M>0$ and $L$ such that $W\equiv 
-MA^\T-AM-L^\T B^\T-BL>0$ (where FSF controller gain $K=LM^{-1}$). According to Lm. \ref{Lm:NetworkMatrixProperties}, when $A,L$ are considered as general network matrices, we need to constrain $M,B$ to be block diagonal network matrices so as to ensure $W$ is a network matrix (to enable decentralization). Upon decentralization, the local LMI variables at subsystem $\Sigma_i, i\in\N_N$ are $M_{ii}$ and $L_i\triangleq \{L_{ii}\} \cup \{L_{ij}:j\in\N_{i-1}\}\cup\{L_{ji}:j\in\N_{i-1}\}$, which respectively corresponds to the global LMI variables $M$ and $L$. The local LMI that needs to be enforced is found when enforcing 
$$
W \equiv [-M_{ii}A_{ji}^\T - A_{ij}M_{jj} -L_{ji}^\T B_{jj}^\T - B_{ii}L_{ij}]_{i,j\in\N_N}>0
$$ 
via Alg. \ref{Alg:DistributedPositiveDefiniteness} (i.e., its Step 14: $\tilde{W}_{ii}>0$). Similarly, the local LME that needs to be enforced is found when enforcing 
$$
V \equiv K-LM^{-1} = [K_{ij}-L_{ij}M_{jj}^{-1}]_{i,j\in\N_N} = 0
$$
via Alg. \ref{Alg:DistributedEquality} (i.e., its Step 9: $K_{ij}=L_{ij}M_{jj}^{-1}$, $K_{ji}=L_{ji}M_{ii}^{-1}, \forall j\in\N_{i-1}$, $K_{ii} = L_{ii}M_{ii}^{-1}$). Finally, $K_{ij},K_{ji}$ and $K_{ii}$ matrices determined above (collectively denoted as $K_i$) give the necessary local FSF controller gains \eqref{Eq:CTLocalFSFController}.

\subsection{Stability Related Decentralized Results}

Formally, the aforementioned decentralized stability analysis and FSF stabilization techniques can be summarized respectively as in the following two theorems (correspond to Props. \ref{Pr:CTLTIStability} and \ref{Pr:StabilizationUnderFSF}, and Tab. \ref{Tab:CT-LTILocalResultsSummary}-Rows 1 and 7).

\begin{theorem} \label{Th:CTLTIStability}(Stability analysis)
The CTNS \eqref{Eq:CTNSDynamics} under $u(t)=\0$ and $w(t)=\0$ is stable if at each subsystem $\Sigma_i,i\in\N_N$, the problem  
\begin{equation}\label{Eq:Th:CTLTIStability}
    \mathbb{P}_1: \ \ \text{Find} \ \ P_{ii}\ \ \text{such that} \ \ P_{ii} > 0, \ \ \tilde{W}_{ii}>0,
\end{equation}
is feasible, where $\tilde{W}_{ii}$ is computed from Alg. \ref{Alg:DistributedPositiveDefiniteness} (Steps: 3-16) when analyzing $W = [W_{ij}]_{i,j\in\N_N}>0$ with  
\begin{equation}\label{Eq:Th:CTLTIStability2}
    W_{ij} = -A_{ji}^\T P_{jj} - P_{ii}A_{ij}.
\end{equation}
\end{theorem}

\begin{proof}
Let us define $P \triangleq \diag(P_{ii}:i\in\N_N)$ and $W \triangleq -A^\T P - P A$. According to Prop. \ref{Pr:CTLTIStability}, we need to find $P>0$ such that $W>0$ to establish the global stability. Under the above definitions, it easy to see that $W_{ij} = -A_{ji}^\T P_{jj} - P_{ii}A_{ij}$, i.e., \eqref{Eq:Th:CTLTIStability2}. Note also that $W$ is a symmetric network matrix (see Def. \ref{Def:NetworkMatrices}). Therefore, $W>0$ can be tested in a decentralized manner by applying Alg. \ref{Alg:DistributedPositiveDefiniteness} and testing $\tilde{W}_{ii}>0$ at each subsystem $\Sigma_i,i\in\N_N$. Note also that $P_{ii}>0, \forall i\in\N_N \implies P>0$. Therefore, the existence of a matrix $P>0$ such that $W>0$ can be evaluated in a decentralized manner by solving $\mathbb{P}_1$ \eqref{Eq:Th:CTLTIStability} at each subsystem $\Sigma_i, i\in\N_N$. 
\end{proof}

\begin{theorem}\label{Th:StabilizationUnderFSF}
(FSF Stabilization) The CTNS \eqref{Eq:CTNSDynamics} (where $B$ is block diagonal) under $D=\0$, $w(t)=\0$ and local FSF control \eqref{Eq:CTLocalFSFController} is stable if at each subsystem $\Sigma_i,i\in\N_N$, the problem    
\begin{equation}\label{Eq:Th:StabilizationUnderFSF}
    \mathbb{P}_2: \ \ \text{Find} \ \ M_{ii},\ \ L_i\ \ \text{such that} \ \ M_{ii} > 0, \ \ \tilde{W}_{ii}>0, 
\end{equation}
is feasible, where $\tilde{W}_{ii}$ is computed from Alg. \ref{Alg:DistributedPositiveDefiniteness} (Steps: 3-16) when enforcing $W = [W_{ij}]_{i,j\in\N_N}>0$ with 
\begin{equation}\label{Eq:Th:StabilizationUnderFSF2}
    W_{ij} = -M_{ii}A_{ji}^\T - A_{ij}M_{jj} -L_{ji}^\T B_{jj}^\T - B_{ii}L_{ij}.
\end{equation}
The local FSF controller gains $K_i$ are computed using $M_{ii}$ and $L_i$, from Alg. \ref{Alg:DistributedEquality} (Steps 3-10) when enforcing $V = [V_{ij}]_{i,j\in\N_N}=0$ with 
\begin{equation}\label{Eq:Th:StabilizationUnderFSF3}
    V_{ij} = K_{ij}-L_{ij}M_{jj}^{-1}.
\end{equation}
\end{theorem}

\begin{proof}
Let us define $M \triangleq \diag(M_{ii}:i\in\N_N)$, $W \triangleq -MA^\T-AM-L^\T B^\T-BL$ and $V\triangleq K-LM^{-1}$. According to Prop. \ref{Pr:StabilizationUnderFSF}, we need to find $M>0$ and $L$ such that $W>0$ to ensure the closed loop stability, and then, the controller gains can be found using $K=LM^{-1}$. Under the above definitions, it easy to see that $W_{ij} = -M_{ii}A_{ji}^\T - A_{ij}M_{jj} -L_{ji}^\T B_{jj}^\T - B_{ii}L_{ij}$, i.e., \eqref{Eq:Th:StabilizationUnderFSF2} and $V_{ij}=K_{ij}-L_{ij}M_{jj}^{-1}$, i.e., \eqref{Eq:Th:StabilizationUnderFSF3}. Note also that $W$ is a symmetric network matrix (see Def. \ref{Def:NetworkMatrices}) while $V$ is of the form \eqref{Eq:LinearMatrixEqualityCondition} and satisfies the conditions stated in Lm. \ref{Lm:matrixEquality}. Consequently, $W>0$ can be enforced in a decentralized manner by applying Alg. \ref{Alg:DistributedPositiveDefiniteness} and enforcing $\tilde{W}_{ii}>0$ at each subsystem $\Sigma_i,i\in\N_N$. Since $M_{ii}>0, \forall i\in\N_N \implies M>0$, the existence of a matrix $M>0$ and $L$ such that $W>0$ can be evaluated in a decentralized manner by solving $\mathbb{P}_2$ \eqref{Eq:Th:StabilizationUnderFSF}. Subsequently, $V=0$ can be enforced in a decentralized manner by applying Alg. \ref{Alg:DistributedPositiveDefiniteness} to determine the local FSF controller gains $K_i$ at each subsystem $\Sigma_i,i\in\N_N$.
\end{proof}

Note that, the problems $\mathbb{P}_1$ \eqref{Eq:Th:CTLTIStability} and $\mathbb{P}_2$ \eqref{Eq:Th:StabilizationUnderFSF} are LMI problems due to the applicability of Lm. \ref{Lm:TwoByTwoBlockMatrixPDF} to simplify the matrix inequality $\tilde{W}_{ii}>0$ in Alg. \ref{Alg:DistributedPositiveDefiniteness} (Step 14). Therefore, such problems can be solved conveniently and efficiently using readily available LMI software toolboxes \cite{Boyd1994}.

To conclude the discussion on stability-based decentralized results, in what follows, we provide two more theorems, respectively, regarding decentralized stable observer design and DOF stabilization techniques (correspond to Props. \ref{Pr:Observer} and  \ref{Pr:StabilizationUnderDOF}, and Tab. \ref{Tab:CT-LTILocalResultsSummary}-Rows 11 and 15).

\begin{theorem}\label{Th:Observer}
(Stable observer design) For the CTNS \eqref{Eq:CTNSDynamics} (where $C,D$ are block diagonal) under $w(t)=\0$, 
the local Luenberger observers \eqref{Eq:CTLocalObserver} render the state estimation error dynamics \eqref{Eq:LuenbergerObserver} stable if at each subsystem $\Sigma_i,i\in\N_N$, the problem
\begin{equation}\label{Eq:Th:Observer}
    \mathbb{P}_3: \ \ \text{Find} \ \ P_{ii},\ \ K_i\ \ \text{such that} \ \ P_{ii} > 0, \ \ \tilde{W}_{ii}>0, 
\end{equation}
is feasible, where $\tilde{W}_{ii}$ is computed from Alg. \ref{Alg:DistributedPositiveDefiniteness} (Steps: 3-16) when enforcing $W = [W_{ij}]_{i,j\in\N_N}>0$ with 
\begin{equation}\label{Eq:Th:Observer2}
    W_{ij} = -A_{ji}^\T P_{jj} - P_{ii}A_{ij} + C_{ii}^\T K_{ji}^\T + K_{ij}C_{jj}.
\end{equation}
The local Luenberger observer parameters $L_i, \hat{A}_{i}$ and $\hat{B}_i$ \eqref{Eq:CTLocalObserver} are computed using $P_{ii}$ and $K_i$, from Alg. \ref{Alg:DistributedEquality} (Steps 3-10) when enforcing $V^{(k)} = [V_{ij}^{(k)}]_{i,j\in\N_N}=0$, for $k=1,2,3,$ where  
\begin{equation}\label{Eq:Th:Observer3}
\begin{aligned}
    V_{ij}^{(1)} =&\ L_{ij}-P_{ii}^{-1} K_{ij},\\
    V_{ij}^{(2)} =&\ \hat{A}_{ij} - A_{ij} + L_{ij}C_{jj},\\
    V_{ij}^{(3)} =&\ \hat{B}_{ij} - B_{ij} + L_{ij}D_{jj}.
\end{aligned}
\end{equation}
\end{theorem}
\begin{proof}
The proof starts by defining  $P \triangleq \diag(P_{ii}:i\in\N_N)$, $W \triangleq -A^\T P - PA + C^\T K + KC$, $V^{(1)} \triangleq L-P^{-1}K$, $V^{(2)} \triangleq \hat{A}-A+LC$, $V^{(3)} \triangleq \hat{B}-B+LD$ and proceeds using Prop. \ref{Pr:Observer} in a similar manner to the proof of Th. \ref{Th:StabilizationUnderFSF}. Therefore, here we omit providing explicit details. 
\end{proof}

\begin{theorem}\label{Th:StabilizationUnderDOF}
(DOF Stabilization) The CTNS \eqref{Eq:CTNSDynamics} (where $B,C$ are block diagonal) under $D=\0$, $w(t)=\0$ and local DOF control \eqref{Eq:CTLocalDOFController} is stable if at each subsystem $\Sigma_i,i\in\N_N$, the problem    
\begin{equation}\label{Eq:Th:StabilizationUnderDOF}
\begin{aligned}
    \mathbb{P}_4: \ \ \ \ \ \ \text{Find}& \ \ X_{ii},\ Y_{ii},\ A_{n,i},\ B_{n,i},\ C_{n,i},\ D_{n,i} \\ 
    \text{such that}& \ \ X_{ii} > 0,\ Y_{ii} > 0,\ \tilde{W}_{ii}^{(1)}>0,\ \tilde{W}_{ii}^{(2)}>0,  
\end{aligned}
\end{equation}
is feasible, where $\tilde{W}_{ii}^{(k)}$ is computed from Alg. \ref{Alg:DistributedPositiveDefiniteness} (Steps: 3-16) when enforcing $W^{(k)} = [W_{ij}^{(k)}]_{i,j\in\N_N}>0$, for $k=1,2,$ with 
\begin{equation}\label{Eq:Th:StabilizationUnderDOF2}
\begin{aligned}
    &W_{ij}^{(1)} = \bm{Y_{ii}e_{ij} & e_{ij} \\ e_{ij} & X_{ij}e_{ij}},\ \ \ \ 
    W_{ij}^{(2)}=\\
    &\bm{-\H_s(A_{ij}Y_{jj}+B_{ii}C_{n,ij}) & -A_{ij} - B_{ii}D_{n,ij}C_{jj} - A_{n,ji}^\T \\ 
    \star & -\H_s(X_{ii}A_{ij} + B_{n,ij}C_{jj})}. 
\end{aligned}
\end{equation}
The local DOF controller parameters $A_{c,i},B_{c,i},C_{c,i}$ and $D_{c,i}$ are computed using $X_{ii},Y_{ii},A_{n,i},B_{n,i},C_{n,i}$ and $D_{n,i}$, from Alg. \ref{Alg:DistributedEquality} (Steps 3-10) when enforcing $V^{(k)} = [V_{ij}^{(k)}]_{i,j\in\N_N}=0$, for $k=1,2,3,4$ with 
\begin{equation}\label{Eq:Th:StabilizationUnderDOF3}
\begin{aligned}
    V_{ij}^{(1)} =&\ D_{c,ij}-D_{n,ij},\\
    V_{ij}^{(2)} =&\ C_{c,ij}-(C_{n,ij}-D_{n,ij}C_{jj}Y_{jj})N^{-\T}_{jj},\\
    V_{ij}^{(3)} =&\ B_{c,ij}-M^{-1}_{ii}(B_{n,ij}-X_{ii}B_{ii}D_{n,ij}),\\
    V_{ij}^{(4)} =&\ A_{c,ij}-M^{-1}_{ii}(A_{n,ij}-B_{n,ij}C_{jj}Y_{jj}-X_{ii}B_{ii}C_{n,ij}\\
    &\ -X_{ii}(A_{ij}-B_{ii}D_{n,ij}C_{jj})Y_{jj})N^{-\T}_{jj},
\end{aligned}
\end{equation}
where $M_{ii}$ and $N_{ii}, i\in\N_N$ are any two matrices that satisfy 
\begin{equation}\label{Eq:Th:StabilizationUnderDOF4}
    X_{ii}Y_{ii} + M_{ii}N_{ii}^\T = \I.
\end{equation}
\end{theorem}

\begin{proof}
The proof starts by defining (inspired by \eqref{Eq:Pr:StabilizationUnderDOF1}, \eqref{Eq:Pr:StabilizationUnderDOF2}, \eqref{Eq:COVsAuxToDOF} and \eqref{Eq:CoVsXYToMN}): 
\begin{equation*}
    \begin{aligned}
        &X \triangleq \diag(X_{ii}:i\in\N_N),\ \ Y \triangleq \diag(Y_{ii}:i\in\N_N),\\
        &W^{(1)} \triangleq \text{BEW}\left(\bm{Y & \I \\ \I & X}\right), \ \ W^{(2)} \triangleq \\
        &\text{BEW}\left(\bm{-\H_s(AY+BC_n) & -A - BD_nC - A_n^\T \\ \star & -\H_s(XA + B_nC)}\right),\\
        &V^{(1)} \triangleq D_c - D_n, \ \ 
        V^{(2)} \triangleq C_c - (C_n-D_nCY)N^{-\T},\\ 
        &V^{(3)} \triangleq B_c - M^{-1}(B_n-XBD_n),\ \ V^{(4)} \triangleq \\ 
        &A_c - M^{-1}(A_n-B_nCY-XBC_n-X(A-BD_nC)Y)N^{-\T},\\ 
        &M \triangleq \diag(M_{ii}:i\in\N_N),\ \ N \triangleq \diag(Y_{ii}:i\in\N_N).
    \end{aligned}
\end{equation*}
and proceeds using Prop. \ref{Pr:StabilizationUnderDOF} in a similar manner to the proof of Th. \ref{Th:StabilizationUnderFSF} (except for the use of Lm. \ref{Lm:ColumnandRowPermutations}: $W=\text{BEW}(\Psi)>0 \iff \Psi>0$). Hence, here we omit providing explicit details. 
\end{proof}

\begin{figure*}[!b]
    \centering
    \hrulefill
    \begin{equation}\label{Eq:Th:DissipativationUnderFSF2}
    W_{ij} = \bm{
    -\H_s(A_{ij}M_{jj}+B_{ii}L_{ij}) & -E_{ij}+M_{ii}C_{ii}^\T S_{ij} & M_{ii}C^\T_{ii}e_{ij} \\ 
    \star & \H_e(F_{ii}^\T S_{ij})+R_{ij} & F_{ii}^\T e_{ij} \\
    \star & \star & -Q_{ii}^{-1}e_{ij} 
    }
    \end{equation}
    \begin{equation} \label{Eq:Th:DissipativeObserver2}
    W_{ij} = 
    \bm{
    -\H_s(P_{ii}A_{ij}-K_{ij}C_{jj}) & -P_{ii}E_{ij}+K_{ij}F_{jj}+G_{ii}^\T S_{ij} & G_{ii}^\T e_{ij} \\ 
    \star & \H_s(J_{ii}^\T S_{ij})+R_{ij} & J_{ii}^\T e_{ij} \\ \star & \star & -Q_{ii}^{-1} e_{ij}
    }
    \end{equation}
    \begin{equation}\label{Eq:Th:DissipativationUsingDOF2}
    \scriptsize
    W_{ij}^{(2)} = 
    \bm{
    -\H_s(A_{ij}Y_{jj}+B_{ii}C_{n,ij}) & -A_{ij}-B_{ii}D_{n,ij}C_{jj}-A_{n,ji}^\T & -E_{ij}-B_{ii}D_{n,ij}F_{jj}+(Y_{ii}G_{ji}^\T+C_{n,ji}^\T H_{jj}^\T)S_{jj} & Y_{ii}G_{ji}^\T+C_{n,ji}^\T H_{jj}^\T\\
    \star & -\H_s(X_{ii}A_{ij}+B_{n,ij}C_{jj}) & -X_{ii}E_{ij}-B_{n,ij}F_{jj}+(G_{ji}^\T+C_{ii}^\T D_{n,ji}^\T H_{jj}^\T)S_{jj} & G_{ji}^\T+C_{ii}^\T D_{n,ji}^\T H_{jj}^\T\\
    \star & \star & \H_s((J_{ji}^\T+F_{ii}^\T D_{n,ji}^\T H_{jj}^\T)S_{jj})+R_{ij} & J_{ji}^\T+F_{ii}^\T D_{n,ji}^\T H_{jj}^\T\\
    \star & \star & \star & -Q_{ii}^{-1}e_{ij}
    }
    \end{equation}
    \begin{equation}\label{Eq:Th:DissipativationUsingDOFStep1}
    \scriptsize
        W^{(2)} \triangleq \text{BEW}\Big( 
    \bm{
    -\H_s(AY+BC_n) & -A-BD_nC-A_n^\T & -E-BD_nF+(YG^\T+C_n^\T H^\T)S & YG^\T+C_n^\T H^\T\\
    \star & -\H_s(XA+B_nC) & -XE-B_nF+(G^\T+C^\T D_n^\T H^\T)S & G^\T+C^\T D_n^\T H^\T\\
    \star & \star & \H_s((J^\T+F^\T D_n^\T H^\T)S)+R & J^\T+F^\T D_n^\T H^\T\\
    \star & \star & \star & -Q^{-1}
    }\Big)
    \end{equation}
\end{figure*}

\subsection{Dissipativity Related Decentralized Results}

In this subsection, analogous to Theorems \ref{Th:CTLTIStability}-\ref{Th:StabilizationUnderDOF}, we provide $(Q,S,R)$-dissipativity based results on:
\begin{enumerate}
    \item decentralized $(Q,S,R)$-dissipativity analysis,
    \item decentralized FSF $(Q,S,R)$-dissipativation,
    \item decentralized $(Q,S,R)$-dissipative observer design and 
    \item decentralized DOF $(Q,S,R)$-dissipativation,
\end{enumerate}
that respectively correspond to Props. \ref{Pr:CTLTIQSRDissipativity}, \ref{Pr:DissipativationUnderFSF}, \ref{Pr:DissipativeObserver} and \ref{Pr:DissipativationUsingDOF} (Tab. \ref{Tab:CT-LTILocalResultsSummary}-Rows 2, 8, 12 and 16). In what follows, regarding the given specification matrices $Q,S,R$, it is assumed that: (1) $Q$ is a block diagonal network matrix, (2) $-Q>0$, and (3) $R=R^\T$.

\begin{theorem} \label{Th:CTLTIQSRDissipativity}(Dissipativity analysis)
The CTNS \eqref{Eq:CTNSDynamics} (where $C,D$ are block diagonal) under $w(t)=\0$ is $(Q,S,R)$-dissipative from $u(t)$ to $y(t)$ if at each subsystem $\Sigma_i,i\in\N_N$, the problem  
\begin{equation}\label{Eq:Th:CTLTIQSRDissipativity}
    \mathbb{P}_5: \ \ \text{Find} \ \ P_{ii}\ \ \text{such that} \ \ P_{ii} > 0, \ \ \tilde{W}_{ii}>0,
\end{equation}
is feasible, where $\tilde{W}_{ii}$ is computed from Alg. \ref{Alg:DistributedPositiveDefiniteness} (Steps: 3-16) when analyzing $W = [W_{ij}]_{i,j\in\N_N}>0$ with  
\begin{equation}\label{Eq:Th:CTLTIQSRDissipativity2}
    W_{ij} = 
    \bm{
    -\H_e(P_{ii} A_{ij}) &  -P_{ii}B_{ij} + C_{ii}^\T S_{ij}    & C_{ii}^\T e_{ij} \\
\star          & \H_e(D_{ii}^\T S_{ij}) + R_{ij}    & D_{ii}^\T e_{ij} \\ 
\star & \star & -Q_{ii}^{-1}e_{ij}
    }.
\end{equation}
\end{theorem}
\begin{proof}
The proof starts by defining $P \triangleq \diag(P_{ii}:i\in\N_N),$ 
$$
W \triangleq \text{BEW}\big(\bm{
-A^\T P - P A     &  -PB + C^\T S            & C^\T \\
\star          & D^\T S + S^\T D + R    & D^\T \\ 
\star & \star & -Q^{-1}
}\big)
$$
(inspired by \eqref{Eq:Pr:CTLTIQSRDissipativity}), and proceeds using Prop. \ref{Pr:CTLTIQSRDissipativity} in a similar manner to the proof of Th. \ref{Th:StabilizationUnderFSF} (except for the use of Lm. \ref{Lm:ColumnandRowPermutations}: $W=\text{BEW}(\Psi)>0 \iff \Psi>0$). Therefore, here we omit providing explicit details. 
\end{proof}

\begin{theorem}\label{Th:DissipativationUnderFSF}
(FSF Dissipativation) The CTNS \eqref{Eq:CTNSDynamics} (where $B,C,F$ are block diagonal) under $D=\0$ and local FSF control \eqref{Eq:CTLocalFSFController} is $(Q,S,R)$-dissipative from $w(t)$ to $y(t)$ if at each subsystem $\Sigma_i,i\in\N_N$, the problem    
\begin{equation}\label{Eq:Th:DissipativationUnderFSF}
    \mathbb{P}_6: \ \ \text{Find} \ \ M_{ii},\ \ L_i\ \ \text{such that} \ \ M_{ii} > 0, \ \ \tilde{W}_{ii}>0, 
\end{equation}
is feasible, where $\tilde{W}_{ii}$ is computed from Alg. \ref{Alg:DistributedPositiveDefiniteness} (Steps: 3-16) when enforcing $W = [W_{ij}]_{i,j\in\N_N}>0$ with $W_{ij}$ given in \eqref{Eq:Th:DissipativationUnderFSF2}. 
The local FSF controller gains $K_i$ are computed using $M_{ii}$ and $L_i$ identically to Th. \ref{Th:StabilizationUnderFSF} via \eqref{Eq:Th:StabilizationUnderFSF3}.
\end{theorem}

\begin{proof}
The proof starts by defining $P \triangleq \diag(P_{ii}:i\in\N_N)$, 
$$
W \triangleq \text{BEW}
\big(\bm{
-\H_s(AM+BL) & -E+MC^\T S & MC^\T \\ 
\star & \H_s(F^\T S)+R & F^\T \\
\star & \star & -Q^{-1} 
}\big)
$$
(inspired by \eqref{Eq:Pr:DissipativationUnderFSF}), and proceeds using Prop. \ref{Pr:DissipativationUnderFSF} and Lm. \ref{Lm:ColumnandRowPermutations} in a similar manner to the proof of Th. \ref{Th:CTLTIQSRDissipativity}. Therefore, here we omit providing explicit details. 
\end{proof}

\begin{theorem}\label{Th:DissipativeObserver}
(Dissipative observer design) For the CTNS \eqref{Eq:CTNSDynamics} (where $C,D,F$ are block diagonal), the local Luenberger observers \eqref{Eq:CTLocalObserver} with the local performance metrics \eqref{Eq:CTLocalObserverPerf} (such that $G,J$ in \eqref{Eq:CTGlobalObserverPerf} are block diagonal) render the global state estimation error dynamics \eqref{Eq:LuenbergerObserverWithPerformance} $(Q,S,R)$-dissipative from $w(t)$ to $z(t)$ if at each subsystem $\Sigma_i,i\in\N_N$, the problem
\begin{equation}\label{Eq:Th:DissipativeObserver}
    \mathbb{P}_7: \ \ \text{Find} \ \ P_{ii},\ \ K_i\ \ \text{such that} \ \ P_{ii} > 0, \ \ \tilde{W}_{ii}>0, 
\end{equation}
is feasible, where $\tilde{W}_{ii}$ is computed from Alg. \ref{Alg:DistributedPositiveDefiniteness} (Steps: 3-16) when enforcing $W = [W_{ij}]_{i,j\in\N_N}>0$ with $W_{ij}$ given in \eqref{Eq:Th:DissipativeObserver2}. The local Luenberger observer parameters $L_i, \hat{A}_{i}$ and $\hat{B}_i$ \eqref{Eq:CTLocalObserver} are computed using $P_{ii}$ and $K_i$ identically to Th. \ref{Th:Observer} via \eqref{Eq:Th:Observer3}.
\end{theorem}
\begin{proof}
The proof starts by defining  $P \triangleq \diag(P_{ii}:i\in\N_N)$, $W \triangleq 
\text{BEW}(\Psi)$, where (inspired by \eqref{Eq:Pr:DissipativeObserver}):\\ 
$$
\Psi \triangleq \bm{
-\H_s(PA-KC) & -PE+KF+G^\T S & G^\T \\ \star & \H_s(J^\T S)+R & J^\T \\ \star & \star & -Q^{-1}
},$$ 
and proceeds using Prop. \ref{Pr:DissipativeObserver} and Lm. \ref{Lm:ColumnandRowPermutations} in a similar manner to the proof of Th. \ref{Th:CTLTIQSRDissipativity}. Therefore, here we omit providing explicit details.
\end{proof}

\begin{theorem}\label{Th:DissipativationUsingDOF}
(DOF Dissipativation) The CTNS \eqref{Eq:CTNSDynamics} (where $B,C,F$ are block diagonal) under $D=\0$, local DOF control \eqref{Eq:CTLocalDOFController} and local performance metrics \eqref{Eq:CTLocalDOFControllerPerf} (such that $H$ in \eqref{Eq:CTGlobalDOFControllerPerf} is block diagonal), i.e., \eqref{Eq:CTLTIUnderDOF}, is $(Q,S,R)$-dissipative (where $S$ is block diagonal) from $w(t)$ to $z(t)$ if at each subsystem $\Sigma_i,i\in\N_N$, the problem    
\begin{equation}\label{Eq:Th:DissipativationUsingDOF}
\begin{aligned}
    \mathbb{P}_8: \ \ \ \ \ \ \text{Find}& \ \ X_{ii},\ Y_{ii},\ A_{n,i},\ B_{n,i},\ C_{n,i},\ D_{n,i} \\ 
    \text{such that}& \ \ X_{ii} > 0,\ Y_{ii} > 0,\ \tilde{W}_{ii}^{(1)}>0,\ \tilde{W}_{ii}^{(2)}>0,  
\end{aligned}
\end{equation}
is feasible, where $\tilde{W}_{ii}^{(k)}$ is computed from Alg. \ref{Alg:DistributedPositiveDefiniteness} (Steps: 3-16) when enforcing $W^{(k)} = [W_{ij}^{(k)}]_{i,j\in\N_N}>0$ for $k=1,2,$ with 
$\scriptsize W_{ij}^{(1)} = \bm{Y_{ii}e_{ij} & e_{ij} \\ e_{ij} & X_{ij}e_{ij}}$ and $W_{ij}^{(2)}$ given in  \eqref{Eq:Th:DissipativationUsingDOF2}. The local DOF controller parameters $A_{c,i},B_{c,i},C_{c,i}$ and $D_{c,i}$ are computed using $X_{ii},Y_{ii},A_{n,i},B_{n,i},C_{n,i}$ and $D_{n,i}$, identically to Th. \ref{Th:StabilizationUnderDOF} via \eqref{Eq:Th:StabilizationUnderDOF3} and \eqref{Eq:Th:StabilizationUnderDOF4}.
\end{theorem}

\begin{proof}
The proof starts by defining (inspired by \eqref{Eq:Pr:DissipativationUsingDOF1} and \eqref{Eq:Pr:DissipativationUsingDOF2}): 
$X \triangleq \diag(X_{ii}:i\in\N_N)$, $Y \triangleq \diag(Y_{ii}:i\in\N_N)$, 
$\scriptsize W^{(1)} \triangleq \text{BEW}\left(\bm{Y & \I \\ \I & X}\right)$, and $W^{(2)}$ as in \eqref{Eq:Th:DissipativationUsingDOFStep1}, 
and proceeds using Prop. \ref{Pr:DissipativationUsingDOF} and Lm. \ref{Lm:ColumnandRowPermutations} in a similar manner to the proof of Th. \ref{Th:StabilizationUnderDOF}. Hence, here we omit providing explicit details. 
\end{proof}

To conclude this section, we re-emphasizing that local problems $\mathbb{P}_1$-$\mathbb{P}_8$ stated respectively in the theorems \ref{Th:CTLTIStability}-\ref{Th:DissipativationUsingDOF} are LMI problems due to the applicability of Lm. \ref{Lm:TwoByTwoBlockMatrixPDF} to simplify the matrix inequality $\tilde{W}_{ii}>0$ in Alg. \ref{Alg:DistributedPositiveDefiniteness} (Step 14). Consequently, such problems can be solved conveniently and efficiently using readily available LMI software toolboxes \cite{Boyd1994}. Moreover, based on the remaining propositions given in Sec. \ref{Sec:BasicsOfCTLTISystems}, a similar set of theorems can be proposed to provide respective decentralized techniques (as summarized in the respective remaining rows in Tab. \ref{Tab:CT-LTILocalResultsSummary}).


\section{Preliminaries: Discrete-Time Linear Time Invariant (DT-LTI) Systems}
\label{Sec:BasicsOfDTLTISystems}

In parallel to Sec. \ref{Sec:BasicsOfCTLTISystems}, in this section, we provide a comprehensive collection of linear matrix inequality (LMI) conditions that arise when analyzing or synthesizing controllers for discrete-time linear time-invariant (DT-LTI) systems.

Consider the DT-LTI system given by 
\begin{equation}\label{Eq:DTLTISystem}
\begin{aligned}
    x(t+1) = Ax(t) + Bu(t),\\
    y(t) = Cx(t) + Du(t),
\end{aligned}
\end{equation}
where $x(t)\in\R^n, u(t)\in\R^p, y(t)\in\R^m$ and $t\in\N$.

\subsection{Analysis of DT-LTI Systems}

\subsubsection{\textbf{Stability}}
A necessary and sufficient LMI condition for the \emph{stability} of \eqref{Eq:DTLTISystem} is given in the following proposition.

\begin{proposition} \label{Pr:DTLTIStability}
The DT-LTI system \eqref{Eq:DTLTISystem} under $u(t)=\0$ is stable iff $\exists P>0$ such that 
\begin{equation}\label{Eq:Pr:DTLTIStability}
    \bm{P & A^\T P \\ PA & P} > 0.
\end{equation}
\end{proposition}
\begin{proof}
It is well-known (e.g., see \cite{Antsaklis2006}) that \eqref{Eq:DTLTISystem} under $u(t)=\0$ is stable iff $\exists P>0$ such that 
\begin{equation}\label{Eq:Pr:DTLTIStabilityStep1}
    P-A^\T P A > 0.
\end{equation}
Therefore, we only need to prove the equivalence of \eqref{Eq:Pr:DTLTIStability} and \eqref{Eq:Pr:DTLTIStabilityStep1}. This is achieved by applying Lm. \ref{Lm:TwoByTwoBlockMatrixPDF} and using Lm. \ref{Lm:PreAndPostMultiplication} with $\diag(I,P)$ on \eqref{Eq:Pr:DTLTIStabilityStep1} as: 
\begin{align*}
P-A^\T P A > 0 \iff \bm{P & A^\T \\ A & P^{-1}}>0 \iff \bm{P & A^\T P \\ PA & P}.
\end{align*}
\end{proof}

\begin{remark}\label{Rm:DTCaseAlternativeLMIs}
When $A$ is a network matrix (and $P$ is a block diagonal network matrix), the matrix term $A^\T P A$ in \eqref{Eq:Pr:DTLTIStabilityStep1} will not be a network matrix (as opposed to matrix terms $A^\T P$ and $PA$ in \eqref{Eq:Pr:DTLTIStability}). Therefore, \eqref{Eq:Pr:DTLTIStabilityStep1} cannot be implemented in a decentralized manner for analysis and control synthesis tasks via Alg. \ref{Alg:DistributedPositiveDefiniteness}. This is the motivation behind establishing a slightly different stability condition in \eqref{Eq:Pr:DTLTIStability} (as opposed to \eqref{Eq:Pr:DTLTIStabilityStep1}).  
\end{remark}

\subsubsection{\textbf{$(Q,S,R)$-dissipativity}} 
For DT-LTI systems \eqref{Eq:DTLTISystem}, the \emph{$(Q,S,R)$-dissipativity} \cite{Kottenstette2014} property is defined as follows.

\begin{definition} \cite{Kottenstette2014} \label{Def:DiscreteQSRDissipativity}
The DT-LTI system \eqref{Eq:DTLTISystem} is $(Q,S,R)$-dissipative (from $u(t)$ to $y(t)$), if there exists a positive definite function $V(x):\R^n\rightarrow\R_{\geq0}$ called the storage function such that for all $t_1,t_0\in\N, t_1 \geq t_0 \geq 0, x(t_0) \in \R^n$ and $u(t)\in\R^p$, the inequality
\begin{equation*}
    V(x(t_1))-V(x(t_0)) \leq \sum_{t=t_0}^{t_1-1} 
    \begin{bmatrix}
    y(t)\\u(t)
    \end{bmatrix}^\T
    \begin{bmatrix}
    Q & S \\ S^\T & R
    \end{bmatrix}
    \begin{bmatrix}
    y(t)\\u(t)
    \end{bmatrix}
\end{equation*}
holds, where $Q\in\R^{m \times m}, S\in \R^{m \times p}, R\in\R^{p\times p}$ are given.
\end{definition}

At this point, it is worth noting that Remark \ref{Rm:QSRDissipativityVariants} is equally applicable for DT-LTI systems \eqref{Eq:DTLTISystem}. In parallel to Prop. \ref{Pr:DTLTIStability}, a necessary and sufficient condition for the $(Q,S,R)$-dissipativity of \eqref{Eq:DTLTISystem} is given in the following proposition.

\begin{proposition} \label{Pr:DTLTIQSRDissipativity}
The DT-LTI system \eqref{Eq:DTLTISystem} is $(Q,S,R)$-dissipative (with $-Q>0$,\ $R=R^\T$) from $u(t)$ to $y(t)$ iff $\exists P>0$ such that 
\begin{equation}\label{Eq:Pr:DTLTIQSRDissipativity}
\bm{
P & C^\T S & A^\T P & C^\T \\ 
\star & D^\T S + S^\T D + R & B^\T P & D^\T \\ 
\star & \star & P & \0 \\
\star & \star & \star & -Q^{-1}
}>0.
\end{equation}
\end{proposition}
\begin{proof}
It is well-known (e.g., see \cite{Kottenstette2014}) that \eqref{Eq:DTLTISystem} is $(Q,S,R)$-dissipative from $u(t)$ to $y(t)$ iff $\exists P>0$ such that 
\begin{equation}\label{Eq:Pr:DTLTIQSRDissipativityStep1}
\bm{
P -A^\T P A + \hat{Q} & -A^\T PB + \hat{S}\\
\star         & -B^\T P B + \hat{R}
} \geq 0, 
\end{equation}
where 
$ 
\hat{Q} = C^\T Q C,\ \ 
\hat{S} = C^\T S + C^\T Q D,\ \ 
\hat{R} = D^\T Q D + (D^\T S + S^\T D) + R
$. 
Therefore, we only need to prove the equivalence of \eqref{Eq:Pr:DTLTIQSRDissipativity} and \eqref{Eq:Pr:DTLTIQSRDissipativityStep1} under $-Q>0$ and $R=R^\T$. 

This is achieved by following the steps: (1) applying Lm. \ref{Lm:TwoByTwoBlockMatrixPDF} to isolate $P$, (2) using Lm. \ref{Lm:PreAndPostMultiplication} with $\diag(I,I,P)$ and (3) applying Lm. \ref{Lm:TwoByTwoBlockMatrixPDF} to isolate $Q$, on \eqref{Eq:Pr:DTLTIQSRDissipativityStep1}, to respectively obtain:
\begin{align*}
&\bm{
P + C^\T Q C & C^\T S + C^\T Q D & A^\T \\
\star         & D^\T Q D + \H_s(D^\T S) + R & B^\T \\
\star & \star & P^{-1}
} > 0  
\\ 
&\iff \bm{
P + C^\T Q C & C^\T S + C^\T Q D & A^\T P \\
\star         & D^\T Q D + \H_s(D^\T S) + R & B^\T P\\
\star & \star & P
} > 0
\\ 
&\iff \bm{
P & C^\T S & A^\T P & C^\T \\
\star & \H_s(D^\T S) + R & B^\T P & D^\T\\
\star & \star & P & \0 \\
\star & \star & \star & -Q^{-1}
} > 0 \iff \mbox{\eqref{Eq:Pr:DTLTIQSRDissipativity}}.
\end{align*}
\end{proof}

As mentioned in Rm. \ref{Rm:DTCaseAlternativeLMIs}, we remind that the motivation behind establishing a slightly different $(Q,S,R)$-dissipativity condition in \eqref{Eq:Pr:DTLTIQSRDissipativity} (as opposed to \eqref{Eq:Pr:DTLTIQSRDissipativityStep1}) is to ensure the decentralized analysis and enforcement of \eqref{Eq:Pr:DTLTIQSRDissipativity} via Alg. \ref{Alg:DistributedPositiveDefiniteness}.

\subsubsection{\textbf{$\H_2$-Norm}} 
Let $\mathcal{G}:\mathcal{L}_{2e} \rightarrow \mathcal{L}_{2e}$ be the transfer matrix of the DT-LTI system \eqref{Eq:DTLTISystem} from $u(t)$ to $y(t)$. If $A$ in \eqref{Eq:DTLTISystem} is Schur, $\mathcal{G}(z)=C(z\I-A)^{-1}B+D$. 
As shown in \cite{Caverly2019}, the $\H_2$-norm of $\mathcal{G}$ is 
\begin{equation*}
    \Vert \mathcal{G} \Vert_{\H_2}^2 = \tr{B^\T M B^\T + D^\T D}  = \tr{C N C^\T + DD^\T}, 
\end{equation*}
where $M,N>0$ with $A^\T MA - M + C^\T C = 0$ and $ANA^\T - N + BB^\T=0$. Consider the following proposition.

\begin{proposition} \hspace{-2mm} \cite[pp. 64]{Caverly2019}\label{Pr:DTH2Norm}
The $\H_2$-Norm of the transfer matrix of the DT-LTI system \eqref{Eq:DTLTISystem} (i.e., $\Vert \mathcal{G} \Vert_{\H_2}$) is $\Vert \mathcal{G} \Vert_{\H_2} < \gamma$  iff  $\exists P,Q>0$ and $\gamma\in\R_{>0}$ such that 
\begin{equation}\label{Eq:Pr:DTH2Norm1}
\bm{P & A P & B \\ \star & P & \0 \\ \star & \star & \I}>0,\ 
\bm{Q & CP & D \\ \star & P & \0 \\ \star & \star & \I}>0,\ 
\tr{Q}< \gamma^2,
\end{equation}
or
\begin{equation}\label{Eq:Pr:DTH2Norm2}
\bm{P & A P & C^\T \\ \star & P & \0 \\ \star & \star & \I}>0,\ 
\bm{Q & B^\T P & D^\T \\ \star & P & \0 \\ \star & \star & \I}>0,\ 
\tr{Q}<\gamma^2.
\end{equation}
\end{proposition}

\subsubsection{\textbf{$\H_\infty$-Norm}}
The $\H_\infty$-norm of $\mathcal{G}$ is \cite[pp.49]{Caverly2019}:
\begin{equation}
    \Vert \mathcal{G} \Vert_{\H_\infty}^2 
    =\sup_{u\in\mathcal{L}_2,u\neq0} \frac{\Vert y \Vert_{\mathcal{L}_2}^2}{\Vert u \Vert_{\mathcal{L}_2}^2}.
\end{equation}
Now, consider the following proposition.

\begin{proposition}\cite[pp.50]{Caverly2019}\label{Pr:DTHInfNorm}
The $\H_\infty$-Norm of the transfer matrix of the DT-LTI system \eqref{Eq:DTLTISystem} (i.e., $\Vert \mathcal{G} \Vert_{\H_\infty}$) is $\Vert \mathcal{G} \Vert_{\H_\infty} < \gamma$  iff  $\exists P>0$ and $\gamma\in\R_{>0}$ such that 
\begin{equation}\label{Eq:Pr:DTHInfNorm1}
\bm{
P & AP & B & \0 \\ \star & P & \0 & PC^\T \\ 
\star  & \star & \gamma \I & D^\T \\ \star & \star & \star & \gamma \I
}>0,
\end{equation}
or
\begin{equation}\label{Eq:Pr:DTHInfNorm2}
\bm{
P & PA & PB & \0 \\ \star & P & \0 & C^\T \\ 
\star  & \star & \gamma \I & D^\T \\ \star & \star & \star & \gamma \I
}>0.
\end{equation}
\end{proposition}

\subsubsection{\textbf{Controllability}}
Regarding the \emph{controllability} \cite{Antsaklis2006} of the DT-LTI system \eqref{Eq:DTLTISystem}, consider the following proposition. 
\begin{proposition}\cite[pp. 86]{Caverly2019}\label{Pr:DTStabilizability}
A DT-LTI system \eqref{Eq:DTLTISystem} is: (1) stable and controllable iff $\exists P>0$ such that 
\begin{equation}
    P-APA^\T-BB^\T = 0,
\end{equation}
and (2) stabilizable iff $\exists P>$ such that
\begin{equation}\label{Eq:Pr:DTStabilizability}
    \bm{P & PA^\T \\ \star & P+BB^\T} > 0 
\end{equation}
(with $K=-(2\I+B^\T P^{-1} B)^{-1}B^\T P^{-1} A$, $A+BK$ is Schur).
\end{proposition}

\subsubsection{\textbf{Observability}}
Regarding the \emph{observability} \cite{Antsaklis2006} of the DT-LTI system \eqref{Eq:DTLTISystem}, consider the following proposition. 
\begin{proposition}\cite[pp. 87]{Caverly2019}\label{Pr:DTDetectability}
A DT-LTI system \eqref{Eq:DTLTISystem} is: (1) stable and observable iff $\exists P>0$ such that 
\begin{equation}
    P - APA^\T - C^\T C = 0,
\end{equation}
and (2) detectable iff $\exists P>$ such that
\begin{equation}\label{Eq:Pr:DTDetectability}
    \bm{P & PA \\ \star & P+C^\T C} > 0 
\end{equation}
(under $L=-AP^{-1}C^\T(2\I+CP^{-1}C^\T)^{-1}$, $A+LC$ is Schur).
\end{proposition}

\subsection{Full-State Feedback (FSF) Controller Synthesis for DT-LTI Systems}

Consider the DT-LTI system  \eqref{Eq:DTLTISystem} with noise $w(t)\in\R^q$:
\begin{equation}\label{Eq:NoisyDTLTISystem}
\begin{aligned}
    x(t+1) = Ax(t) + Bu(t) + Ew(t),\\
    y(t) = Cx(t) + Du(t) + Fw(t).
\end{aligned}
\end{equation}
Under full-state feedback (FSF) control $u(t)=Kx(t)$, the closed-loop DT-LTI system takes the form
\begin{equation}\label{Eq:DTLTIUnderFSF}
    \begin{aligned}
        x(t+1) = (A+BK)x(t)+Ew(t),\\
        y(t) = (C+DK)x(t)+Fw(t).
    \end{aligned}
\end{equation}

In parallel to Sec. \ref{SubSec:FSFControllerSynthesis}, in the subsequent subsections, we provide LMI conditions for FSF controller synthesis so as to stabilize, $(Q,S,R)$-dissipativate or optimize $\H_2$/$\H_\infty$-norm of the closed-loop system \eqref{Eq:DTLTIUnderFSF}.

\subsubsection{\textbf{Stabilization}}
The following proposition gives an LMI condition that leads to synthesizing a FSF controller $K$ such that the closed-loop system  \eqref{Eq:DTLTIUnderFSF} is stabilized.

\begin{proposition}\label{Pr:DTStabilizationUnderFSF}
Under $D=w(t)=\0$, the closed-loop DT-LTI system \eqref{Eq:DTLTIUnderFSF} is stable iff $\exists M>0$ and $L$ such that 
\begin{equation}\label{Eq:Pr:DTStabilizationUnderFSF}
    \bm{M & MA^\T+L^\T B^\T \\ \star & M}>0
\end{equation}
and $K=LM^{-1}$.
\end{proposition}
\begin{proof}
The proof is complete by following the steps: (1) apply Prop. \ref{Pr:DTLTIStability} for \eqref{Eq:DTLTIUnderFSF} to get an LMI in $P$ and $K$, (2) transform the obtained LMI using Lm. \ref{Lm:PreAndPostMultiplication} with $\diag(P^{-1},P^{-1})$, and (3) change the LMI variables via $M \triangleq P^{-1}$ and $L \triangleq KP^{-1}$. 
\end{proof}

\subsubsection{\textbf{$(Q,S,R)$-Dissipativation}}
The following proposition provides an LMI condition that leads to synthesize a FSF controller $K$ such that the closed-loop system \eqref{Eq:DTLTIUnderFSF} is $(Q,S,R)$-dissipative from $w(t)$ to $y(t)$. 

\begin{proposition}\label{Pr:DTDissipativationUnderFSF}
Under $D=\0$, the closed-loop DT-LTI system \eqref{Eq:DTLTIUnderFSF} is $(Q,S,R)$-dissipative (with $-Q>0$, $R=R^\T$) from $w(t)$ to $y(t)$ iff $\exists M>0$ and $L$ such that 
\begin{equation}\label{Eq:Pr:DTDissipativationUnderFSF}
    \bm{
    M & MC^\T S & MA^\T + L^\T B^\T & MC^\T \\ \star & F^\T S+S^\T F + R & E^\T &F^\T \\ 
    \star & \star & M & \0 \\ \star & \star & \star & -Q^{-1}}>0 
\end{equation}
and $K=LM^{-1}$.
\end{proposition}
\begin{proof}
The proof is complete by following the steps: 
(1) apply Prop. \ref{Pr:DTLTIQSRDissipativity} for \eqref{Eq:DTLTIUnderFSF} to get an LMI in $P$ and $K$, 
(2) transform the LMI using Lm. \ref{Lm:PreAndPostMultiplication} with $\diag(P^{-1},\I,P^{-1},\I)$, and (3) change the LMI variables using $M \triangleq P^{-1}$ and $L \triangleq KP^{-1}$.  
\end{proof}

\subsubsection{\textbf{$\H_2$-Optimal Control}}
Under FSF control $u(t)=Kx(t)$, the goal of $\H_2$-optimal control is to synthesize a controller $K$ that minimizes the $\H_2$-norm of the closed-loop system \eqref{Eq:DTLTIUnderFSF} (from $w(t)$ to $y(t)$). For this purpose, the following proposition provides an LMI based approach.

\begin{proposition}\label{Pr:DTH2ControlUnderFSF}
The $\H_2$-optimal FSF controller $K$ for the closed-loop system \eqref{Eq:DTLTIUnderFSF} is found by solving the LMI:
\begin{equation}\label{Eq:Pr:DTH2ControlUnderFSF}
\begin{aligned}
    \min_{M,Q,L,\gamma}&\ \gamma\\
    \mbox{sub. to:}&\ M>0,\ Q>0,\ \gamma>0,\\
    &\bm{M & AM+BL & E \\ \star & M & \0 \\ \star & \star & \I}>0,\\
    &\bm{Q & CM+DL & F \\ \star & M & \0 \\ \star & \star & \I}>0,\ \tr{Q}<\gamma^2, 
\end{aligned}
\end{equation}
and $K=LM^{-1}$.
\end{proposition}
\begin{proof}
The proof is complete by applying Prop. \ref{Pr:DTH2Norm} to \eqref{Eq:DTLTIUnderFSF} and executing the change of variables: $M=P$ and $L=KP$.
\end{proof}

\subsubsection{\textbf{$\H_\infty$-Optimal Control}}
Similarly, the goal of $\H_\infty$-optimal control is to synthesize a controller $K$ that minimizes the $\H_\infty$-norm of the closed-loop system \eqref{Eq:CTLTIUnderFSF} (from $w(t)$ to $y(t)$). 

\begin{proposition}\label{Pr:DTHInfControlUnderFSF}
The $\H_\infty$-optimal FSF controller $K$ for the closed-loop system \eqref{Eq:DTLTIUnderFSF} is found by solving the LMI:
\begin{equation}\label{Eq:Pr:DTHInfControlUnderFSF}
\begin{aligned}
    \min_{M,L,\gamma}&\ \gamma\\
    \mbox{sub. to:}&\ M>0, \gamma>0,\\
    &\bm{
    M & AM+BL & E & \0 \\ \star & M & \0 & MC^\T+L^\T D^\T \\
    \star & \star & \gamma \I & F^\T \\ \star & \star & \star & \gamma \I}>0,
\end{aligned}
\end{equation}
and $K=LM^{-1}$.
\end{proposition}
\begin{proof}
The proof is complete by following the steps: (1) apply Prop. \ref{Pr:DTHInfNorm} to \eqref{Eq:DTLTIUnderFSF}, (2) transform the main LMI using Lm. \ref{Lm:PreAndPostMultiplication} with $\diag(P^{-1},P^{-1},\I,\I)$, and (3) change the LMI variables using $M \triangleq P^{-1}$ and $L=KP^{-1}$. 
\end{proof}

\subsection{Observer Design for DT-LTI Systems}
For full-state feedback control $u(t)=Kx(t)$ the controller requires the state information $x(t)$ of the DT-LTI system \eqref{Eq:NoisyDTLTISystem}. Typically, $x(t)$ is not available at the controller, and thus, an \emph{observer} is required to keep an \emph{estimate} of $x(t)$ as $\hat{x}(t)$.

\subsubsection{\textbf{Luenberger Observer}}
For the DT-LTI system \eqref{Eq:NoisyDTLTISystem}, consider a \emph{Luenberger observer} implemented at the controller:
\begin{equation}\label{Eq:DTLuenbergerObserver}
    \hat{x}(t+1) = \hat{A}\hat{x}(t) + \hat{B}u(t) + Ly(t),
\end{equation}
that has the estimation error ($e(t)\triangleq x(t)-\hat{x}(t)$) dynamics:
\begin{equation}\label{Eq:DTLuenbergerObserverErrorDynamics}
\begin{aligned}
    e(t+1) = &\hat{A}e(t) + (A-\hat{A}-LC)x(t) \
    \\ &+ (B-\hat{B}-LD)u(t) + (E-LF)w(t).
\end{aligned}
\end{equation}
The Luenberger observer parameters: $\hat{A}, \hat{B}$ and $L$ can be selected according to the following proposition. 

\begin{proposition}\label{Pr:DTObserver}
Under $w(t)=0$ and Luenberger observer \eqref{Eq:DTLuenbergerObserver} parameters: $\hat{A} \triangleq A-LC$ and $\hat{B} \triangleq B-LD$, the estimation error dynamics \eqref{Eq:DTLuenbergerObserverErrorDynamics} are stable iff $\exists P > 0$ and $K$ such that 
\begin{equation}\label{Eq:Pr:DTObserver}
    \bm{ P & A^\T P- C^\T K^\T\\ \star & P} > 0 
\end{equation}
and $L=P^{-1}K$.
\end{proposition}
\begin{proof}
Under $w(t)=0$ and the given observer parameter choices, the estimation error dynamics \eqref{Eq:DTLuenbergerObserverErrorDynamics} reduces to: 
\begin{equation}\label{Eq:Pr:DTObserverStep1}
e(t+1) = (A-LC)e(t).    
\end{equation}
The proof is complete by applying Prop. \ref{Pr:DTLTIStability} to \eqref{Eq:Pr:DTObserverStep1} and then executing a change of variables using $K=PL$. 
\end{proof}

The Luenberger observer design proposed above assumes the noise-less case of \eqref{Eq:NoisyDTLTISystem} (i.e., \eqref{Eq:DTLTISystem}). This assumption is relaxed in the $(Q,S,R)$-dissipative and $\H_2/\H_\infty$-optimal observer designs described in subsequent subsections. However, before getting into those details, first, consider the Luenberger observer \eqref{Eq:DTLuenbergerObserver} parameters: $\hat{A} \triangleq A-LC$ and $\hat{B} \triangleq B-LD$ under which the estimation error dynamics \eqref{Eq:DTLuenbergerObserverErrorDynamics} take the form 
\begin{equation}\label{Eq:DTLuenbergerObserverWithPerformance}
    \begin{aligned}
        e(t+1) =&\ (A-LC)e(t)+(E-LF)w(t),\\
        z(t) =&\ Ge(t)+Jw(t).
    \end{aligned}
\end{equation}
where $z(t)$ is a pre-defined performance metric.

\subsubsection{\textbf{$(Q,S,R)$-Dissipative Observer}}
The $(Q,S,R)$-dissipative observer synthesizes the Luenberger observer gain $L$ such that \eqref{Eq:DTLuenbergerObserverWithPerformance} is $(Q,S,R)$-dissipative from $w(t)$ to $z(t)$. For this purpose, the following proposition can be used.

\begin{proposition}\label{Pr:DTDissipativeObserver}
The estimation error dynamics \eqref{Eq:DTLuenbergerObserverWithPerformance} is $(Q,S,R)$-dissipative (with $-Q>0$, $R=R^\T$) from $w(t)$ to $z(t)$ iff $\exists P>0$ and $K$ such that 
\begin{equation}\label{Eq:Pr:DTDissipativeObserver}
    \bm{
    P & G^\T S & A^\T P - C^\T K^\T & G^\T \\ \star & J^\T S + S^\T J + R & E^\T P - F^\T K^\T & J^\T \\
    \star & \star & P & \0 \\ \star & \star & \star & -Q^{-1}
    }>0
\end{equation}
and $L=KP^{-1}$.
\end{proposition}
\begin{proof}
The proof is complete by applying Prop. \ref{Pr:DTLTIQSRDissipativity} to \eqref{Eq:DTLuenbergerObserverWithPerformance} and executing a change of variables using $K=PL$. 
\end{proof}

\subsubsection{\textbf{$\H_2$-Optimal Observer}}
The goal of $\H_2$-optimal observer is to synthesize the Luenberger observer gain $L$ that minimizes the $\H_2$-norm of \eqref{Eq:LuenbergerObserverWithPerformance} (from $e(t)$ to $z(t)$). 

\begin{proposition}\label{Pr:DTH2Observer}
The $\H_2$-optimal observer gain $L$ for \eqref{Eq:DTLuenbergerObserverWithPerformance} is found by solving the LMI:
\begin{equation}\label{Eq:Pr:DTH2Observer}
\begin{aligned}
    \min_{P,Q,K,\gamma}&\ \gamma\\
    \mbox{sub. to:}&\ P>0,\ Q>0,\ \gamma>0,\\
    &\bm{P & PA-KC & PE-KF \\ \star & P & \0 \\\star & \star & \I}>0,\\
    &\bm{Q & PG & J \\ \star  & P & \0 \\ \star  & \star & \I}>0,\ \tr{Q}<\gamma^2,
\end{aligned}
\end{equation}
and $L=P^{-1}K$.
\end{proposition}
\begin{proof}
The proof is complete by applying Prop. \ref{Pr:DTH2Norm} to \eqref{Eq:DTLuenbergerObserverWithPerformance} and executing a change of variables using $K=PL$. 
\end{proof}

\subsubsection{\textbf{$\H_\infty$-Optimal Observer}}
The goal of $\H_\infty$-optimal observer is to synthesize the Luenberger observer gain $L$ that minimizes the $\H_\infty$-norm of \eqref{Eq:DTLuenbergerObserverWithPerformance} (from $e(t)$ to $z(t)$).

\begin{proposition}\label{Pr:DTHInfObserver}
The $\H_\infty$-optimal observer gain $L$ for \eqref{Eq:DTLuenbergerObserverWithPerformance} is found by solving the LMI:
\begin{equation}\label{Eq:Pr:DTHInfObserver}
\begin{aligned}
    \min_{P,K,\gamma}&\ \gamma\\
    \mbox{sub. to:}&\ P>0,\ \gamma>0,\\
    &\bm{ P & PA-KC & PE-KF & \0 \\\star & P & \0 & G^\T \\
    \star & \star & \gamma \I & J^\T \\ \star & \star & \star & \gamma \I}>0,
\end{aligned}
\end{equation}
and $L=P^{-1}K$.
\end{proposition}
\begin{proof}
The proof is complete by applying Prop. \ref{Pr:DTHInfNorm}\eqref{Eq:Pr:DTHInfNorm2} to \eqref{Eq:DTLuenbergerObserverWithPerformance} and executing a change of variables using $K=PL$. 
\end{proof}

\subsection{Dynamic Output Feedback (DOF) Control Synthesis for DT-LTI Systems}

Consider the DT-LTI system \eqref{Eq:DTLTISystem} with noise $w(t)\in\R^{q}$ and performance $z(t)\in\R^{l}$:
\begin{equation}\label{Eq:NoisyDTLTISystemWithPerf}
    \begin{aligned}
        x(t+1) = Ax(t) + Bu(t) + Ew(t),\\
        y(t) = Cx(t) + Du(t) + Fw(t),\\
        z(t) = Gx(t) + Hu(t) + Jw(t).
    \end{aligned}
\end{equation}
Under $D=\0$ and dynamic output feedback (DOF) control (from $y(t)$ to $u(t)$):
\begin{equation}\label{Eq:DTDOFController}
    \begin{aligned}
        \zeta(t+1) = A_c \zeta(t) + B_c y(t),\\
        u(t) = C_c \zeta(t) + D_c y(t),
    \end{aligned}
\end{equation}
where $\zeta(t)\in\R^{r}$, the closed-loop DT-LTI system \eqref{Eq:NoisyDTLTISystemWithPerf} takes the form:
\begin{equation}\label{Eq:DTLTIUnderDOF}
    \begin{aligned}
        \theta(t+1) = \bar{A}\theta(t) + \bar{B}w(t),\\
        z(t) = \bar{C}\theta(t) + \bar{D}w(t),
    \end{aligned}
\end{equation}
with $\theta(t)=\bm{x^\T(t) & \zeta^\T(t)}^\T$ and 
\begin{equation*}
\begin{aligned}
    \bar{A} \triangleq \bm{A+BD_cC & BC_c\\B_cC & A_c},\ 
    \bar{B} \triangleq \bm{E+BD_cF\\B_cF},\\ 
    \bar{C} \triangleq \bm{G+HD_cC & HC_c},\ 
    \bar{D} \triangleq \bm{J+HD_cF}.
\end{aligned}
\end{equation*}

In parallel to Sec. \ref{SubSec:DOFControllerSynthesis}, in the subsequent subsections, we provide LMI conditions for DOF controller synthesis (i.e., to design $A_c,B_c,C_c,D_c$ in \eqref{Eq:DTDOFController}) so as to stabilize, $(Q,S,R)$-dissipativate or optimize $\H_2$/$\H_\infty$-norm of the closed-loop system \eqref{Eq:DTLTIUnderDOF}. For this purpose, we can use the same change of variables process introduced in section \ref{SubSubSec:ChangeOfVariables} due to the equivalence of coefficient matrices in \eqref{Eq:DTLTIUnderDOF} and \eqref{Eq:CTLTIUnderDOF}.

\subsubsection{\textbf{Stabilization}}
The following proposition gives an LMI condition that leads to synthesize a DOF controller \eqref{Eq:DTDOFController} such that the closed-loop system  \eqref{Eq:DTLTIUnderDOF} is stabilized.

\begin{proposition}\label{Pr:DTStabilizationUnderDOF}
Under $D=w(t)=\0$, the closed-loop DT-LTI system \eqref{Eq:DTLTIUnderDOF} is stable iff $\exists X,Y>0$ and $A_n,B_n,C_n,D_n$ such that 
\begin{align}
\label{Eq:Pr:DTStabilizationUnderDOF1}
\bm{Y & \I \\ \I & X}>0,\\ 
\label{Eq:Pr:DTStabilizationUnderDOF2}
\bm{Y & \I & YA^\T+C_n^\T B^\T & A_n^\T \\ \star & X & A^\T+C^\T D_n^\T B^\T & A^\T X+C^\T B_n^\T \\ 
\star & \star & Y & \I \\ \star & \star & \star & X}>0,
\end{align}
and $A_c,B_c,C_c,D_c$ \eqref{Eq:DTDOFController} are found by CoVs \eqref{Eq:CoVsXYToMN} and \eqref{Eq:COVsAuxToDOF}. 
\end{proposition}

\begin{proof}
Applying Prop. \ref{Pr:DTLTIStability} to \eqref{Eq:DTLTIUnderDOF} give the LMI conditions necessary and sufficient for the stabilization of \eqref{Eq:DTLTIUnderDOF} as: $\exists P>0$ such that 
\begin{equation}\label{Eq:Pr:DTStabilizationUnderDOFStep1}
    \bm{ P & \bar{A}^\T P \\ \star & P } >0.  
\end{equation}
Using the CoVs in \eqref{Eq:CoVPAndPiDef} and \eqref{Eq:CoVPisPositiveDefinite}, the LMI $P>0$ can be transformed to \eqref{Eq:Pr:DTStabilizationUnderDOF1}. Finally, \eqref{Eq:Pr:DTStabilizationUnderDOFStep1} can be transformed to \eqref{Eq:Pr:DTStabilizationUnderDOF2} by applying Lm. \ref{Lm:PreAndPostMultiplication} with $\diag(\Pi_x^\T P^{-1}, \Pi_x^\T P^{-1})$ and substituting from \eqref{Eq:CoVPixPiy} as:
\begin{align}
\bm{ P & \bar{A}^\T P \\ \star & P } >0
\iff 
\bm{\Pi_x^\T P^{-1}\Pi_x & \Pi_x^\T P^{-1}\bar{A}^\T\Pi_x \\ \star & \Pi_x^\T P^{-1}\Pi_x}>0 \nonumber\\
\iff 
\bm{\Pi_y^\T \Pi_x & \Pi_y^\T \bar{A}^\T \Pi_x \\ \star & \Pi_y^\T\Pi_x}>0 \iff \mbox{\eqref{Eq:Pr:DTStabilizationUnderDOF2}}.\nonumber
\end{align} 
Note that the last step above results from \eqref{Eq:COVsKeyResults}.
\end{proof}

\subsubsection{\textbf{$(Q,S,R)$-Dissipativation}}
The following proposition provides an LMI condition that leads to synthesize a DOF controller \eqref{Eq:DTDOFController} such that the closed-loop system \eqref{Eq:DTLTIUnderDOF} is $(Q,S,R)$-dissipative from $w(t)$ to $z(t)$. 

\begin{proposition}\label{Pr:DTDissipativationUsingDOF}
Under $D=\0$, the closed-loop DT-LTI system \eqref{Eq:DTLTIUnderDOF} is $(Q,S,R)$-dissipative (with $-Q>0$, $R=R^\T$) from $w(t)$ to $z(t)$ iff $\exists X,Y>0$ and $A_n,B_n,C_n,D_n$ such that 
\begin{equation}\label{Eq:Pr:DTDissipativationUsingDOF1}
    \bm{Y & \I \\ \I & X}>0 \mbox{ and  \eqref{Eq:Pr:DTDissipativationUsingDOF2}},
\end{equation}
and $A_c,B_c,C_c,D_c$ \eqref{Eq:DTDOFController} are found by CoVs \eqref{Eq:CoVsXYToMN} and \eqref{Eq:COVsAuxToDOF}. 
\end{proposition}
\begin{proof}
The proof starts with applying Prop. \ref{Pr:DTLTIQSRDissipativity} to \eqref{Eq:DTLTIUnderDOF} to obtain the  LMI conditions necessary and sufficient for the $(Q,S,R)$-dissipativation of \eqref{Eq:DTLTIUnderDOF} as: 
$\exists P>0$ such that 
\begin{equation}\label{Eq:Pr:DTDissipativationUsingDOFStep1}
    \bm{
    P & \bar{C}^\T S & \bar{A}^\T P & \bar{C}^\T \\ 
    \star & \bar{D}^\T S + S^\T \bar{D} + R & \bar{B}^\T P & \bar{D}^\T \\ 
    \star & \star & P & \0 \\
    \star & \star & \star & -Q^{-1}
    }>0.
\end{equation}
Similar to the proof of Prop. \ref{Pr:DTStabilizationUnderDOF}, using the CoVs in \eqref{Eq:CoVPAndPiDef} and \eqref{Eq:CoVPisPositiveDefinite}, the LMI $P>0$ can be transformed to \eqref{Eq:Pr:DTDissipativationUsingDOF1}. To obtain \eqref{Eq:Pr:DTDissipativationUsingDOF2} from \eqref{Eq:Pr:DTDissipativationUsingDOFStep1}, first, Lm. \ref{Lm:PreAndPostMultiplication} is applied with $\diag(\Pi_x^\T P^{-1},\I,\Pi_x^\T P^{-1},\I\}$ and then the results are substituted using \eqref{Eq:CoVPixPiy} to obtain:
\begin{equation}
    \bm{
    \Pi_y^\T\Pi_x & \Pi_y^\T \bar{C}^\T S & \Pi_y^\T \bar{A}^\T\Pi_x & \Pi_y^\T \bar{C}^\T \\
    \star & \H_s(\bar{D}^\T S)+R & \bar{B}^\T \Pi_x & \bar{D}^\T\\ 
    \star & \star & \Pi_y^\T\Pi_x & \0\\
    \star & \star & \star & -Q^{-1}
    }>0.
\end{equation}
Finally, the above matrix inequality can be transformed to get the LMI in \eqref{Eq:Pr:DTDissipativationUsingDOF2} using \eqref{Eq:COVsKeyResults}.
\end{proof}

\subsubsection{\textbf{$\H_2$-Optimal Control}}
The goal of $\H_2$-optimal control here is to synthesize a DOF controller \eqref{Eq:DTDOFController} that minimizes the $\H_2$-norm of the closed-loop system \eqref{Eq:DTLTIUnderDOF} from $w(t)$ to $z(t)$. 

\begin{proposition}\label{Pr:DTH2ControlUnderDOF}
Under $D=\0$, the $\H_2$-optimal DOF controller \eqref{Eq:DTDOFController} for the closed-loop system \eqref{Eq:DTLTIUnderDOF} is found by solving the LMI:
\begin{equation}\label{Eq:Pr:DTH2ControlUnderDOF1}
\begin{aligned}
    \min_{\substack{X,Y,Q,\gamma\\A_n,B_n,C_n,D_n}}\ \ &\gamma\\
    \mbox{sub. to:}\ \ &X>0,\ Y>0,\ Q>0,\ \gamma>0,\ \mbox{\eqref{Eq:Pr:DTH2ControlUnderDOF2}},\\
    &\tr{Q}<\gamma^2,
\end{aligned}
\end{equation}
and $A_c,B_c,C_c,D_c$ \eqref{Eq:DTDOFController} are found by CoVs \eqref{Eq:CoVsXYToMN} and \eqref{Eq:COVsAuxToDOF}.
\end{proposition}
\begin{proof}
The proof follows similar steps as that of Prop. \ref{Pr:DTDissipativationUsingDOF}.  
\end{proof}

\subsubsection{\textbf{$\H_\infty$-Optimal Control}}
The goal of $\H_\infty$-optimal control is to synthesize a DOF controller \eqref{Eq:DTDOFController} that minimizes the $\H_\infty$-norm of the closed-loop system \eqref{Eq:DTLTIUnderDOF} (from $w(t)$ to $z(t)$). 

\begin{proposition}\label{Pr:DTHInfControlUnderDOF}
Under $D=0$, the $\H_\infty$-optimal DOF controller \eqref{Eq:DTDOFController} for the closed-loop system \eqref{Eq:DTLTIUnderDOF} is found by solving the LMI:
\begin{equation}\label{Eq:Pr:DTHInfControlUnderDOF1}
\begin{aligned}
    \min_{\substack{X,Y,\gamma\\A_n,B_n,C_n,D_n}} &\ \gamma\\
    \mbox{sub. to:}&\ X>0, Y>0, \gamma>0, \eqref{Eq:Pr:DTHInfControlUnderDOF2}.
\end{aligned}
\end{equation}
and $A_c,B_c,C_c,D_c$ \eqref{Eq:DTDOFController} are found by CoVs \eqref{Eq:CoVsXYToMN} and \eqref{Eq:COVsAuxToDOF}.
\end{proposition}
\begin{proof}
The proof follows similar steps as that of Prop. \ref{Pr:DTDissipativationUsingDOF}.  
\end{proof}

An outline of all the established theoretical results in this section can be found in Tab. \ref{Tab:DTLTISystemsLMIResults}.

\begin{figure*}[!hb]
    \centering
    \hrulefill
    \begin{equation}\label{Eq:Pr:DTDissipativationUsingDOF2}
    \bm{
    Y & \I & (YG^\T+C_n^\T H^\T)S & YA^\T + C_n^\T B^\T & A_n^\T & YG^\T+C_n^\T H^\T\\
    \star & X & (G^\T+C^\T D_n^\T H^\T)S & A^\T+C^\T D_n^\T B^\T & A^\T X+C^\T B_n^\T & G^\T+C^\T D_n^\T H^\T \\ 
    \star & \star & \H_s((J^\T+F^\T D_n^\T H^\T)S)+R & E^\T+F^\T D_n^\T B^\T & E^\T X+F^\T B_n^\T & J^\T+F^\T D_n^\T H^\T
    \\\star & \star & \star & Y & \I & \0 \\ 
    \star & \star & \star & \star & X & \0 \\
    \star  & \star & \star  & \star & \star & -Q^{-1}
    }>0
    \end{equation}
    \begin{equation}\label{Eq:Pr:DTH2ControlUnderDOF2}
    \begin{aligned}
    &\bm{
    Y & \I & YA^\T + C_n^\T B^\T & A_n^\T & YG^\T + C_n^\T H^\T\\
    \star & X & A^\T + C^\T D_n^\T B^\T & A^\T X + C^\T B_n^\T & G^\T + C^\T D_n^\T H^\T\\
    \star & \star & Y & \I & \0\\
    \star & \star & \star & X & \0\\
    \star & \star & \star & \star & \I
    }>0 \ \ \mb{ and }\\
    &\bm{
    Q & E^\T + F^\T D_n^\T B^\T & E^\T X + F^\T B_n^\T & J^\T+F^\T D_n^\T H^\T \\
    \star & Y & \I & \0 \\ 
    \star & \star & X & \0 \\
    \star & \star & \star & \I
    }>0
    \end{aligned}
    \end{equation}
    \begin{equation}\label{Eq:Pr:DTHInfControlUnderDOF2}
        \bm{
        Y & \I & AY+BC_n & A+BD_nC & E+BD_nF & \0\\
        \star & X & A_n & XA+B_nC & XE + B_nF & \0\\
        \star & \star & Y & \I & \0 & YG^\T+C_n^\T H^\T\\
        \star & \star & \star & X & \0 & G^\T+C^\T D_n\T H^\T\\
        \star & \star & \star & \star & \gamma \I & J^\T + F^\T D_n^\T H^\T\\
        \star & \star & \star & \star & \star & \gamma \I 
        }>0
    \end{equation}
\end{figure*}

\begin{table}[!t]
\caption{Summary of Theoretical Results (DT-LTI Systems).}
\label{Tab:DTLTISystemsLMIResults}
\resizebox{\columnwidth}{!}{
\begin{tabular}{|ll|llll|}
\hline
\multicolumn{2}{|c|}{\multirow{3}{*}{\begin{tabular}[c]{@{}c@{}}Proposition \#\\ (For DT-LTI Systems) \end{tabular} }} &
  \multicolumn{4}{c|}{Concept} \\ \cline{3-6} 
\multicolumn{2}{|c|}{} &
  \multicolumn{1}{c|}{Stability} &
  \multicolumn{1}{c|}{\begin{tabular}[c]{@{}c@{}}$(Q,S,R)$-\\ Dissipativity\end{tabular}} &
  \multicolumn{1}{c|}{\begin{tabular}[c]{@{}c@{}}$\H_2$-\\ Norm\end{tabular}} &
  \multicolumn{1}{c|}{\begin{tabular}[c]{@{}c@{}}$\H_\infty$-\\ Norm\end{tabular}} \\ \hline
\multicolumn{1}{|l|}{\multirow{8}{*}{Task}} &
  \begin{tabular}[c]{@{}l@{}}DT-LTI System\\ Analysis\end{tabular} &
  \multicolumn{1}{c|}{\ref{Pr:DTLTIStability}} & 
  \multicolumn{1}{c|}{\ref{Pr:DTLTIQSRDissipativity}} &
  \multicolumn{1}{c|}{\ref{Pr:DTH2Norm}} & 
  \multicolumn{1}{c|}{\ref{Pr:DTHInfNorm}} 
   \\ \cline{2-6} 
\multicolumn{1}{|l|}{} &
  \begin{tabular}[c]{@{}l@{}}FSF Controller\\ Synthesis\end{tabular} &
  \multicolumn{1}{c|}{\ref{Pr:DTStabilizationUnderFSF},\ref{Pr:DTStabilizability}} &
  \multicolumn{1}{c|}{\ref{Pr:DTDissipativationUnderFSF}} &
  \multicolumn{1}{c|}{\ref{Pr:DTH2ControlUnderFSF}} & 
  \multicolumn{1}{c|}{\ref{Pr:DTHInfControlUnderFSF}}  
   \\ \cline{2-6} 
\multicolumn{1}{|l|}{} &
  \begin{tabular}[c]{@{}l@{}}Observer\\ Design\end{tabular} &
  \multicolumn{1}{c|}{\ref{Pr:DTObserver},\ref{Pr:DTDetectability}} &
  \multicolumn{1}{c|}{\ref{Pr:DTDissipativeObserver}} &
  \multicolumn{1}{c|}{\ref{Pr:DTH2Observer}} &
  \multicolumn{1}{c|}{\ref{Pr:DTHInfObserver}} 
   \\ \cline{2-6} 
\multicolumn{1}{|l|}{} &
  \begin{tabular}[c]{@{}l@{}}DOF Controller\\ Synthesis\end{tabular} &
  \multicolumn{1}{c|}{\ref{Pr:DTStabilizationUnderDOF}} &
  \multicolumn{1}{c|}{\ref{Pr:DTDissipativationUsingDOF}} &
  \multicolumn{1}{c|}{\ref{Pr:DTH2ControlUnderDOF}} & 
  \multicolumn{1}{c|}{\ref{Pr:DTHInfControlUnderDOF}} 
   \\ \hline
\end{tabular}}
\end{table}

\section{The Discrete-Time Networked System (DTNS)}
\label{Sec:DTNetworkedSystem}

This section briefly provides the unique details of the considered discrete-time networked system (DTNS).

\subsection{Subsystems of the DTNS}

\subsubsection{\textbf{Dynamics}} 
We consider the dynamics of the $i$\tsup{th} subsystem $\Sigma_i, i\in\N_N$ of the DTNS to be 
\begin{subequations}\label{Eq:DTSubsystemDynamics}
\begin{align}
    x_i(t+1) =& \sum_{j\in\bar{\E}_i}A_{ij}x_j(t) + \sum_{j\in\bar{\E}_i}B_{ij}u_j(t)+\sum_{j\in\bar{\E}_i}E_{ij}w_{j}(t),\\ 
    y_i(t) =& \sum_{j\in\bar{\E}_i}C_{ij}x_j(t) + \sum_{j\in\bar{\E}_i}D_{ij}u_j(t) + \sum_{j\in\bar{\E}_i}F_{ij}w_j(t),
\end{align}
\end{subequations}
where now $t\in\N$.

\subsubsection{\textbf{Local Controllers and Observers}}
As mentioned before, one main objective of this paper is to design local controllers/observers at the subsystems of the considered DTNS in a decentralized manner. In this setting, at a subsystem $\Sigma_i$: 
\begin{enumerate}
    \item a local FSF controller may take the form:
        \begin{equation}\label{Eq:DTLocalFSFController}
            u_i(t) = \sum_{j\in \bar{\E}_i} K_{ij} x_j(t);
        \end{equation}
    \item a local Luenberger observer may take the form: 
        \begin{equation}\label{Eq:DTLocalObserver}
            \hat{x}_i(t+1) = \sum_{j\in \bar{\E}_i} \hat{A}_{ij} \hat{x}_j(t) + \sum_{j\in \bar{\E}_i} \hat{B}_{ij} u_j + \sum_{j\in \bar{\E}_i} L_{ij}y_j(t);
        \end{equation}
    \item a local DOF controller may take the form:
    \begin{equation}\label{Eq:DTLocalDOFController}
        \begin{aligned}
            \zeta_i(t+1) = \sum_{j\in\bar{\E}_i} A_{c,ij} \zeta_j(t) + \sum_{j\in\bar{\E}_i} B_{c,ij} y_j(t),\\
            u_i(t) = \sum_{j\in\bar{\E}_i} C_{c,ij} \zeta_j(t) + \sum_{j\in\bar{\E}_i} D_{c,ij} y_j(t).
        \end{aligned}    
    \end{equation}
\end{enumerate}

\subsubsection{\textbf{Local Performance Metrics}} 
Recall that a pre-defined local performance metric is required at each subsystem $\Sigma_i$ when designing local controllers and observers in a $(Q,S,R)$-dissipative or $\H_2/\H_\infty$-optimal sense. Therefore, identically to the continuous-time case, at a subsystem $\Sigma_i$, we use the following pre-defined local performance metrics: 
\begin{enumerate}
    \item For local FSF controller design: 
    \begin{equation}\label{Eq:DTLocalFSFControllerPerf}
        y_i(t) = \sum_{j\in\bar{\E}_i}C_{ij}x_j(t) + \sum_{j\in\bar{\E}_i}D_{ij}u_j(t) + \sum_{j\in\bar{\E}_i}F_{ij}w_j(t);
    \end{equation} 
    \item For local Luenberger observer design: 
    \begin{equation}\label{Eq:DTLocalObserverPerf}
        z_i(t) = \sum_{j\in\bar{\E}_i} G_{ij}(x_j(t)-\hat{x}_j(t)) + \sum_{j\in\bar{\E}_i}J_{ij}w_j(t);
    \end{equation}
    \item For local DOF controller design:
    \begin{equation}\label{Eq:DTLocalDOFControllerPerf}
        z_i(t) = \sum_{j\in\bar{\E}_i} G_{ij}x_j(t) + \sum_{j\in\bar{\E}_i} H_{ij}u_j(t) + \sum_{j\in\bar{\E}_i}J_{ij}w_j(t).
    \end{equation}
\end{enumerate}

\subsection{The DTNS}
\subsubsection{\textbf{Dynamics}}

By stating \eqref{Eq:DTSubsystemDynamics} for all $i\in\N_N$, we can obtain the dynamics of the networked system $\mathcal{G}_N$ as 
\begin{subequations}\label{Eq:DTNSDynamics}
\begin{align}
x(t+1) =&\ Ax(t) + Bu(t) + Ew(t),\\
y(t) =&\ Cx(t) + Du(t) + Fw(t),
\end{align}
\end{subequations}
similar to \eqref{Eq:CTNSDynamics} except for the fact that now time $t\in\N$.

\subsubsection{\textbf{Controllers and Observers}}
By composing each local controller/observer forms in \eqref{Eq:DTLocalFSFController},\eqref{Eq:DTLocalObserver}, \eqref{Eq:DTLocalDOFController} for all $i\in \N_N$, we can respectively obtain the network level (i.e., global): 
\begin{enumerate}
    \item FSF controller 
        \begin{equation}\label{Eq:DTGlobalFSFController}
            u(t) = Kx(t);
        \end{equation}
    \item Luenberger observer 
        \begin{equation}\label{Eq:DTGlobalObserver}
            \hat{x}(t+1) = \hat{A}x(t) + \hat{B}u(t) + Ly(t);  
        \end{equation}
    \item DOF controller 
        \begin{equation}\label{Eq:DTGlobalDOFController}
        \begin{aligned}
            \zeta(t+1) = A_c \zeta(t) + B_c y(y),\\
            u(t) = C_c \zeta(t) + D_c y(y),
        \end{aligned}
        \end{equation}
\end{enumerate}
where $K,\hat{A},\hat{B},L,A_c,B_c,C_c,D_c$ are all $N\times N$ block matrices comprised of the corresponding local design parameters.

\subsubsection{\textbf{Performance Metrics}}
Similarly, by composing each pre-defined local controller/observer performance metric forms in \eqref{Eq:DTLocalFSFControllerPerf},\eqref{Eq:DTLocalObserverPerf},\eqref{Eq:DTLocalDOFControllerPerf} for all $i\in\N_N$, we can respectively obtain the global performance metrics considered for:
\begin{enumerate}
    \item FSF controller design as: 
    \begin{equation}\label{Eq:DTGlobalFSFControllerPerf}
        y(t) = Cx(t) + Dy(t) + Fw(t);
    \end{equation}
    \item Luenberger observer design as: 
    \begin{equation}\label{Eq:DTGlobalObserverPerf}
        z(t) = G(x(t)-\hat{x}(t))+Jw(t);
    \end{equation}
    \item DOF controller design as:
    \begin{equation}\label{Eq:DTGlobalDOFControllerPerf}
        z(t) = Gx(t)+Hu(t)+Jw(t).
    \end{equation}
\end{enumerate}
Here also $C,D,F,G,H,J$ are all $N\times N$ block matrices comprised of the corresponding pre-defined local performance metric parameters (e.g., $C=[C_{ij}]_{i,j\in\N_N}$).

\subsection{The Research Problem}
Similar to the continuous-time case, the forms of the DTNS \eqref{Eq:DTNSDynamics}, global controllers/observers \eqref{Eq:DTGlobalFSFController}-\eqref{Eq:DTGlobalDOFController} and global performance metrics \eqref{Eq:DTGlobalFSFControllerPerf}-\eqref{Eq:DTGlobalDOFControllerPerf} are respectively identical to the general DT-LTI system (e.g., \eqref{Eq:NoisyDTLTISystem}), controllers/observers (e.g., \eqref{Eq:DTLuenbergerObserver},\eqref{Eq:DTDOFController}) and performance metrics (e.g., \eqref{Eq:NoisyDTLTISystemWithPerf},\eqref{Eq:DTLuenbergerObserverWithPerformance}) considered in Sec. \ref{Sec:BasicsOfDTLTISystems}. 
Therefore, all the LMI-based control solutions (Prop. \ref{Pr:DTLTIStability}-\ref{Pr:DTHInfControlUnderDOF}) discussed in Sec. \ref{Sec:BasicsOfDTLTISystems} are directly applicable for the DTNS \eqref{Eq:DTNSDynamics}. Moreover, as pointed out in Rm. \ref{Rm:DTCaseAlternativeLMIs}, we have formulated these LMI-based control solutions such that they can be applied in a decentralized setting (via Alg. \ref{Alg:DistributedPositiveDefiniteness}). In other words, we can use the said LMI-based control solutions for decentralized analysis and controller/observer synthesis of the considered DTNS \eqref{Eq:DTNSDynamics}.


\section{Distributed Analysis and Control Synthesis of DTNS}
\label{Sec:DistributedTechniquesForDTNS}

In this section, parallel to Sec. \ref{Sec:DistributedTechniquesForCTNS}, we provide decentralized, compositional and possibly distributed techniques for different analysis and control synthesis tasks discussed in Sec. \ref{Sec:BasicsOfDTLTISystems} for DTNSs introduced in Sec. \ref{Sec:DTNetworkedSystem} exploiting the algorithms proposed in Sec. \ref{Sec:DecentralizedAnalysis}. Following the same steps as before, we have summarized the details of each decentralized technique that we propose in Tab. \ref{Tab:DT-LTILocalResultsSummary}. To provide examples, in the following subsections, stability and dissipativity related decentralized techniques have been formally stated as theorems (analogous to Theorems \ref{Th:CTLTIStability}-\ref{Th:DissipativationUsingDOF} provided in Sec. \ref{Sec:DistributedTechniquesForCTNS}). However, note that we omit providing proofs in this section as they can be obtained by following similar steps to their counterparts in Sec. \ref{Sec:DistributedTechniquesForCTNS}.


\subsection{Stability Related Decentralized Results}

In this subsection, we provide stability based results on:
\begin{enumerate}
    \item decentralized stability analysis,
    \item decentralized FSF stabilization,
    \item decentralized stable observer design,
    \item decentralized DOF stabilization,
\end{enumerate}
that respectively corresponds to Props. \ref{Pr:DTLTIStability}, \ref{Pr:DTStabilizationUnderFSF}, \ref{Pr:DTObserver} and \ref{Pr:DTStabilizationUnderDOF} (Tab. \ref{Tab:DT-LTILocalResultsSummary}-Rows 1, 7, 11 and 15). 

\begin{theorem} \label{Th:DTLTIStability}(Stability analysis)
The DTNS \eqref{Eq:DTNSDynamics} under $u(t)=\0$ and $w(t)=\0$ is stable if at each subsystem $\Sigma_i,i\in\N_N$, the problem  
\begin{equation}\label{Eq:Th:DTLTIStability}
    \mathbb{P}_9: \ \ \text{Find} \ \ P_{ii}\ \ \text{such that} \ \ P_{ii} > 0, \ \ \tilde{W}_{ii}>0,
\end{equation}
is feasible, where $\tilde{W}_{ii}$ is computed from Alg. \ref{Alg:DistributedPositiveDefiniteness} (Steps: 3-16) when analyzing $W = [W_{ij}]_{i,j\in\N_N}>0$ with  
\begin{equation}\label{Eq:Th:DTLTIStability2}
    W_{ij} = \bm{P_{ii}e_{ij} & A_{ji}^\T P_{jj} \\ P_{ii}A_{ij} & P_{ii}e_{ij}}.
\end{equation}
\end{theorem}

\begin{theorem}\label{Th:DTStabilizationUnderFSF}
(FSF Stabilization) The DTNS \eqref{Eq:DTNSDynamics} (where $B$ is block diagonal) under $D=\0$, $w(t)=\0$ and local FSF control \eqref{Eq:DTLocalFSFController} is stable if at each subsystem $\Sigma_i,i\in\N_N$, the problem    
\begin{equation}\label{Eq:Th:DTStabilizationUnderFSF}
    \mathbb{P}_{10}: \ \ \text{Find} \ \ M_{ii},\ \ L_i\ \ \text{such that} \ \ M_{ii} > 0, \ \ \tilde{W}_{ii}>0, 
\end{equation}
is feasible, where $\tilde{W}_{ii}$ is computed from Alg. \ref{Alg:DistributedPositiveDefiniteness} (Steps: 3-16) when enforcing $W = [W_{ij}]_{i,j\in\N_N}>0$ with 
\begin{equation}\label{Eq:Th:DTStabilizationUnderFSF2}
    W_{ij} =  \bm{M_{ii}e_{ij} & M_{ii}A_{ji}^\T+L_{ji}^\T B_{ii}^\T \\ \star & M_{ii}e_{ij}}
\end{equation}
The local FSF controller gains $K_i$ are computed using $M_{ii}$ and $L_i$ identically to Th. \ref{Th:StabilizationUnderFSF} via \eqref{Eq:Th:StabilizationUnderFSF3}.
\end{theorem}

\begin{theorem}\label{Th:DTObserver}
(Stable observer design) For the DTNS \eqref{Eq:DTNSDynamics} (where $C,D$ are block diagonal) under $w(t)=\0$, the local Luenberger observers \eqref{Eq:DTLocalObserver} render the state estimation error dynamics \eqref{Eq:DTLuenbergerObserver} stable if at each subsystem $\Sigma_i,i\in\N_N$, the problem
\begin{equation}\label{Eq:Th:DTObserver}
    \mathbb{P}_{11}: \ \ \text{Find} \ \ P_{ii},\ \ K_i\ \ \text{such that} \ \ P_{ii} > 0, \ \ \tilde{W}_{ii}>0, 
\end{equation}
is feasible, where $\tilde{W}_{ii}$ is computed from Alg. \ref{Alg:DistributedPositiveDefiniteness} (Steps: 3-16) when enforcing $W = [W_{ij}]_{i,j\in\N_N}>0$ with 
\begin{equation}\label{Eq:Th:DTObserver2}
    W_{ij} = \bm{ P_{ii}e_{ij} & A_{ji}^\T P_{jj}- C_{ii}^\T K_{ji}^\T\\ \star & P_{ii}e_{ij}}
\end{equation}
The local Luenberger observer parameters $L_i, \hat{A}_{i}$ and $\hat{B}_i$ \eqref{Eq:DTLocalObserver} are computed using $P_{ii}$ and $K_i$ identically to Th. \ref{Th:Observer} via \eqref{Eq:Th:Observer3}.
\end{theorem}

\begin{theorem}\label{Th:DTStabilizationUnderDOF}
(DOF Stabilization) The DTNS \eqref{Eq:DTNSDynamics} (where $B,C$ are block diagonal) under $D=\0$, $w(t)=\0$ and local DOF control \eqref{Eq:DTLocalDOFController} is stable if at each subsystem $\Sigma_i,i\in\N_N$, the problem    
\begin{equation}\label{Eq:Th:DTStabilizationUnderDOF}
\begin{aligned}
    \mathbb{P}_{12}: \ \ \ \ \ \ \text{Find}& \ \ X_{ii},\ Y_{ii},\ A_{n,i},\ B_{n,i},\ C_{n,i},\ D_{n,i} \\ 
    \text{such that}& \ \ X_{ii} > 0,\ Y_{ii} > 0,\ \tilde{W}_{ii}^{(1)}>0,\ \tilde{W}_{ii}^{(2)}>0,  
\end{aligned}
\end{equation}
is feasible, where $\tilde{W}_{ii}^{(k)}$ is computed from Alg. \ref{Alg:DistributedPositiveDefiniteness} (Steps: 3-16) when enforcing $W^{(k)} = [W_{ij}^{(k)}]_{i,j\in\N_N}>0$, for $k=1,2,$ with 
\begin{equation}\label{Eq:Th:DTStabilizationUnderDOF2}
\begin{aligned}
    &W_{ij}^{(1)} = \bm{Y_{ii}e_{ij} & e_{ij} \\ e_{ij} & X_{ij}e_{ij}},\ \ \ \ W_{ij}^{(2)} =\\
    &\bm{Y_{ii}e_{ij} & e_{ij} & Y_{ii}A_{ji}^\T+C_{n,ji}^\T B_{jj}^\T & A_{n,ji}^\T \\
    \star & X_{ii}e_{ij} & A_{ji}^\T+C_{ii}^\T D_{n,ji}^\T B_{jj}^\T & A_{ji}^\T X_{jj}+C_{ii}^\T B_{n,ji}^\T \\ 
    \star & \star & Y_{ii}e_{ij} & e_{ij} \\ \star & \star & \star & X_{ii}e_{ij}}.
\end{aligned}
\end{equation}
The local DOF controller parameters $A_{c,i},B_{c,i},C_{c,i}$ and $D_{c,i}$ are computed using $X_{ii},Y_{ii},A_{n,i},B_{n,i},C_{n,i}$ and $D_{n,i}$ identically to Th. \ref{Th:StabilizationUnderDOF} via \eqref{Eq:Th:StabilizationUnderDOF3} and \eqref{Eq:Th:StabilizationUnderDOF4}.
\end{theorem}

\subsection{Dissipativity Related Decentralized Results}

In this subsection, analogous to Theorems \ref{Th:DTLTIStability}-\ref{Th:DTStabilizationUnderDOF}, we provide $(Q,S,R)$-dissipativity based results for:
\begin{enumerate}
    \item decentralized $(Q,S,R)$-dissipativity analysis,
    \item decentralized FSF $(Q,S,R)$-dissipativation,
    \item decentralized $(Q,S,R)$-dissipative observer design,
    \item decentralized DOF $(Q,S,R)$-dissipativation,
\end{enumerate}
that respectively corresponds to Props. \ref{Pr:DTLTIQSRDissipativity}, \ref{Pr:DTDissipativationUnderFSF}, \ref{Pr:DTDissipativeObserver} and \ref{Pr:DTDissipativationUsingDOF} (Tab. \ref{Tab:DT-LTILocalResultsSummary}-Rows 2, 8, 12 and 16). In what follows, regarding the given specification matrices $Q,S,R$, it is assumed that: (1) $Q$ is a block diagonal network matrix, (2) $-Q>0$, and (3) $R=R^\T$.

\begin{theorem} \label{Th:DTLTIQSRDissipativity}(Dissipativity analysis)
The DTNS \eqref{Eq:DTNSDynamics} (where $C,D$ are block diagonal) under $w(t)=\0$ is $(Q,S,R)$-dissipative from $u(t)$ to $y(t)$ if at each subsystem $\Sigma_i,i\in\N_N$, the problem  
\begin{equation}\label{Eq:Th:DTLTIQSRDissipativity}
    \mathbb{P}_{13}: \ \ \text{Find} \ \ P_{ii}\ \ \text{such that} \ \ P_{ii} > 0, \ \ \tilde{W}_{ii}>0,
\end{equation}
is feasible, where $\tilde{W}_{ii}$ is computed from Alg. \ref{Alg:DistributedPositiveDefiniteness} (Steps: 3-16) when analyzing $W = [W_{ij}]_{i,j\in\N_N}>0$ with  
\begin{equation}\label{Eq:Th:DTLTIQSRDissipativity2}
    W_{ij} = 
    \bm{
    P_{ii}e_{ij} & C_{ii}^\T S_{ij} & A_{ji}^\T P_{jj} & C_{ii}^\T e_{ij} \\ 
    \star & \H_e(D_{ii}^\T S_{ij}) + R_{ij} & B_{ji}^\T P_{jj} & D_{ii}^\T e_{ij} \\ 
    \star & \star & P_{ii}e_{ij} & \0 \\
    \star & \star & \star & -Q_{ii}^{-1}e_{ij}
    }.
\end{equation}
\end{theorem}

\begin{theorem}\label{Th:DTDissipativationUnderFSF}
(FSF Dissipativation) The DTNS \eqref{Eq:DTNSDynamics} (where $B,C,F$ are block diagonal) under $D=\0$ and local FSF control \eqref{Eq:DTLocalFSFController} is $(Q,S,R)$-dissipative from $w(t)$ to $y(t)$ if at each subsystem $\Sigma_i,i\in\N_N$, the problem    
\begin{equation}\label{Eq:Th:DTDissipativationUnderFSF}
    \mathbb{P}_{14}: \ \ \text{Find} \ \ M_{ii},\ \ L_i\ \ \text{such that} \ \ M_{ii} > 0, \ \ \tilde{W}_{ii}>0, 
\end{equation}
is feasible, where $\tilde{W}_{ii}$ is computed from Alg. \ref{Alg:DistributedPositiveDefiniteness} (Steps: 3-16) when enforcing $W = [W_{ij}]_{i,j\in\N_N}>0$ with $W_{ij}$ given in \eqref{Eq:Th:DTDissipativationUnderFSF2}. 
The local FSF controller gains $K_i$ are computed using $M_{ii}$ and $L_i$ identically to Th. \ref{Th:StabilizationUnderFSF} via \eqref{Eq:Th:StabilizationUnderFSF3}.
\end{theorem}

\begin{theorem}\label{Th:DTDissipativeObserver}
(Dissipative observer design) For the DTNS \eqref{Eq:DTNSDynamics} (where $C,D,F$ are block diagonal), the local Luenberger observers \eqref{Eq:DTLocalObserver} with the local performance metrics \eqref{Eq:DTLocalObserverPerf} (such that $G,J$ in \eqref{Eq:DTGlobalObserverPerf} are block diagonal) render the global state estimation error dynamics \eqref{Eq:DTLuenbergerObserverWithPerformance} $(Q,S,R)$-dissipative from $w(t)$ to $z(t)$ if at each subsystem $\Sigma_i,i\in\N_N$, the problem
\begin{equation}\label{Eq:Th:DTDissipativeObserver}
    \mathbb{P}_{15}: \ \ \text{Find} \ \ P_{ii},\ \ K_i\ \ \text{such that} \ \ P_{ii} > 0, \ \ \tilde{W}_{ii}>0, 
\end{equation}
is feasible, where $\tilde{W}_{ii}$ is computed from Alg. \ref{Alg:DistributedPositiveDefiniteness} (Steps: 3-16) when enforcing $W = [W_{ij}]_{i,j\in\N_N}>0$ with $W_{ij}$ given in \eqref{Eq:Th:DTDissipativeObserver2}. The local Luenberger observer parameters $L_i, \hat{A}_{i}$ and $\hat{B}_i$ \eqref{Eq:CTLocalObserver} are computed using $P_{ii}$ and $K_i$ identically to Th. \ref{Th:Observer} via \eqref{Eq:Th:Observer3}.
\end{theorem}

\begin{theorem}\label{Th:DTDissipativationUsingDOF}
(DOF Dissipativation) The DTNS \eqref{Eq:DTNSDynamics} (where $B,C,F$ are block diagonal) under $D=\0$, local DOF control \eqref{Eq:DTLocalDOFController} and local performance metrics \eqref{Eq:DTLocalDOFControllerPerf} (such that $H$ in \eqref{Eq:DTGlobalDOFControllerPerf} is block diagonal), i.e., \eqref{Eq:DTLTIUnderDOF}, is $(Q,S,R)$-dissipative (where $S$ is block diagonal) from $w(t)$ to $z(t)$ if at each subsystem $\Sigma_i,i\in\N_N$, the problem    
\begin{equation}\label{Eq:Th:DTDissipativationUsingDOF}
\begin{aligned}
    \mathbb{P}_{16}: \ \ \ \ \ \ \text{Find}& \ \ X_{ii},\ Y_{ii},\ A_{n,i},\ B_{n,i},\ C_{n,i},\ D_{n,i} \\ 
    \text{such that}& \ \ X_{ii} > 0,\ Y_{ii} > 0,\ \tilde{W}_{ii}^{(1)}>0,\ \tilde{W}_{ii}^{(2)}>0,  
\end{aligned}
\end{equation}
is feasible, where $\tilde{W}_{ii}^{(k)}$ is computed from Alg. \ref{Alg:DistributedPositiveDefiniteness} (Steps: 3-16) when enforcing $W^{(k)} = [W_{ij}^{(k)}]_{i,j\in\N_N}>0$ for $k=1,2,$ with 
$\scriptsize W_{ij}^{(1)} = \bm{Y_{ii}e_{ij} & e_{ij} \\ e_{ij} & X_{ij}e_{ij}}$ and $W_{ij}^{(2)}$ given in  \eqref{Eq:Th:DTDissipativationUsingDOF2}. The local DOF controller parameters $A_{c,i},B_{c,i},C_{c,i}$ and $D_{c,i}$ are computed using $X_{ii},Y_{ii},A_{n,i},B_{n,i},C_{n,i}$ and $D_{n,i}$, identically to Th. \ref{Th:StabilizationUnderDOF} via \eqref{Eq:Th:StabilizationUnderDOF3} and \eqref{Eq:Th:StabilizationUnderDOF4}.
\end{theorem}

We emphasize that local problems $\mathbb{P}_9$-$\mathbb{P}_{16}$ stated respectively in the theorems \ref{Th:DTLTIStability}-\ref{Th:DTDissipativationUsingDOF} are LMI problems due to the applicability of Lm. \ref{Lm:TwoByTwoBlockMatrixPDF} to simplify the matrix inequality $\tilde{W}_{ii}>0$ in Alg. \ref{Alg:DistributedPositiveDefiniteness} (Step 14). Consequently, such problems can be solved conveniently and efficiently using readily available LMI software toolboxes \cite{Boyd1994}. Moreover, based on the remaining propositions provided in Sec. \ref{Sec:BasicsOfDTLTISystems}, a similar set of theorems can be proposed providing respective decentralized techniques (as summarized in the respective remaining rows in Tab. \ref{Tab:DT-LTILocalResultsSummary}).

\begin{figure*}[!b]
    \centering
    \hrulefill
    \begin{equation}\label{Eq:Th:DTDissipativationUnderFSF2}
    W_{ij} = 
    \bm{M_{ii}e_{ij} & M_{ii}C_{ii}^\T S_{ij} & M_{ii}A_{ji}^\T + L_{ji}^\T B_{jj}^\T & M_{ii}C_{ii}^\T e_{ij} \\ \star & \H_e(F_{ii}^\T S_{ij}) + R_{ij} & E_{ji}^\T & F_{ii}^\T e_{ij} \\ 
    \star & \star & M_{ii}e_{ij} & \0 \\ \star & \star & \star & -Q_{ii}^{-1}e_{ij} } 
    \end{equation}
    \begin{equation} \label{Eq:Th:DTDissipativeObserver2}
    W_{ij} = \bm{P_{ii}e_{ij} & G_{ii}^\T S_{ij} & A_{ji}^\T P_{jj} - C_{ii}^\T K_{ji}^\T & G_{ii}^\T e_{ij} \\ 
    \star & \H_e(J_{ii}^\T S_{ij}) + R_{ij} & E_{ji}^\T P_{jj} - F_{ii}^\T K_{ji}^\T & J_{ii}^\T e_{ij} \\
    \star & \star & P_{ii}e_{ij} & \0 \\ \star & \star & \star & -Q_{ii}^{-1}e_{ij}}
    \end{equation}
    \begin{equation}\label{Eq:Th:DTDissipativationUsingDOF2}
    \scriptsize
    W_{ij}^{(2)} = \bm{Y_{ii}e_{ij} & e_{ij} & (Y_{ii}G_{ji}^\T+C_{n,ji}^\T H_{jj}^\T)S_{jj} & Y_{ii}A_{ji}^\T + C_{n,ji}^\T B_{jj}^\T & A_{n,ji}^\T & Y_{ii}G_{ji}^\T+C_{n,ji}^\T H_{jj}^\T\\
    \star & X_{ii}e_{ij} & (G_{ji}^\T+C_{ii}^\T D_{n,ji}^\T H_{jj}^\T)S_{jj} & A_{ji}^\T+C_{ii}^\T D_{n,ji}^\T B_{jj}^\T & A_{ji}^\T X_{jj}+C_{ii}^\T B_{n,ji}^\T & G_{ji}^\T+C_{ii}^\T D_{n,ji}^\T H_{jj}^\T \\ 
    \star & \star & \H_s((J_{ji}^\T+F_{ii}^\T D_{n,ji}^\T H_{jj}^\T)S_{jj})+R_{ij} & E_{ji}^\T+F_{ii}^\T D_{n,ji}^\T B_{jj}^\T & E_{ji}^\T X_{jj}+F_{ii}^\T B_{n,ji}^\T & J_{ji}^\T+F_{ii}^\T D_{n,ji}^\T H_{jj}^\T
    \\\star & \star & \star & Y_{ii}e_{ij} & e_{ij} & \0 \\ \star & \star & \star & \star & X_{ii}e_{ij} & \0 \\ 
    \star  & \star & \star  & \star & \star & -Q_{ii}^{-1}e_{ij}}
    \end{equation}
\end{figure*}

\section{Simulation Results}
\label{Sec:SimulationResults}

\begin{table*}[!b]
\hrulefill
\begin{equation}\label{Eq:ExampleNetworkedSystem}
\resizebox{0.95\textwidth}{!}{$
\begin{aligned}
\Sigma_1 =&\  
\begin{cases}
\dot{x}_1 = \bm{ 0.198 & 3.412 \\ -3.412 & 0.198} x_1 + \bm{ -0.114 & -0.038 \\ -0.038 & -0.073} x_2 + \bm{ -0.060 & -1.032 \\ 1.032 & -0.060} x_4 + \bm{0.000 \\ 0.905} u_1 + \bm{ 0.000 \\ -0.013} w_1 + \bm{ 0.000 \\ 0.003} w_2 + \bm{ -0.001 \\ 0.006} w_4,\\
y_1 = \bm{1.114 & -2.429} x_1 + \bm{0.000} w_1, \hfill 
z_1 = \bm{1.000 & 1.000} x_1 + \bm{1.000} u_1 + \bm{1.000} w_1.\\
\end{cases}
\\
\Sigma_2 =&\  
\begin{cases}
\dot{x}_2 = \bm{ -0.000 & -0.001 \\ -0.001 & -0.194} x_1 + \bm{ 1.547 & 3.164 \\ -3.164 & 1.547} x_2 + \bm{ -0.258 & -0.008 \\ -0.008 & -0.204} x_4 + \bm{-0.902 \\ 0.000} u_2 + \bm{ -0.004 \\ -0.010} w_1 + \bm{ -0.000 \\ 0.021} w_2 + \bm{ 0.003 \\ 0.002} w_4,\\
y_2 = \bm{0.000 & 1.062} x_2 + \bm{0.002} w_2, \hfill 
z_2 = \bm{1.000 & 1.000} x_2 + \bm{1.000} u_2 + \bm{1.000} w_2.\\
\end{cases}
\\
\Sigma_3 =&\  
\begin{cases}
\dot{x}_3 = \bm{ -0.232 & -0.070 \\ -0.070 & -0.158} x_1 + \bm{ -0.096 & -0.062 \\ -0.062 & -0.085} x_2 + \bm{ 10.791 & 5.354 \\ 5.354 & 3.134} x_3 + \bm{ -0.074 & -0.384 \\ 0.384 & -0.074} x_4 + \bm{-0.324 \\ -1.406} u_3 + \bm{ 0.000 \\ -0.007} w_1 \\ 
+ \bm{ -0.002 \\ -0.004} w_2 + \bm{ 0.000 \\ -0.008} w_3 + \bm{ -0.003 \\ -0.001} w_4, \ \ \ \ 
y_3 = \bm{1.052 & 0.759} x_3 + \bm{-0.011} w_3, \ \ \ \ 
z_3 = \bm{1.000 & 1.000} x_3 \bm{1.000} u_3 \bm{1.000} w_3.\\
\end{cases}
\\
\Sigma_4 =&\  
\begin{cases}
\dot{x}_4 = \bm{ -0.180 & -0.066 \\ -0.066 & -0.078} x_1 + \bm{ 1.669 & 2.302 \\ 2.302 & 3.175} x_4 + \bm{0.000 \\ 0.998} u_4 + \bm{ -0.006 \\ 0.003} w_1 + \bm{ -0.018 \\ 0.008} w_4,\\
y_4 = \bm{0.629 & 0.000} x_4 + \bm{0.000} w_4, \hfill 
z_4 = \bm{1.000 & 1.000} x_4 \bm{1.000} u_4 \bm{1.000} w_4.\\
\end{cases}
\\
\Sigma_5 =&\  
\begin{cases}
\dot{x}_5 = \bm{ -0.059 & 0.033 \\ 0.033 & -0.057} x_4 + \bm{ 0.058 & 0.250 \\ 0.250 & 1.074} x_5 + \bm{0.870 \\ -1.461} u_5 + \bm{ -0.001 \\ -0.004} w_4 + \bm{ -0.006 \\ -0.000} w_5,\\
y_5 = \bm{-0.552 & -0.750} x_5 + \bm{0.000} w_5, \hfill 
z_5 = \bm{1.000 & 1.000} x_5 \bm{1.000} u_5 \bm{1.000} w_5.\\
\end{cases}
\end{aligned}
$}
\end{equation}
\end{table*}

\begin{table*}[!b]
\hrulefill
\begin{equation}\label{Eq:StabilizingControllerGains}
\begin{aligned}
K_1 =&\ \left\{K_{11}=\bm{-0.296 & -1.523}\right\},\ \ \ \ 
K_2 = \left\{K_{21}=\bm{-0.168 & -0.043}, K_{12}=\bm{-0.173 & 0.574}, K_{22}=\bm{5.369 & -3.916}\right\},\\ 
K_3 =&\ \left\{K_{31}=\bm{-0.083 & -0.118}, K_{32}=\bm{-0.057 & -0.067}, K_{33}=\bm{32.449 & 5.429}\right\},\\
K_4 =&\ \left\{K_{41}=\bm{1.023 & 0.279}, K_{14}=\bm{-0.973 & 0.114}, K_{24}=\bm{-0.234 & 0.041}, K_{34}=\bm{0.248 & -0.109}, K_{44}=\bm{-12.974 & -8.081}\right\},\\ 
K_5 =&\ \left\{K_{54}=\bm{0.035 & -0.039}, K_{55}=\bm{1.753 & 2.803}\right\}.
\end{aligned}
\end{equation}
\begin{equation}\label{Eq:StabilizingObserverGains}
\begin{aligned}
L_1 =&\ \left\{L_{11}=\bm{0.306 \\ -0.427}\right\},\ \ \ \ 
L_2 = \left\{L_{21}=\bm{0.060 \\ -0.008}, L_{12}=\bm{-0.042 \\ -0.207}, L_{22}=\bm{-3.403 & 4.596}\right\},\\
L_3 =&\ \left\{L_{31}=\bm{-0.012 \\ 0.043}, L_{32}=\bm{-0.059 \\ -0.080}, L_{33}=\bm{12.219 & 3.616}\right\},\\
L_4 =&\ \left\{L_{41}=\bm{-0.577 \\ -1.591}, L_{14}=\bm{-0.640 \\ 1.395}, L_{24}=\bm{-0.409 \\ -0.012}, L_{34}=\bm{-0.117 \\ 0.610}, L_{44}=\bm{35.602 & 97.048}\right\},\\
L_5 =&\ \left\{L_{54}=\bm{-0.094 \\ 0.053}, L_{55}=\bm{0.589 & -3.585}\right\}.
\end{aligned}
\end{equation}
\end{table*}

\begin{figure}[!b]
    \centering
    \begin{subfigure}[h]{\columnwidth}
        \centering
        \includegraphics[width=0.9\textwidth]{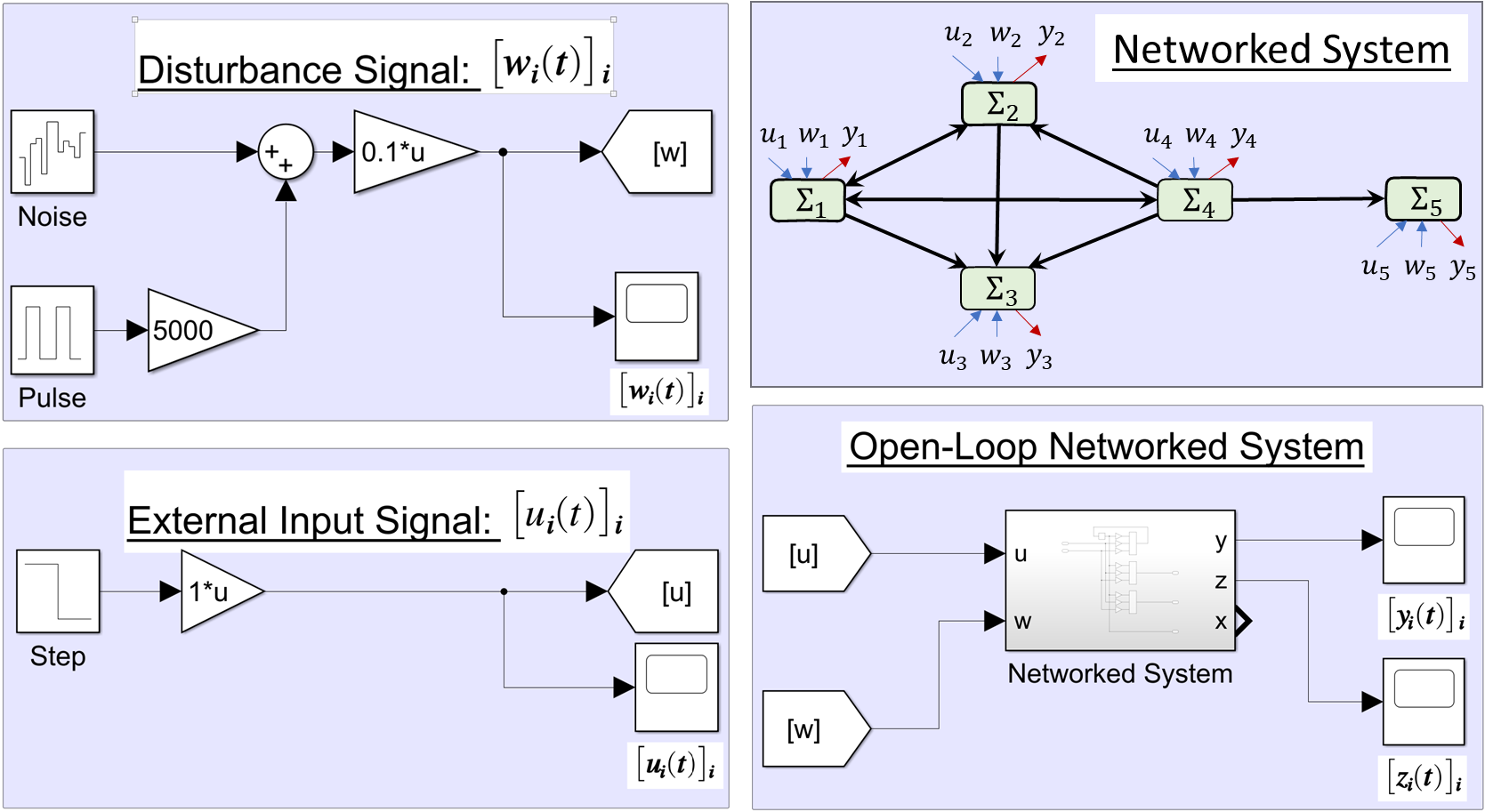}
        \caption{Considered networked system \eqref{Eq:ExampleNetworkedSystem} and simulation environment.}
        \label{Fig:SESetup1}    
    \end{subfigure}
    \begin{subfigure}[h]{\columnwidth}
        \centering
        \includegraphics[width=\textwidth]{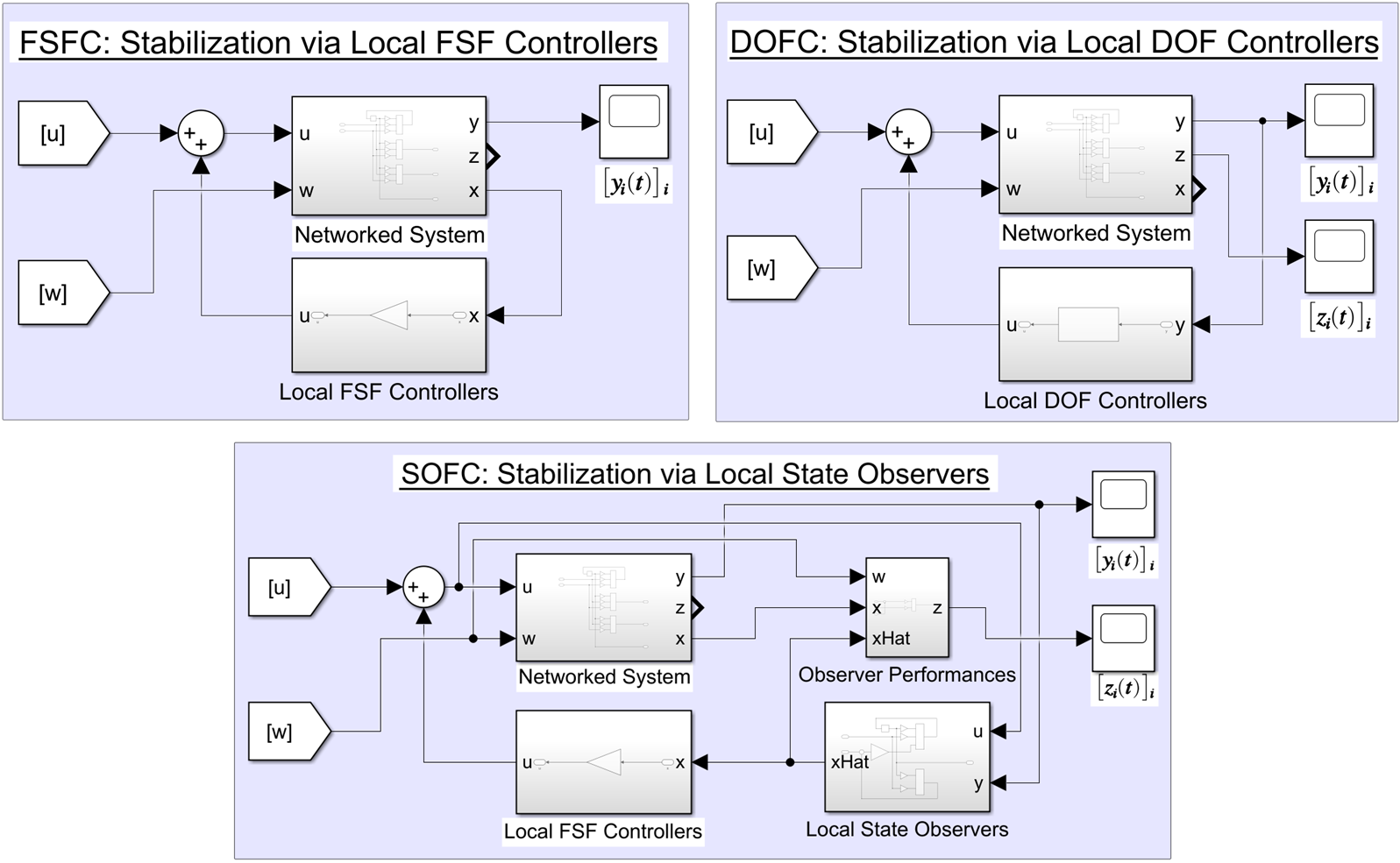}
        \caption{Evaluated three distributed controller configurations for the stabilization of the networked system: (1) \textbf{FSFC}: via local full-state feedback control, (2) \textbf{SOFC}: via local state observer based feedback control, and (3) \textbf{DOFC}: via local dynamic output feedback control.}
        \label{Fig:SESetup2}
    \end{subfigure}
    \caption{Considered simulation example setup: (a) the networked system and (b) the controller configurations.}
    \label{Fig:SESetup}
\end{figure}

\begin{figure}[!b]
    \centering
    \begin{subfigure}[h]{0.48\columnwidth}
        \centering
        \includegraphics[width=\textwidth]{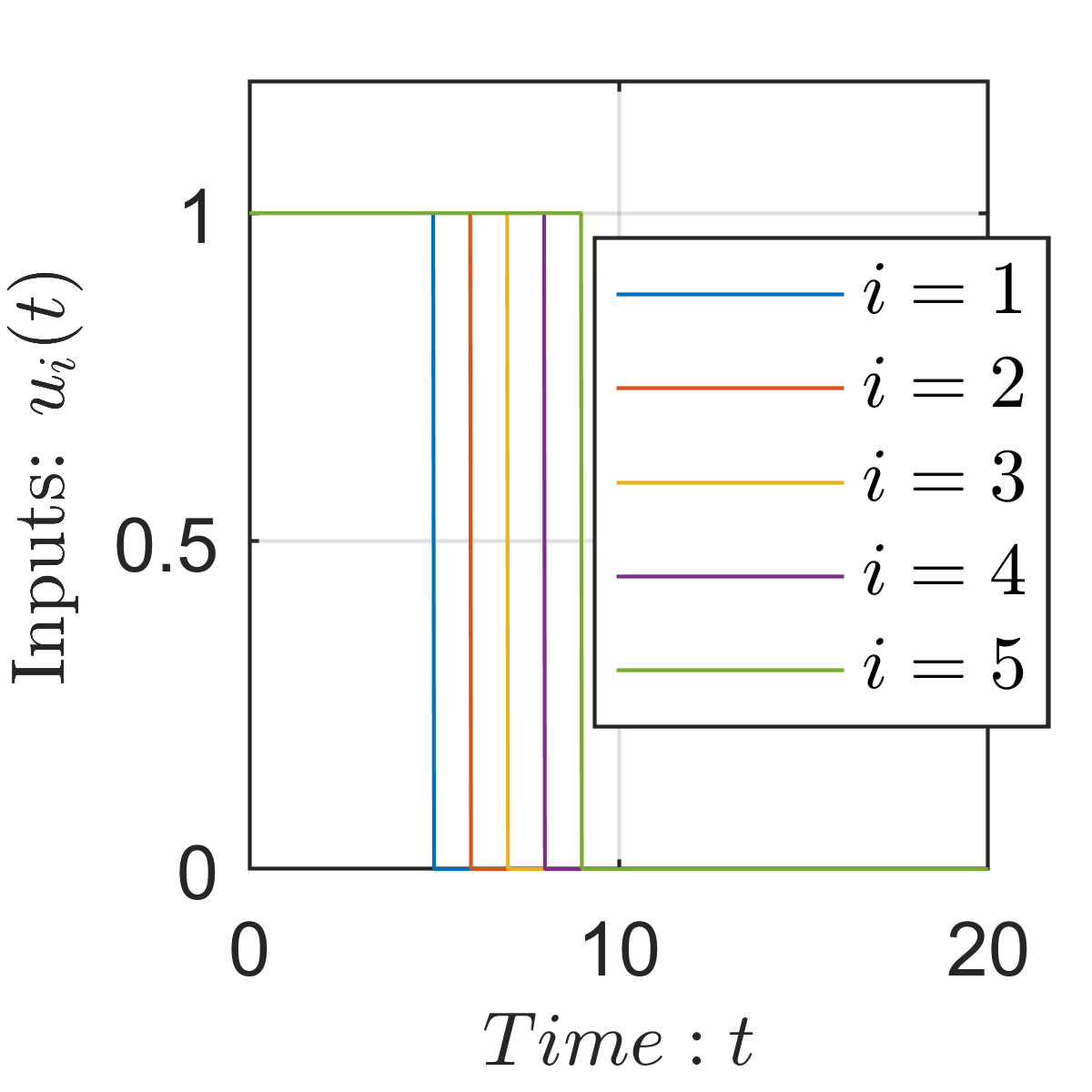}
        \caption{}
        \label{Fig:SEOpenLoopu}    
    \end{subfigure}
    \begin{subfigure}[h]{0.48\columnwidth}
        \centering
        \includegraphics[width=\textwidth]{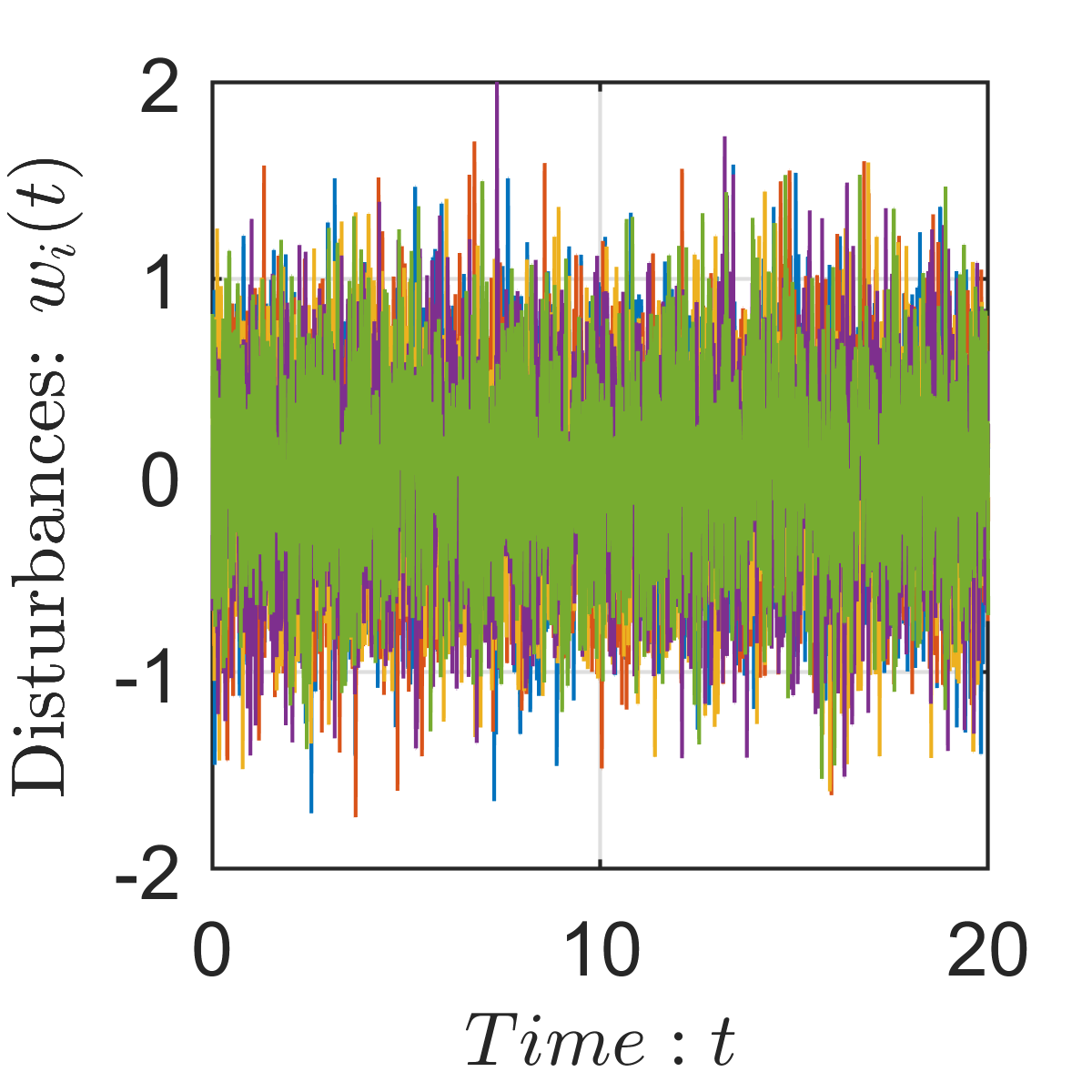}
        \caption{}
        \label{Fig:SEOpenLoopw}    
    \end{subfigure}
    \begin{subfigure}[h]{0.48 \columnwidth}
        \centering
        \includegraphics[width=\textwidth]{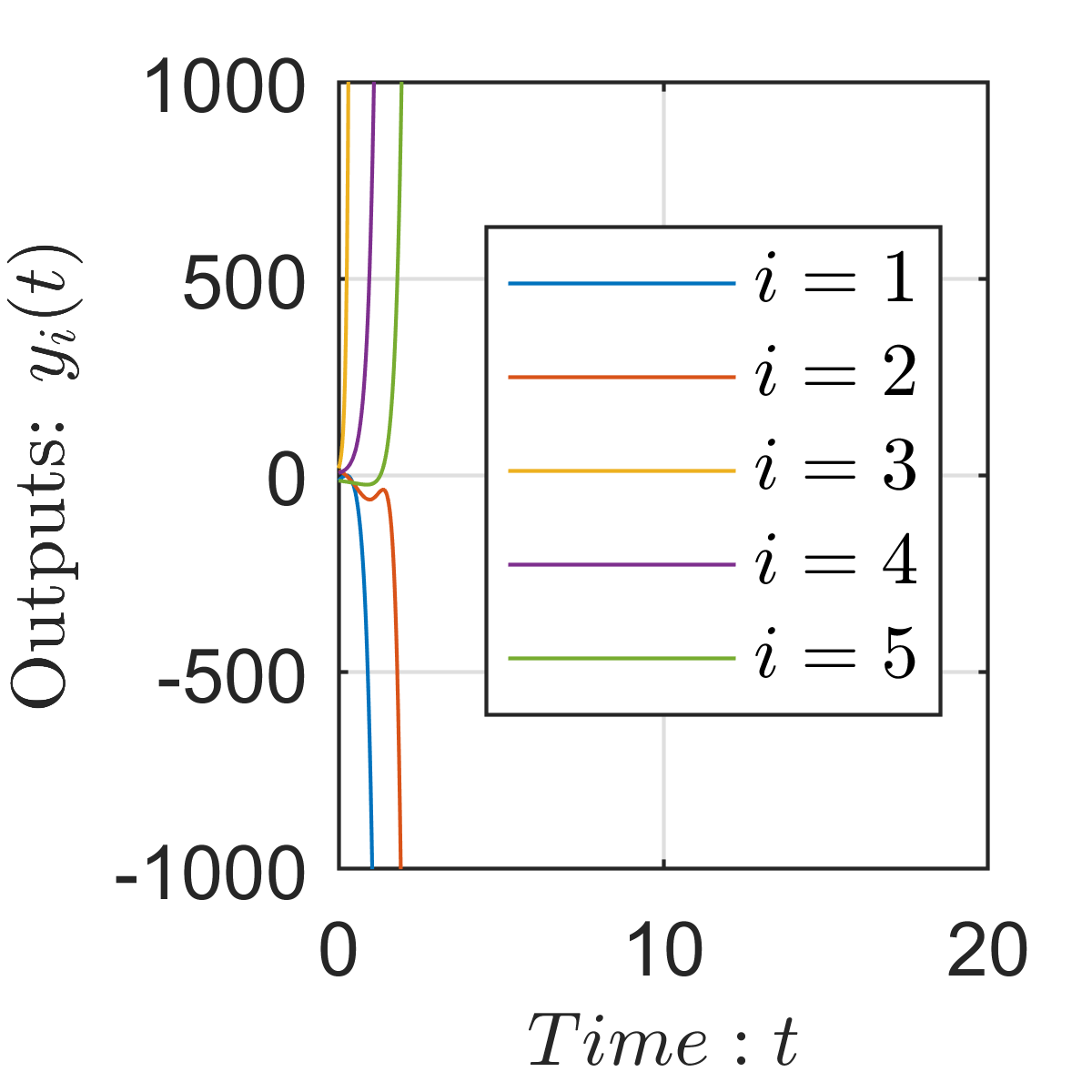}
        \caption{}
        \label{Fig:SEOpenLoopy}
    \end{subfigure}
    \begin{subfigure}[h]{0.48 \columnwidth}
        \centering
        \includegraphics[width=\textwidth]{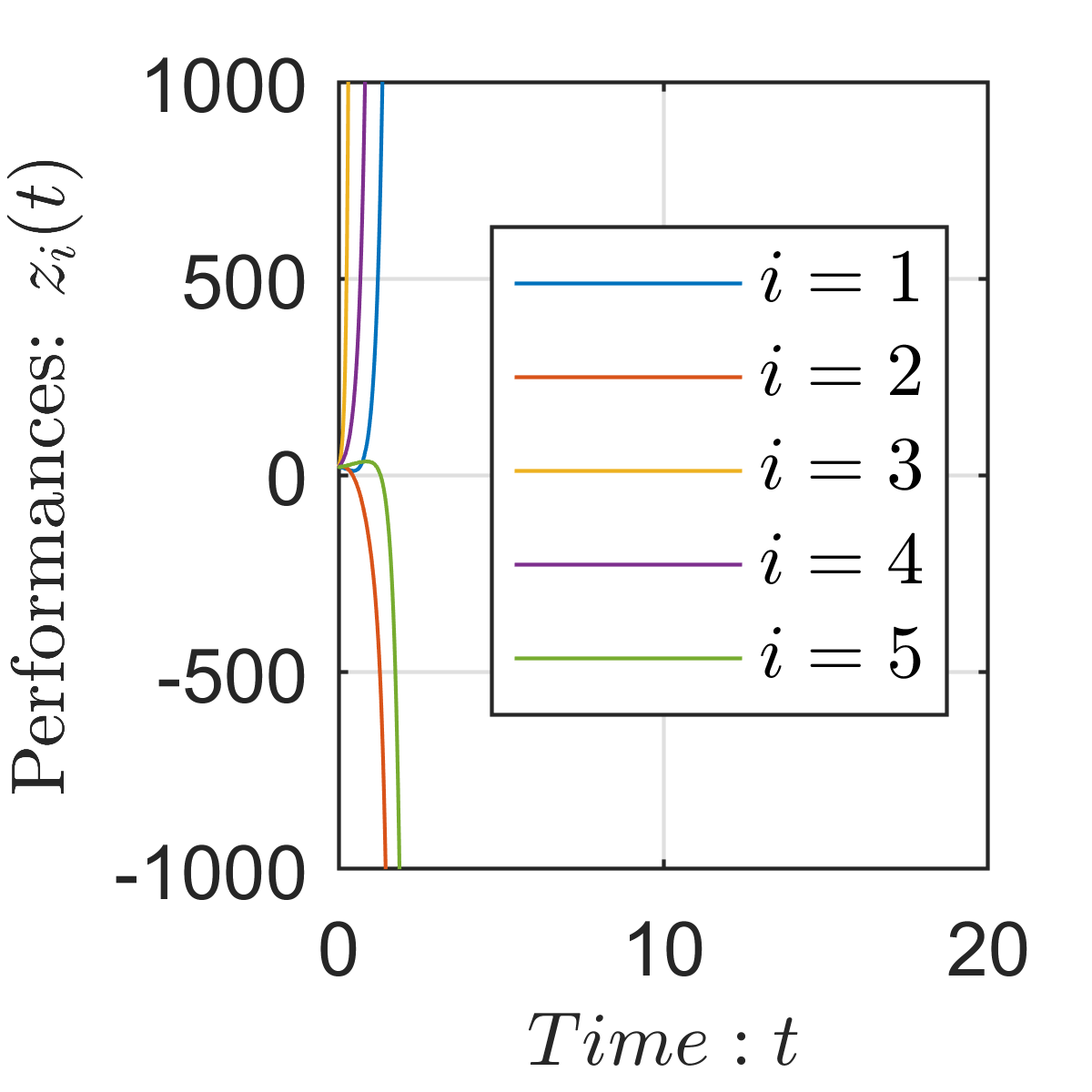}
        \caption{}
        \label{Fig:SEOpenLoopz}
    \end{subfigure}
    \caption{The used signals as subsystem: (a) inputs $[u_i(t)]_{i\in\N_N}$ and (b) disturbances $[w_i(t)]_{i\in\N_N}$, and the resulting open-loop signals for subsystem: (c) outputs $[y_i(t)]_{i\in\N_N}$ and (d) performances $[z_i(t)]_{i\in\N_N}$.}
    \label{Fig:SEOpenLoop}
\end{figure}

\begin{figure}[!hb]
    \centering
    \begin{subfigure}[h]{0.48\columnwidth}
        \centering
        \includegraphics[width=\textwidth]{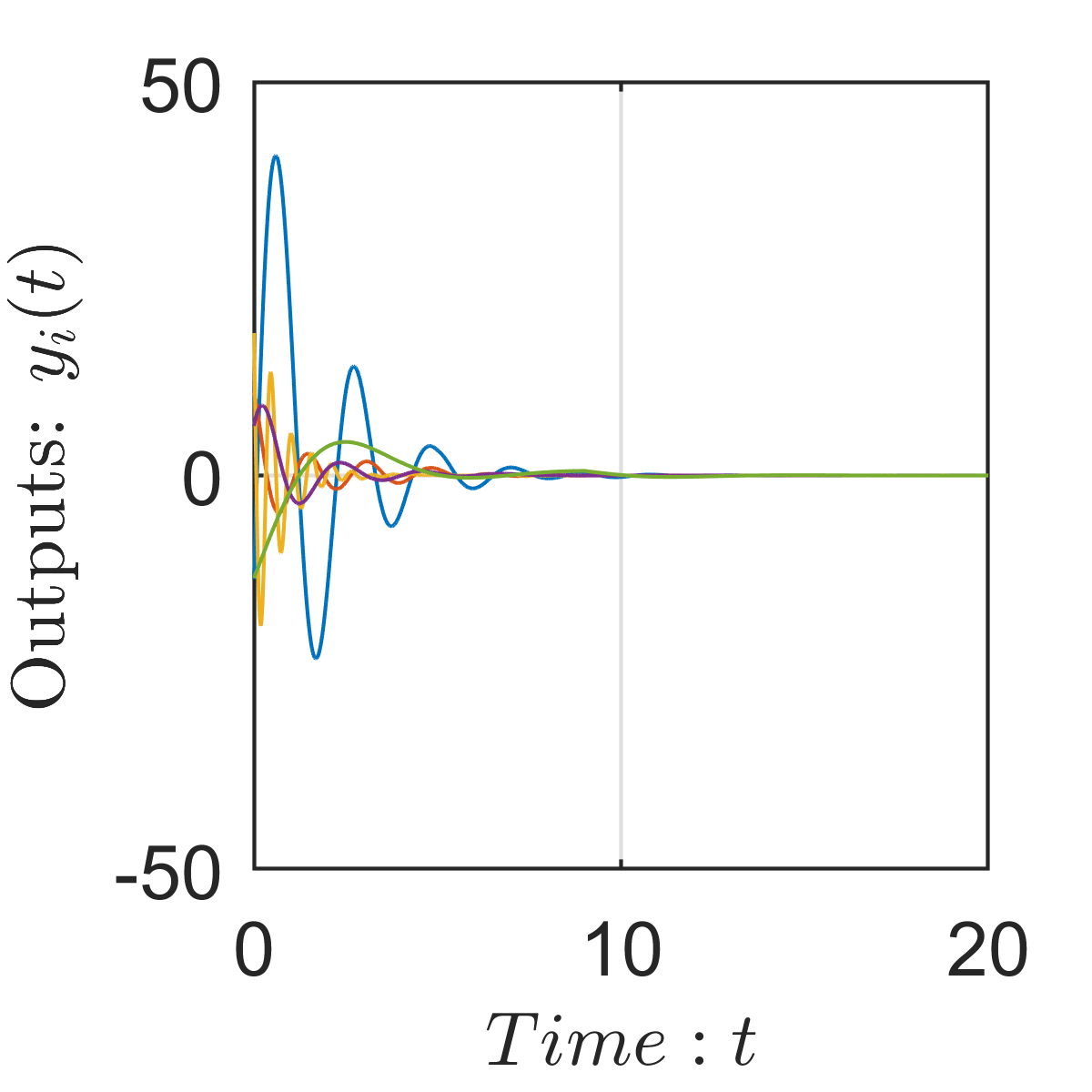}
        \caption{FSFC (centrally derived) \\ \centering MAO = 1.121}
        \label{Fig:}    
    \end{subfigure}
    \begin{subfigure}[h]{0.48\columnwidth}
        \centering
        \includegraphics[width=\textwidth]{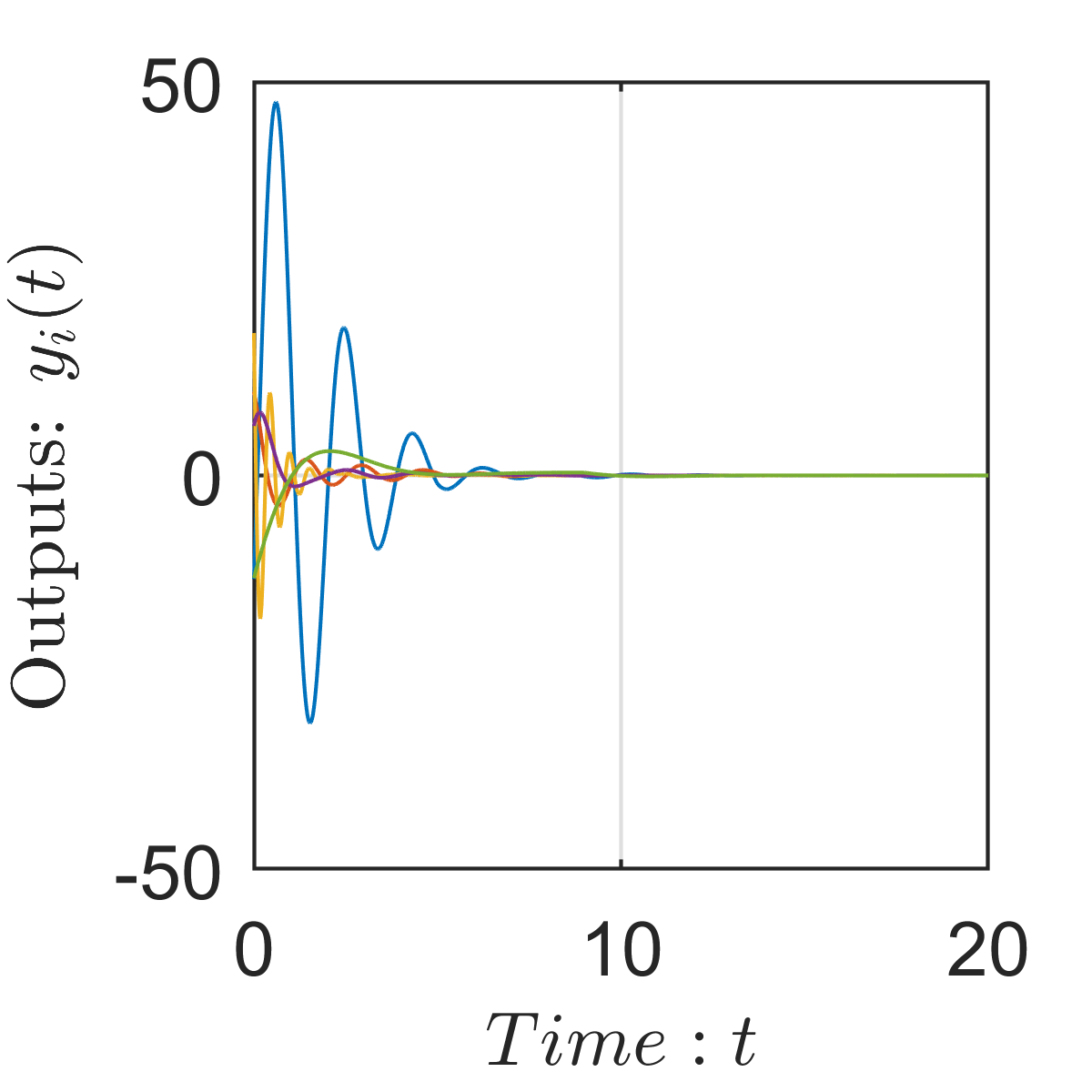}
        \caption{FSFC (decentrally derived) \\ \centering MAO = 1.063 (Impr.:+5.189\%)}
        \label{Fig:}    
    \end{subfigure}
    \begin{subfigure}[h]{0.48 \columnwidth}
        \centering
        \includegraphics[width=\textwidth]{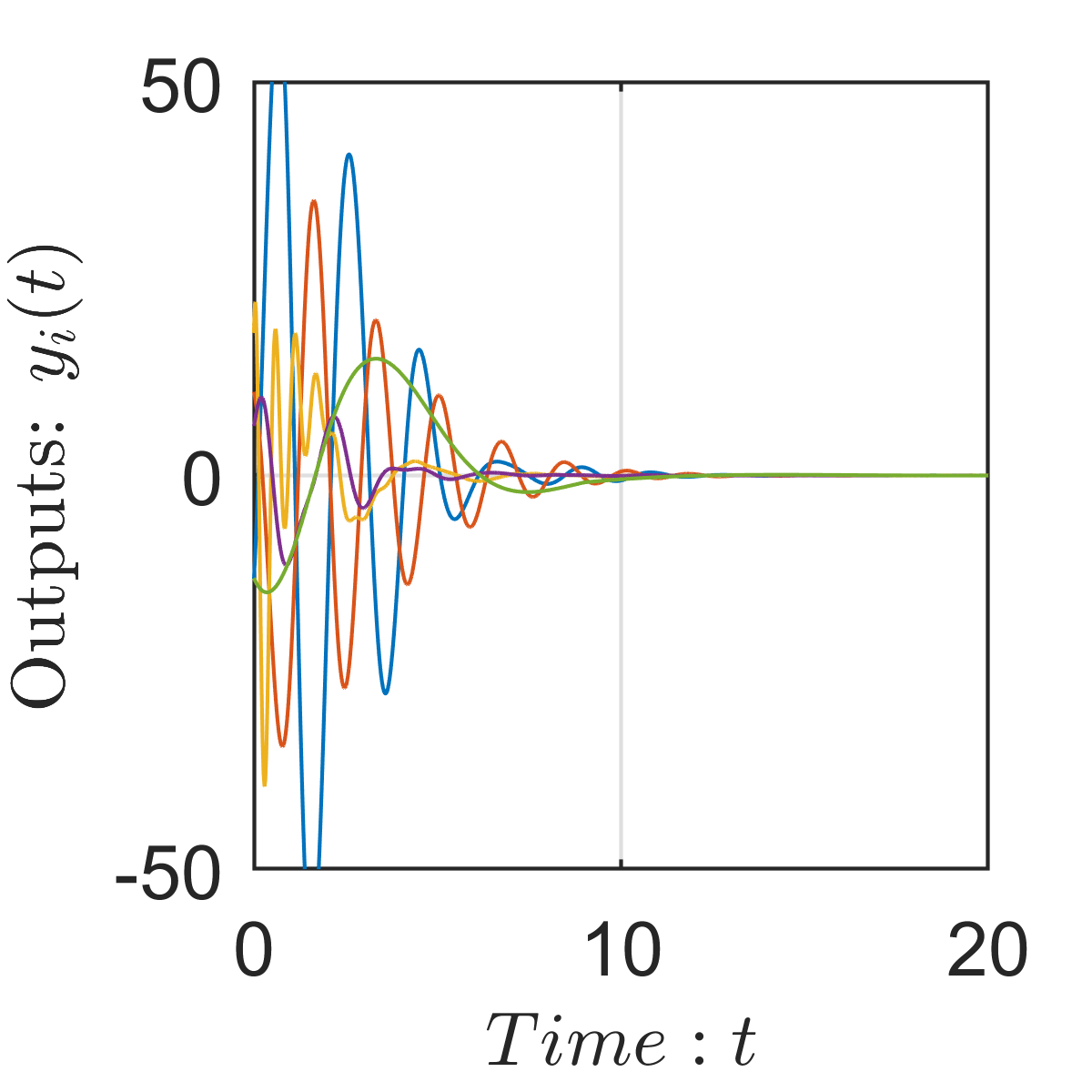}
        \caption{SOFC (centrally derived) \\ \centering MAO = 3.413}
        \label{Fig:}
    \end{subfigure}
    \begin{subfigure}[h]{0.48 \columnwidth}
        \centering
        \includegraphics[width=\textwidth]{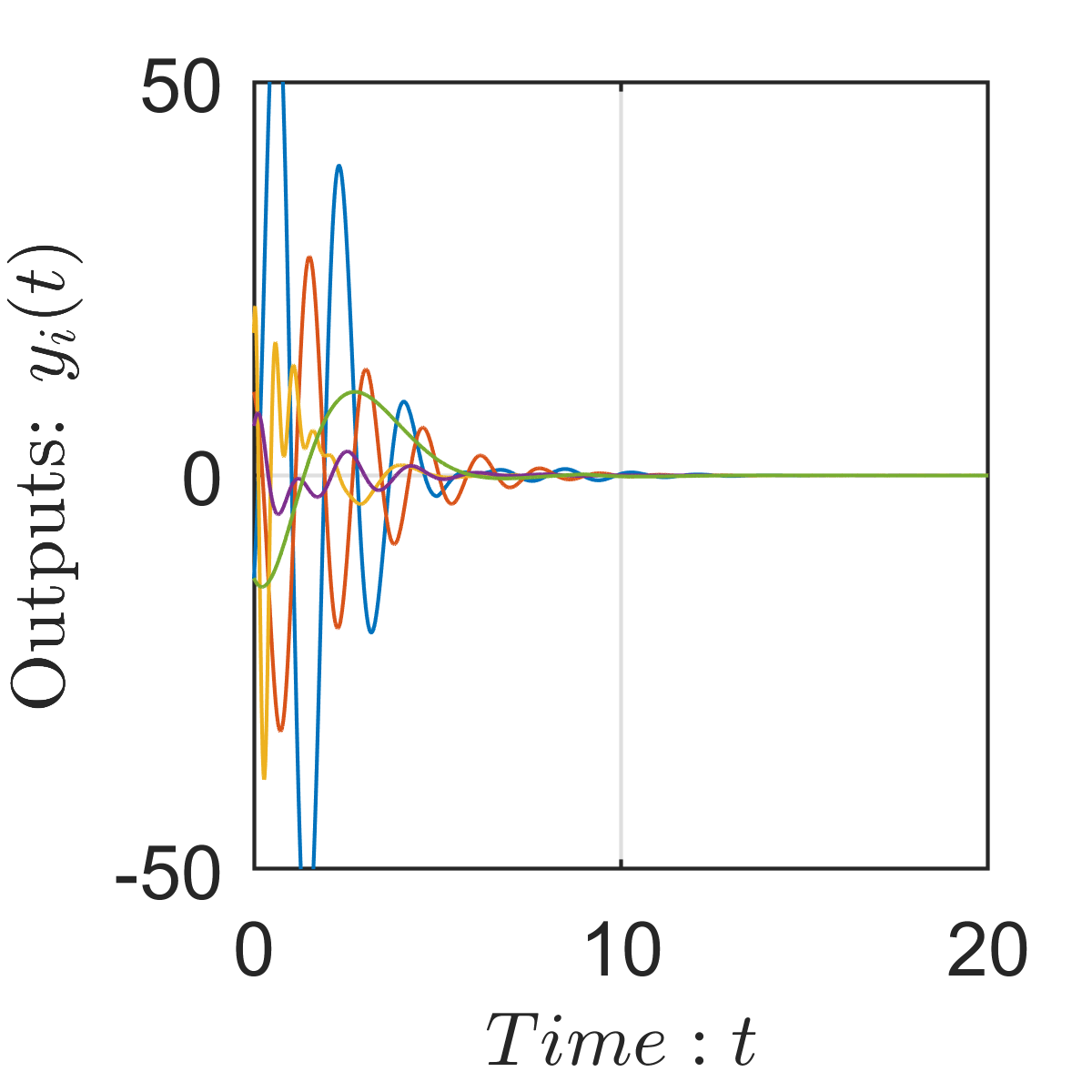}
        \caption{SOFC (decentrally derived) \\ \centering MAO = 2.498 (Impr.:+26.81\%)}
        \label{Fig:}
    \end{subfigure}
    \begin{subfigure}[h]{0.48\columnwidth}
        \centering
        \includegraphics[width=\textwidth]{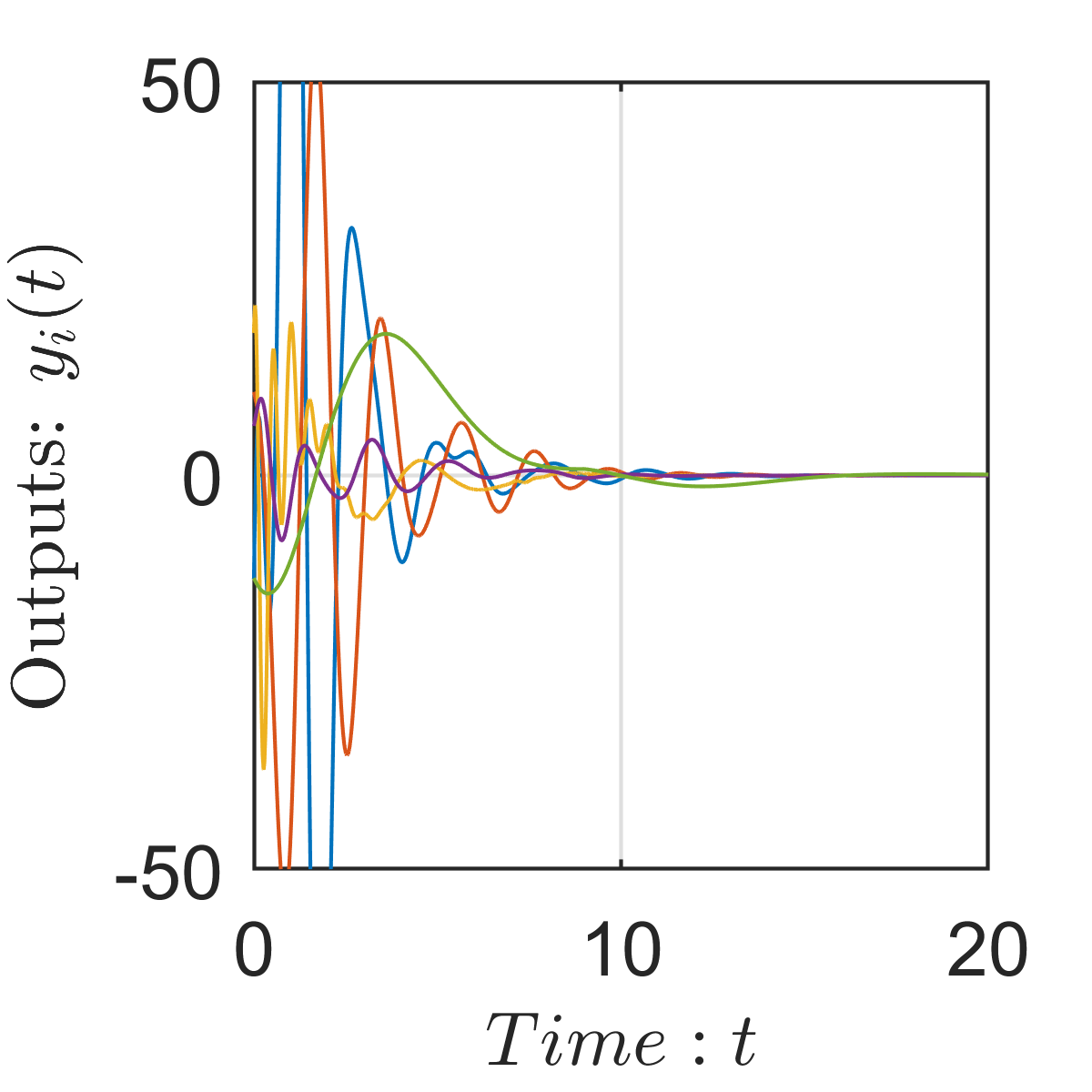}
        \caption{DOFC (centrally derived) \\ \centering MAO = 4.227}
        \label{Fig:}    
    \end{subfigure}
    \begin{subfigure}[h]{0.48\columnwidth}
        \centering
        \includegraphics[width=\textwidth]{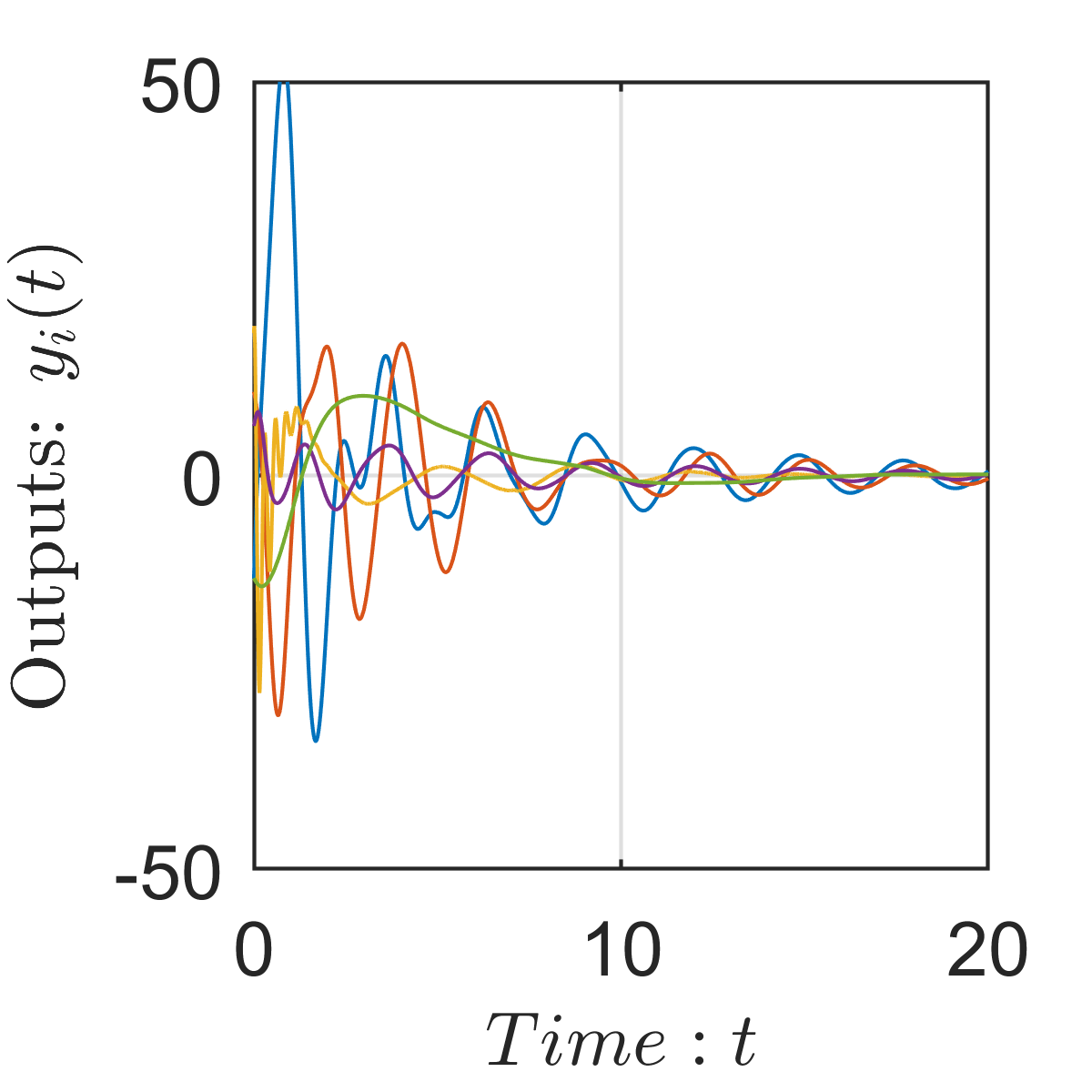}
        \caption{DOFC (decentrally derived) \\ \centering MAO = 3.107 (Impr.:+26.49\%)}
        \label{Fig:}    
    \end{subfigure}
    \caption{\textbf{The effect of decentralization}: Comparison of output trajectories $[y_i(t)]_{i\in\N_N}$ obtained under: \textbf{centrally (left) vs. decentrally (right)} synthesized stabilizing controllers - for the distributed controller configurations: FSFC (a,b), SOFC (c,d) and DOFC (e,f) (see Fig. \ref{Fig:SESetup2}). The observed mean absolute output (MAO) values are given in subcaptions along with the corresponding percentage improvement values.}
    \label{Fig:SEDecentralization}
\end{figure}

\begin{figure}[!ht]
    \centering
    \begin{subfigure}[b]{0.48 \columnwidth}
        \centering
        \includegraphics[width=\textwidth]{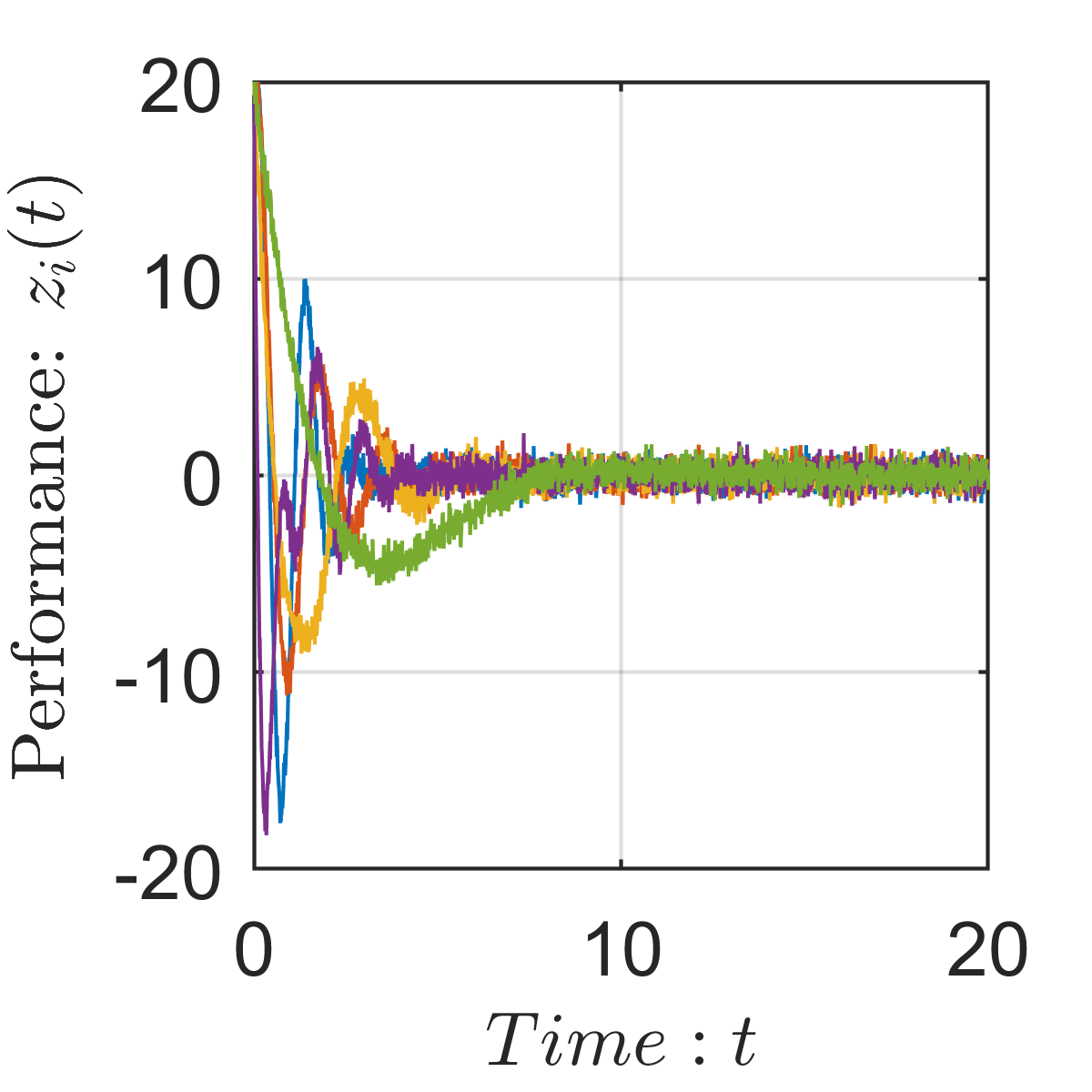}
        \caption{SOFC (derived: \\ centrally to stabilize) \\ \centering MAP = 1.327}
        \label{Fig:}
    \end{subfigure}
    \begin{subfigure}[b]{0.48 \columnwidth}
        \centering
        \includegraphics[width=\textwidth]{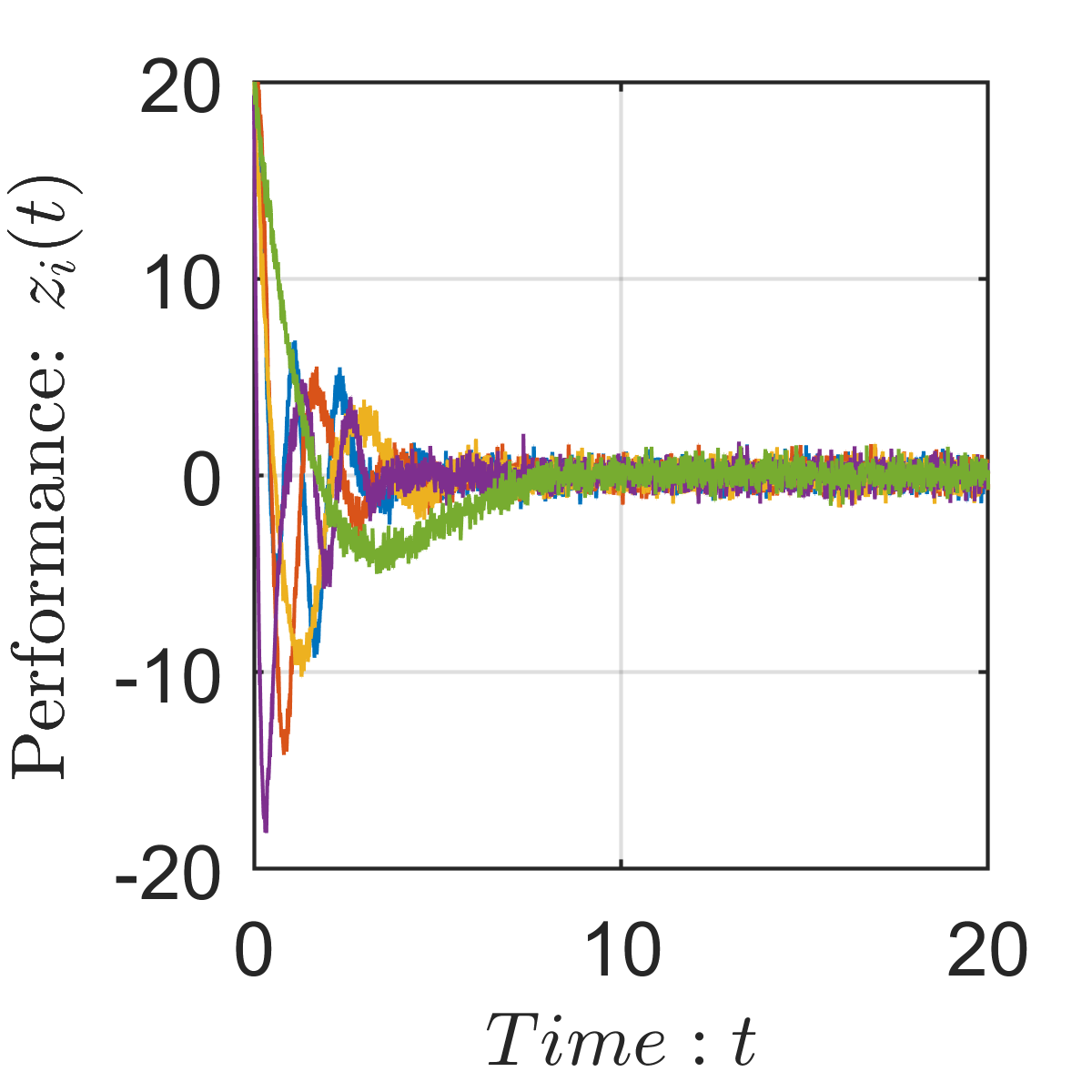}
        \caption{SOFC (derived: \\ centrally to dissipativate) \\ \centering MAP = 1.246 (Impr.:+6.112\%)}
        \label{Fig:}
    \end{subfigure}
    \begin{subfigure}[b]{0.48 \columnwidth}
        \centering
        \includegraphics[width=\textwidth]{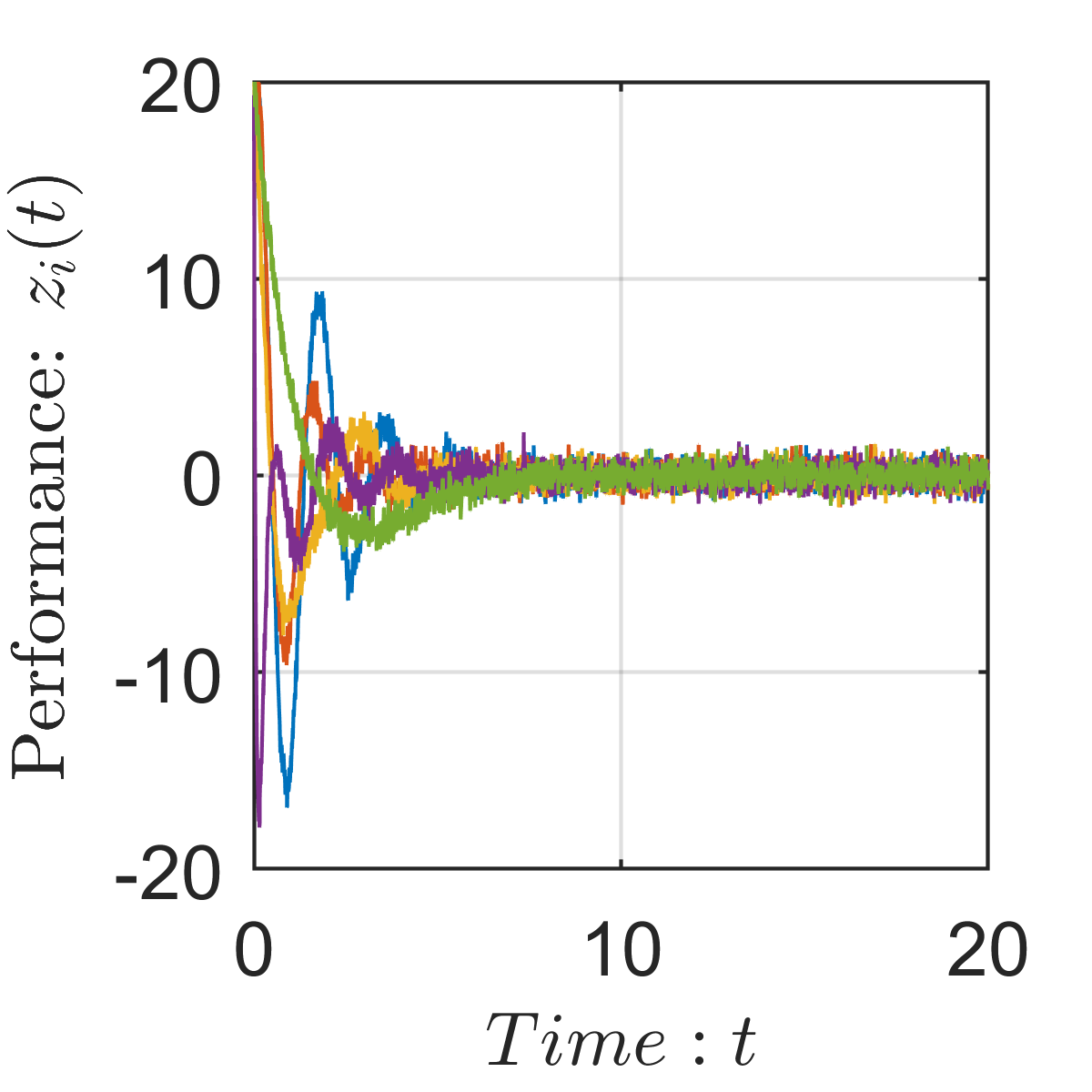}
        \caption{SOFC (derived: \\ decentrally to stabilize) \\ \centering MAP = 1.155}
        \label{Fig:}
    \end{subfigure}
    \begin{subfigure}[b]{0.48 \columnwidth}
        \centering
        \includegraphics[width=\textwidth]{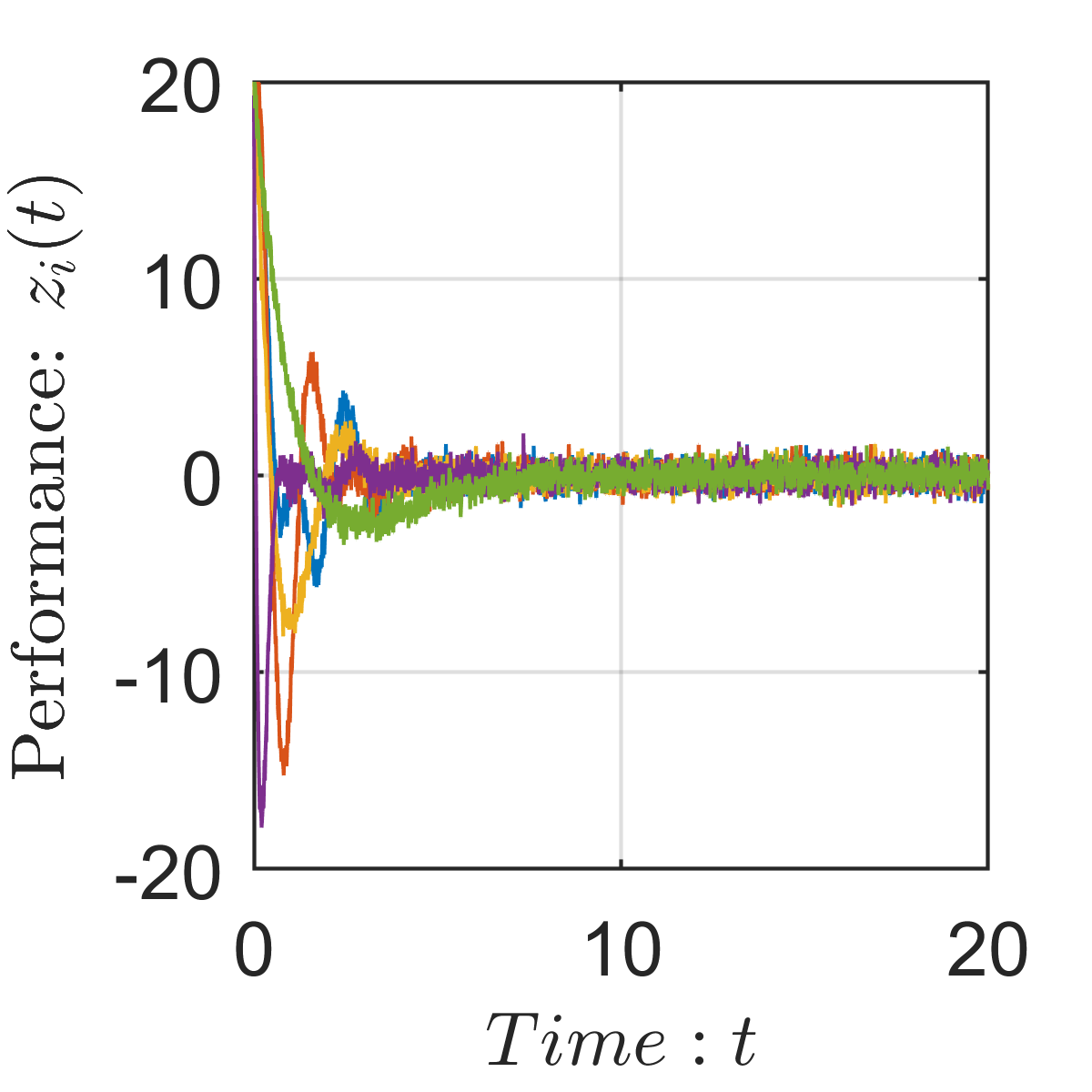}
        \caption{SOFC (derived: \\ decentrally to dissipativate) \\ \centering MAP = 1.015 (Impr.:+12.10\%)}
        \label{Fig:}
    \end{subfigure}
    \begin{subfigure}[b]{0.48\columnwidth}
        \centering
        \includegraphics[width=0.9\textwidth]{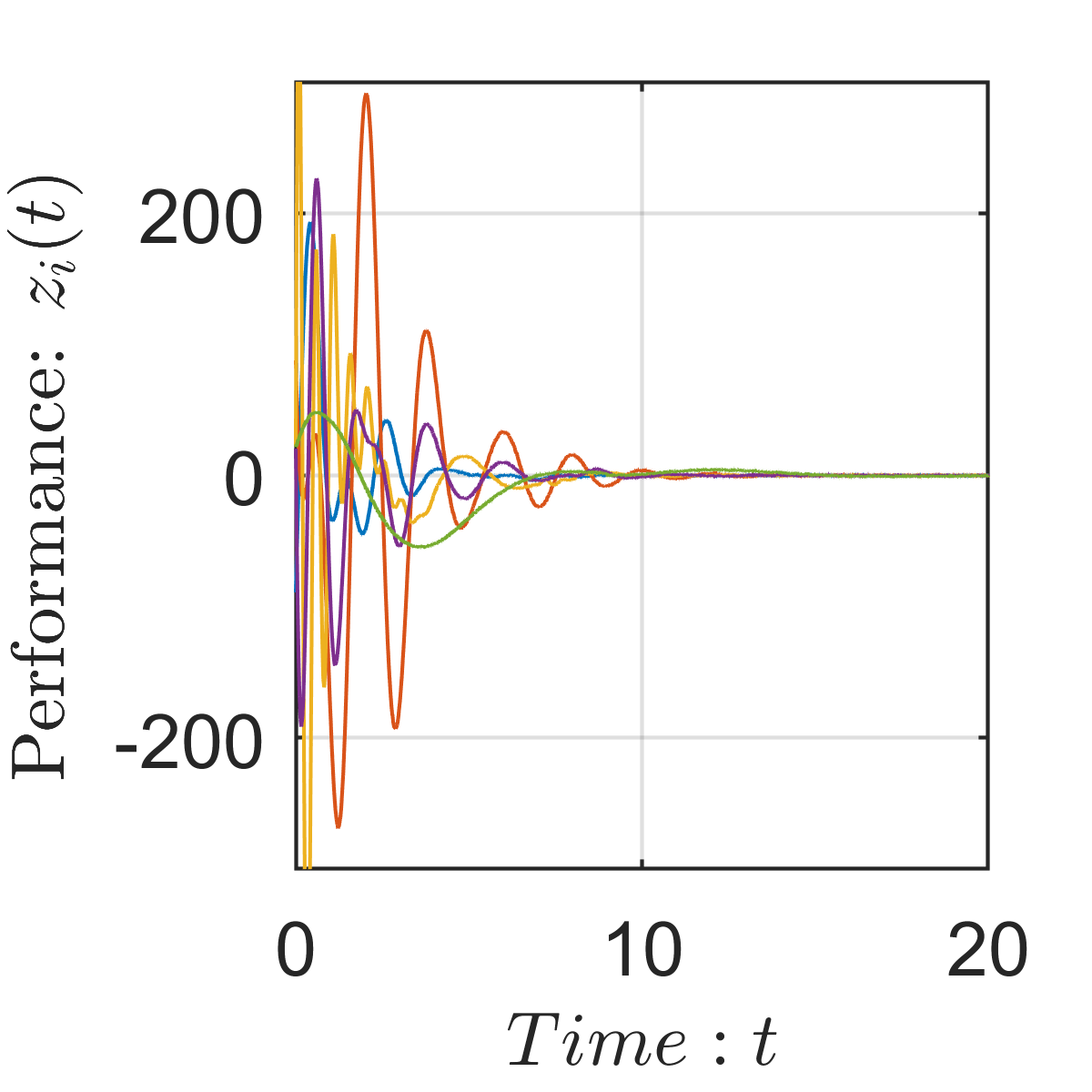}
        \caption{DOFC (derived: \\ centrally to stabilize) \\ \centering MAP = 16.01}
        \label{Fig:}    
    \end{subfigure}
    \begin{subfigure}[b]{0.48\columnwidth}
        \centering
        \includegraphics[width=0.9\textwidth]{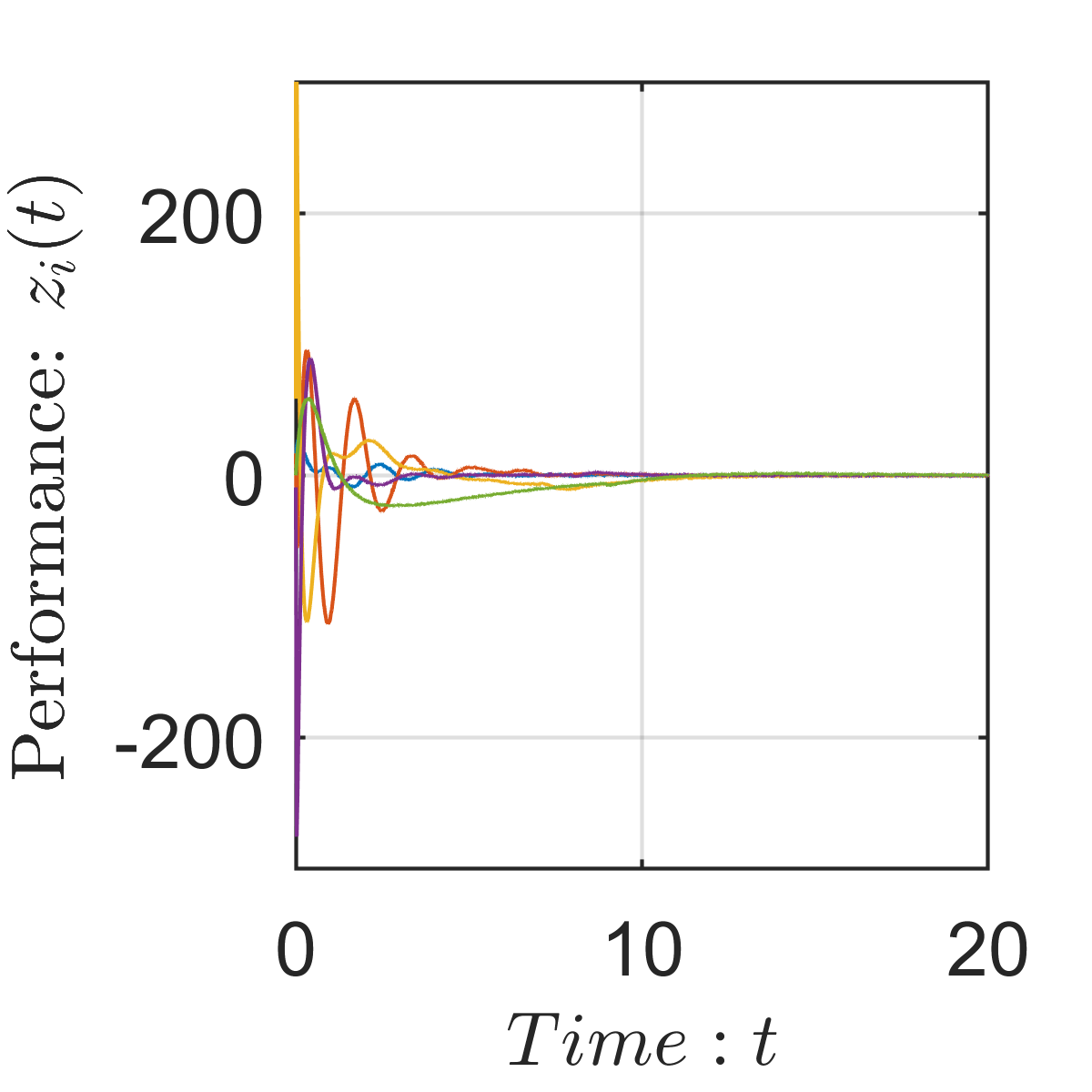}
        \caption{DOFC (derived: \\ centrally to dissipativate) \\ \centering MAP = 5.966 (Impr.:+62.74\%)}
        \label{Fig:}    
    \end{subfigure}
    \begin{subfigure}[b]{0.48\columnwidth}
        \centering
        \includegraphics[width=0.9\textwidth]{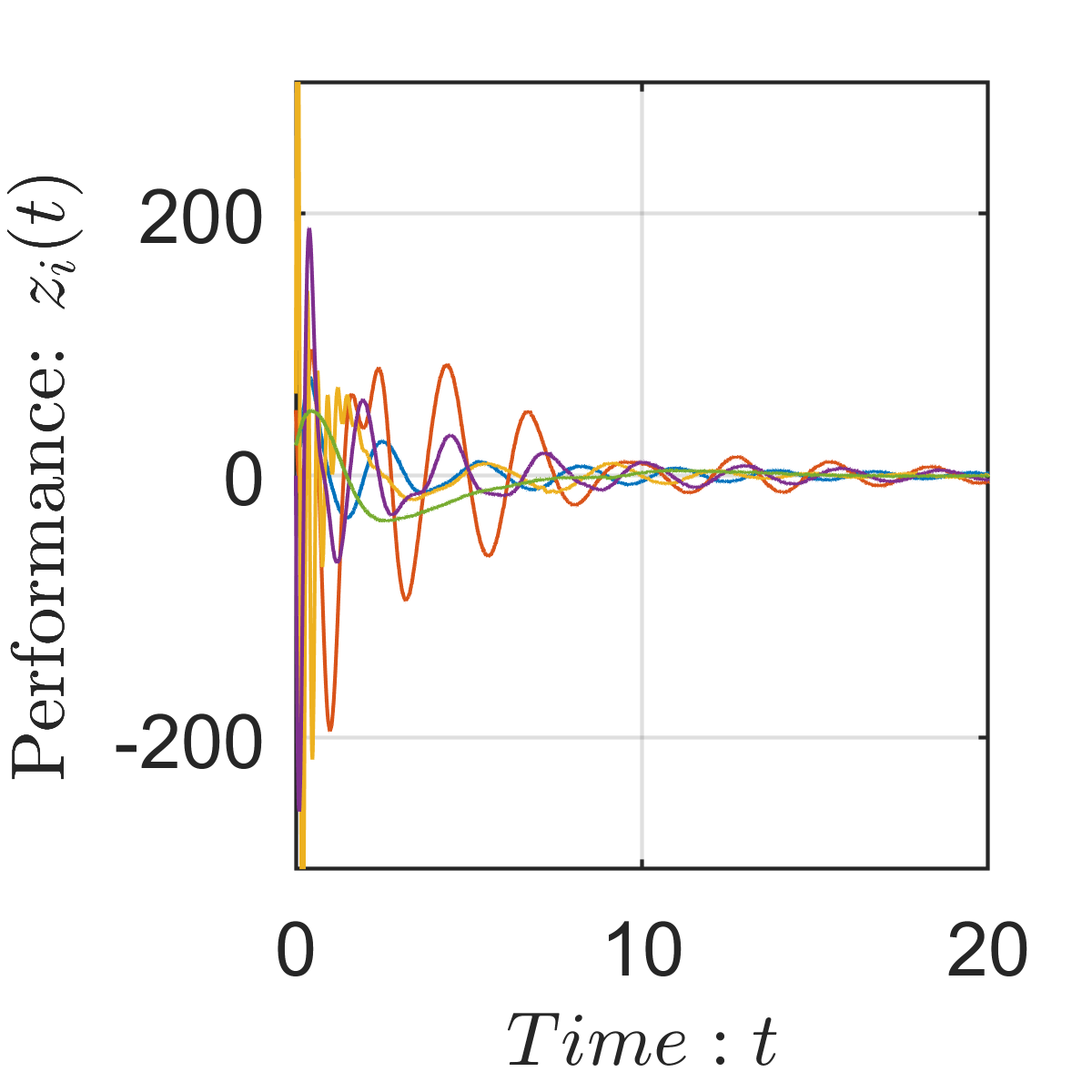}
        \caption{DOFC (derived: \\ decentrally to stabilize) \\ \centering MAP = 13.37}
        \label{Fig:}    
    \end{subfigure}
    \begin{subfigure}[b]{0.48\columnwidth}
        \centering
        \includegraphics[width=0.9\textwidth]{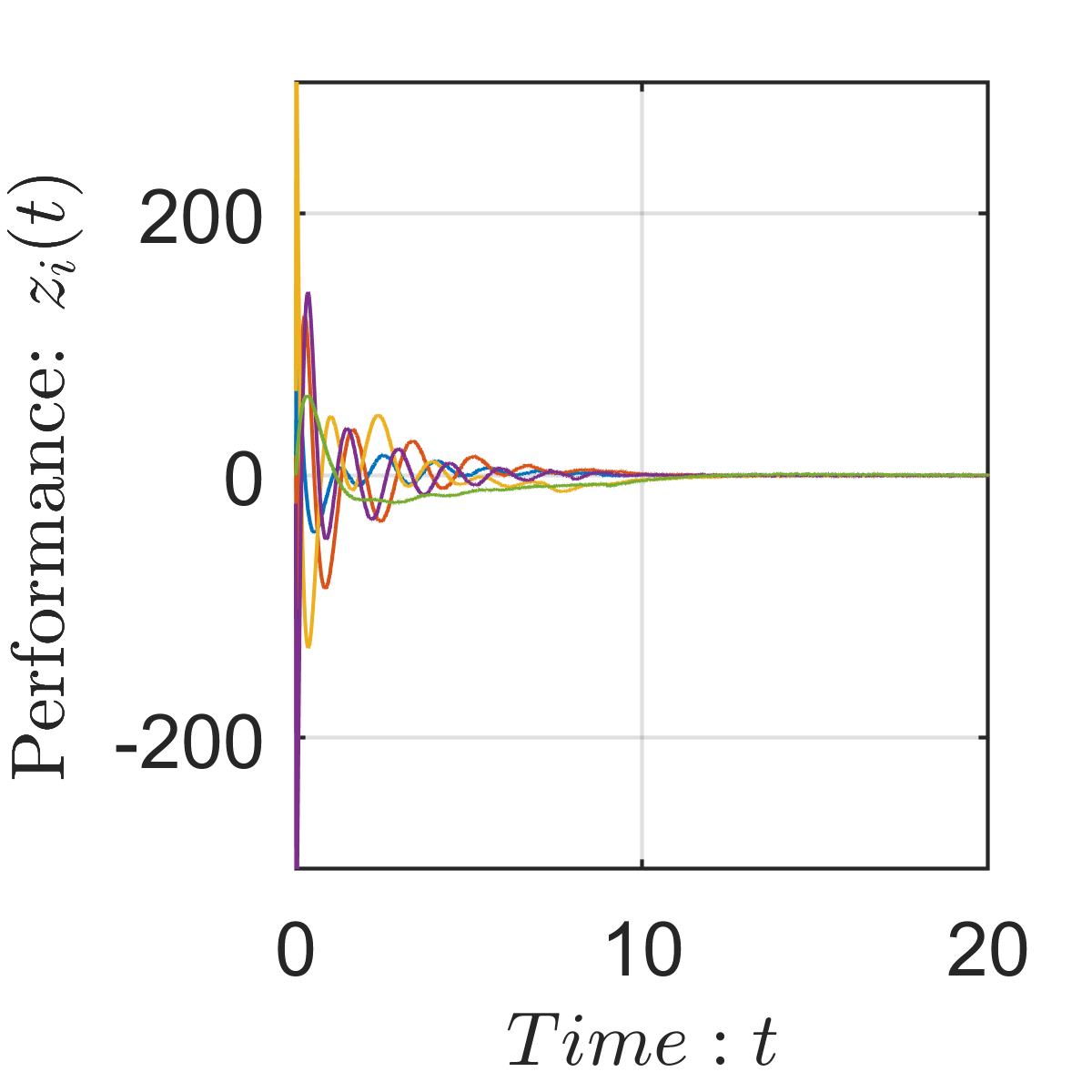}
        \caption{DOFC (derived: \\ decentrally to dissipativate) \\ \centering MAP = 7.323 (Impr.:+45.21\%)}
        \label{Fig:}    
    \end{subfigure}
    \caption{\textbf{The effect of dissipativation on performance}: Comparison of performance metric trajectories $[z_i(t)]_{i\in\N_N}$ obtained under: \textbf{stabilizing (left) vs. dissipativating (right)} controllers - synthesized both centrally (a,b,e,f) and decentrally (c,d,g,h) for SOFC (a,b,c,d) and DOFC (e,f,g,h). The observed mean absolute performance (MAP) values are given in subcaptions along with the corresponding percentage improvement values.}
    \label{Fig:SEDissipativation}
\end{figure}

\begin{figure}[!ht]
    \centering
    \begin{subfigure}[h]{0.48 \columnwidth}
        \centering
        \includegraphics[width=\textwidth]{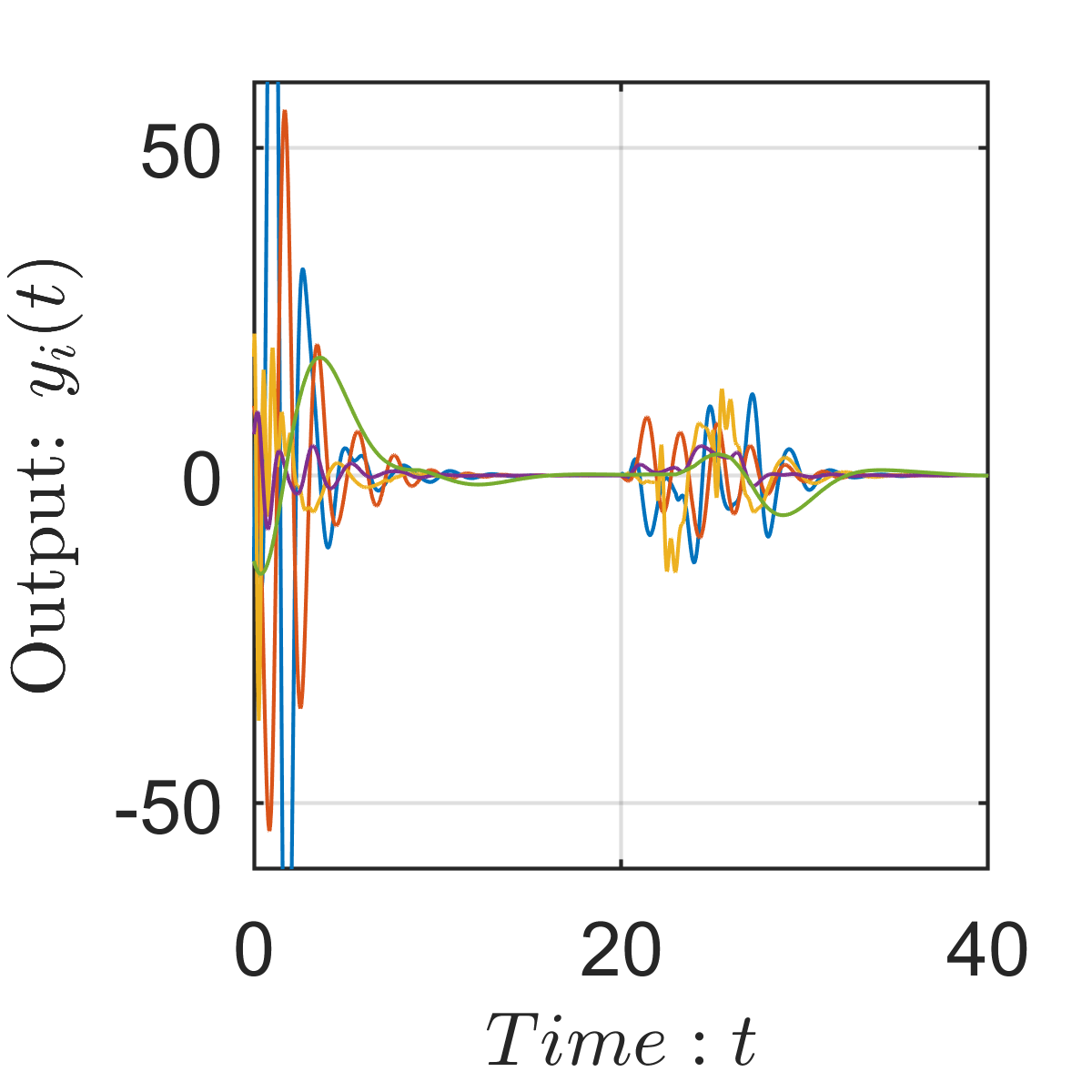}
        \caption{Output (under centrally \\ derived stabilizing DOFC) \\ \centering MAO = 2.978}
        \label{Fig:}
    \end{subfigure}
    \begin{subfigure}[h]{0.48 \columnwidth}
        \centering
        \includegraphics[width=\textwidth]{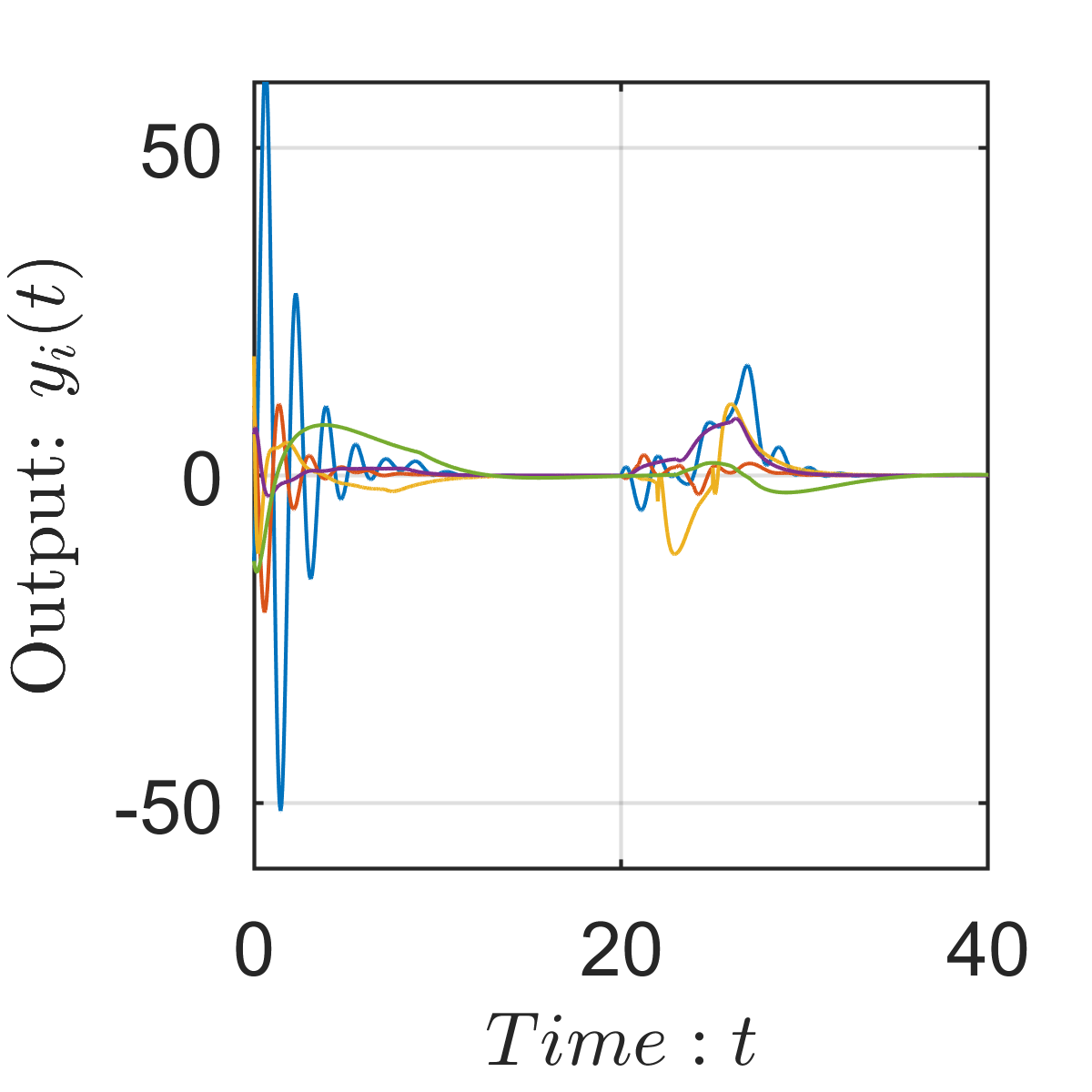}
        \caption{Output (under centrally \\ derived dissipativating DOFC) \\ \centering MAO = 1.890 (Impr.:+36.53\%)}
        \label{Fig:}
    \end{subfigure}
    \begin{subfigure}[h]{0.48 \columnwidth}
        \centering
        \includegraphics[width=\textwidth]{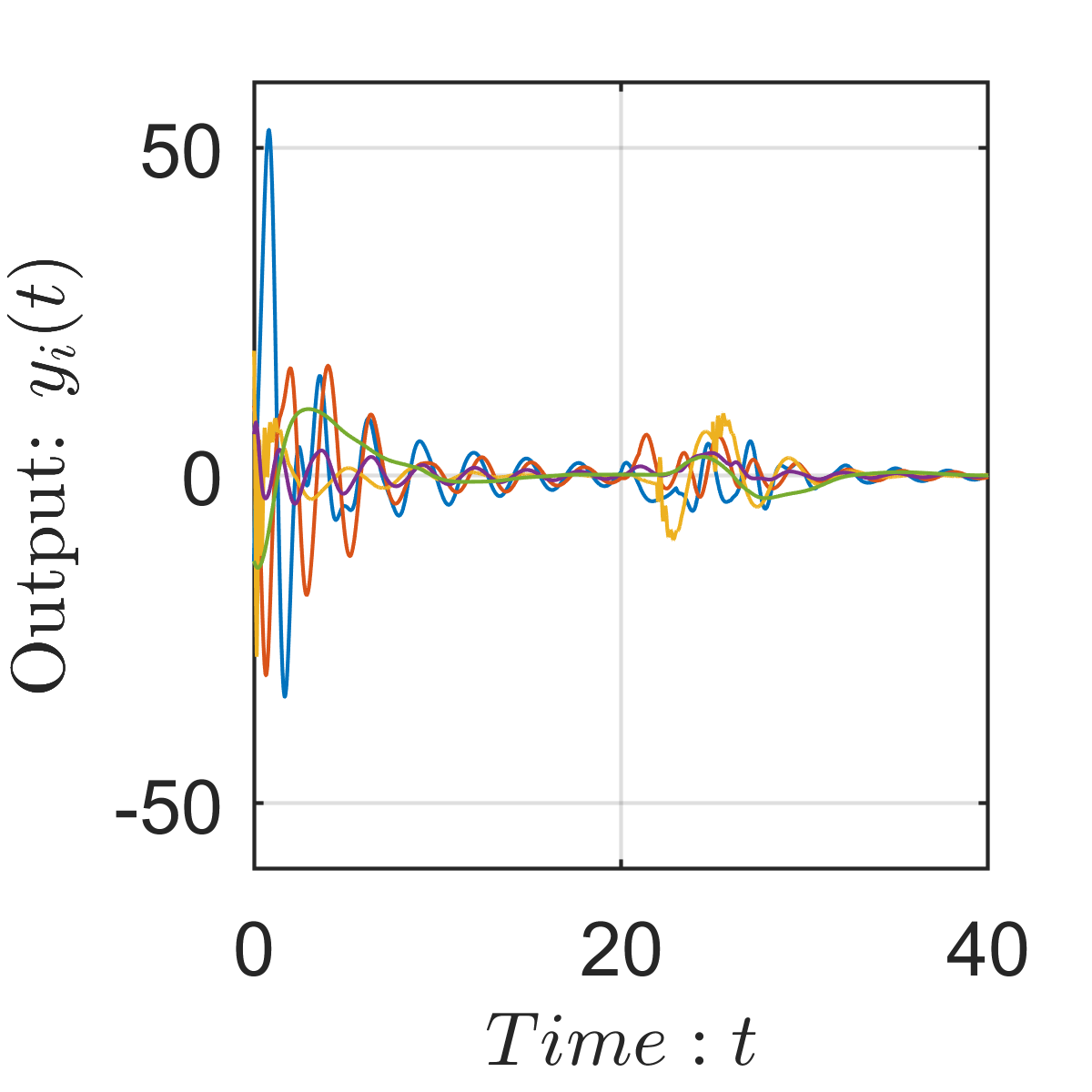}
        \caption{Output (under decentrally \\ derived stabilizing DOFC) \\ \centering MAO = 2.219}
        \label{Fig:}
    \end{subfigure}
    \begin{subfigure}[h]{0.48 \columnwidth}
        \centering
        \includegraphics[width=\textwidth]{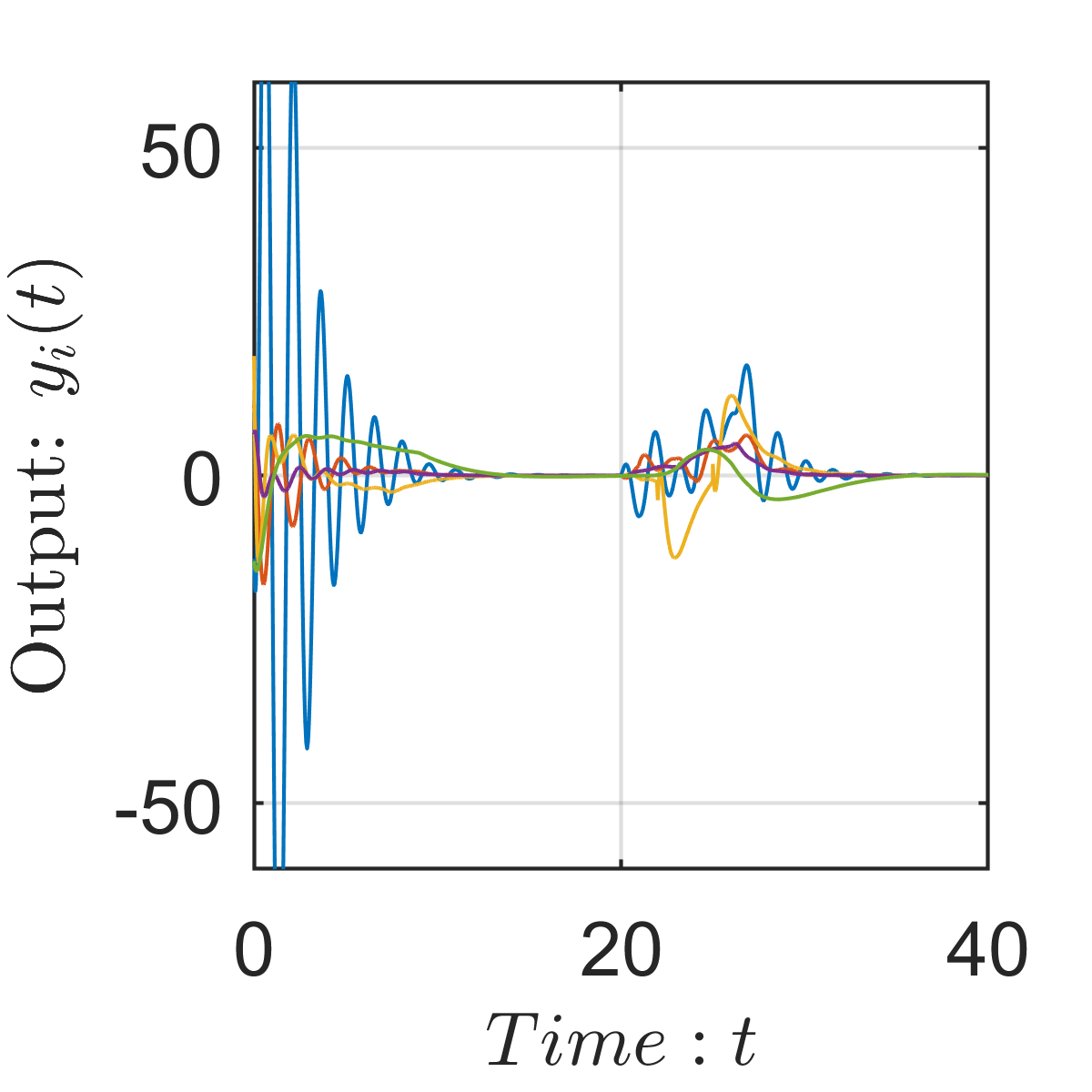}
        \caption{Output (under decentrally \\ derived dissipativating DOFC) \\ \centering MAO = 2.370 (Impr.:-6.791\%)}
        \label{Fig:}
    \end{subfigure}
    \begin{subfigure}[h]{0.48 \columnwidth}
        \centering
        \includegraphics[width=\textwidth]{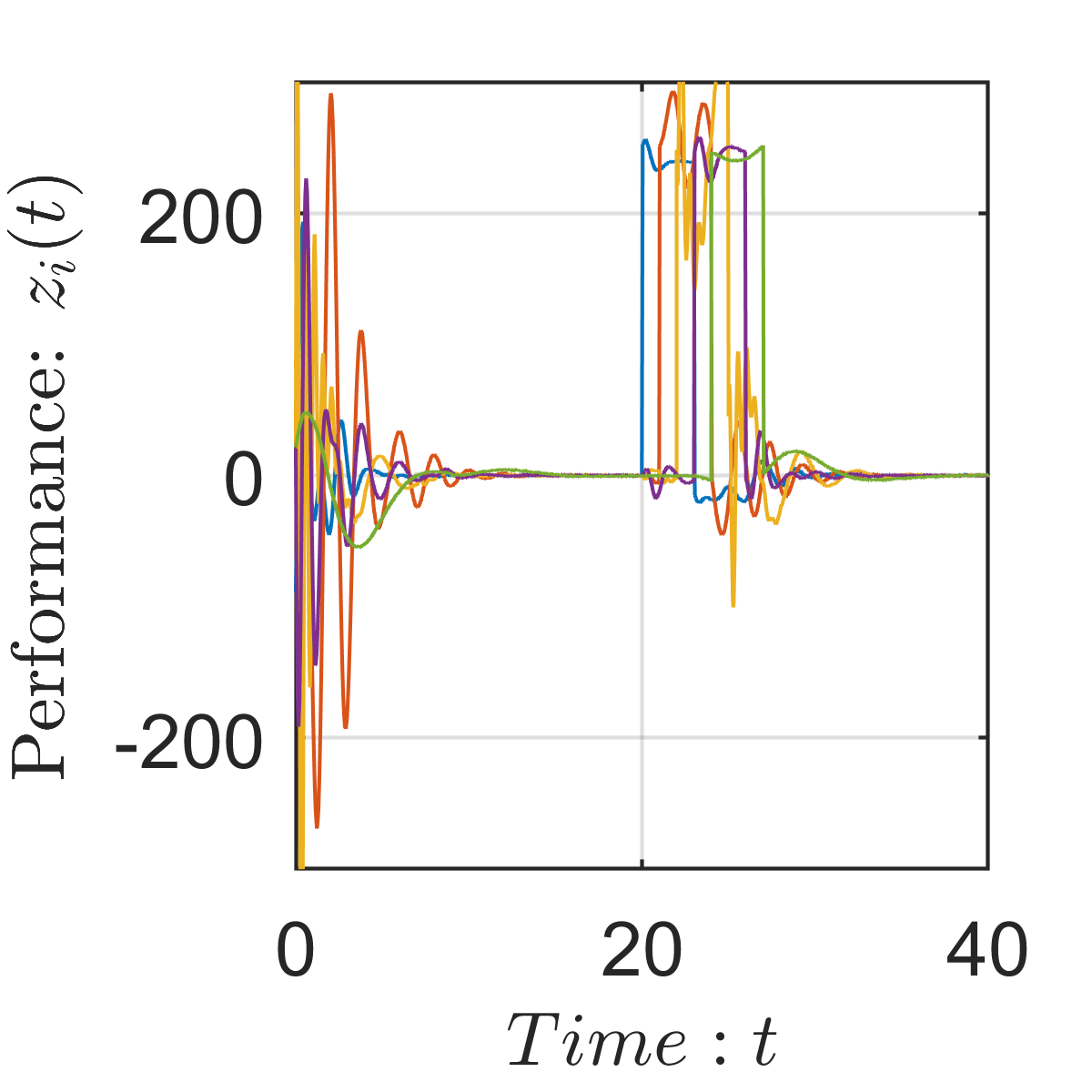}
        \caption{Performance (under centrally \\ derived stabilizing DOFC) \\ \centering MAP = 28.98}
        \label{Fig:}
    \end{subfigure}
    \begin{subfigure}[h]{0.48 \columnwidth}
        \centering
        \includegraphics[width=\textwidth]{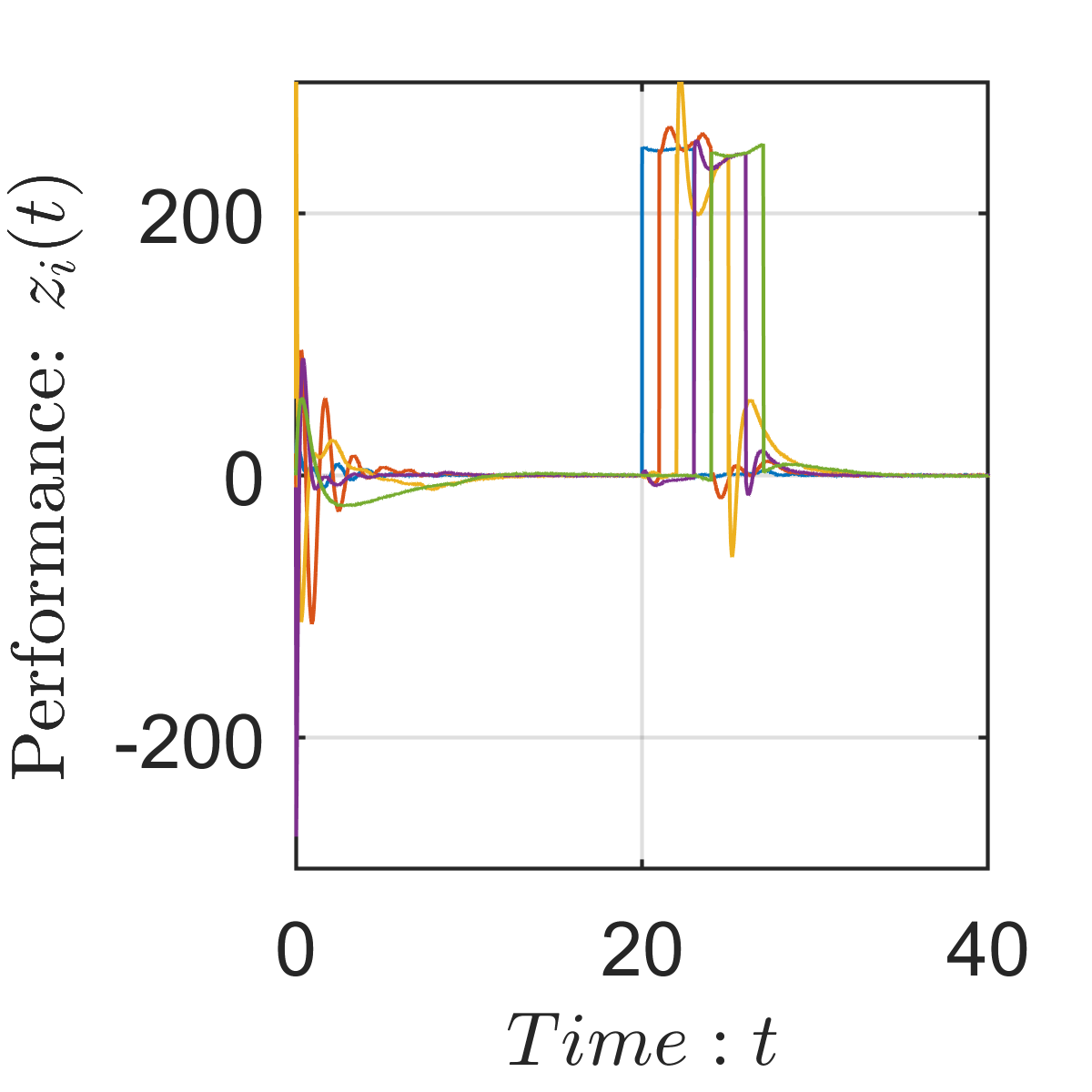}
        \caption{Performance (under centrally \\ derived dissipativating DOFC) \\ \centering MAP = 22.70 (Impr.:+21.66\%)}
        \label{Fig:}
    \end{subfigure}
    \begin{subfigure}[h]{0.48 \columnwidth}
        \centering
        \includegraphics[width=\textwidth]{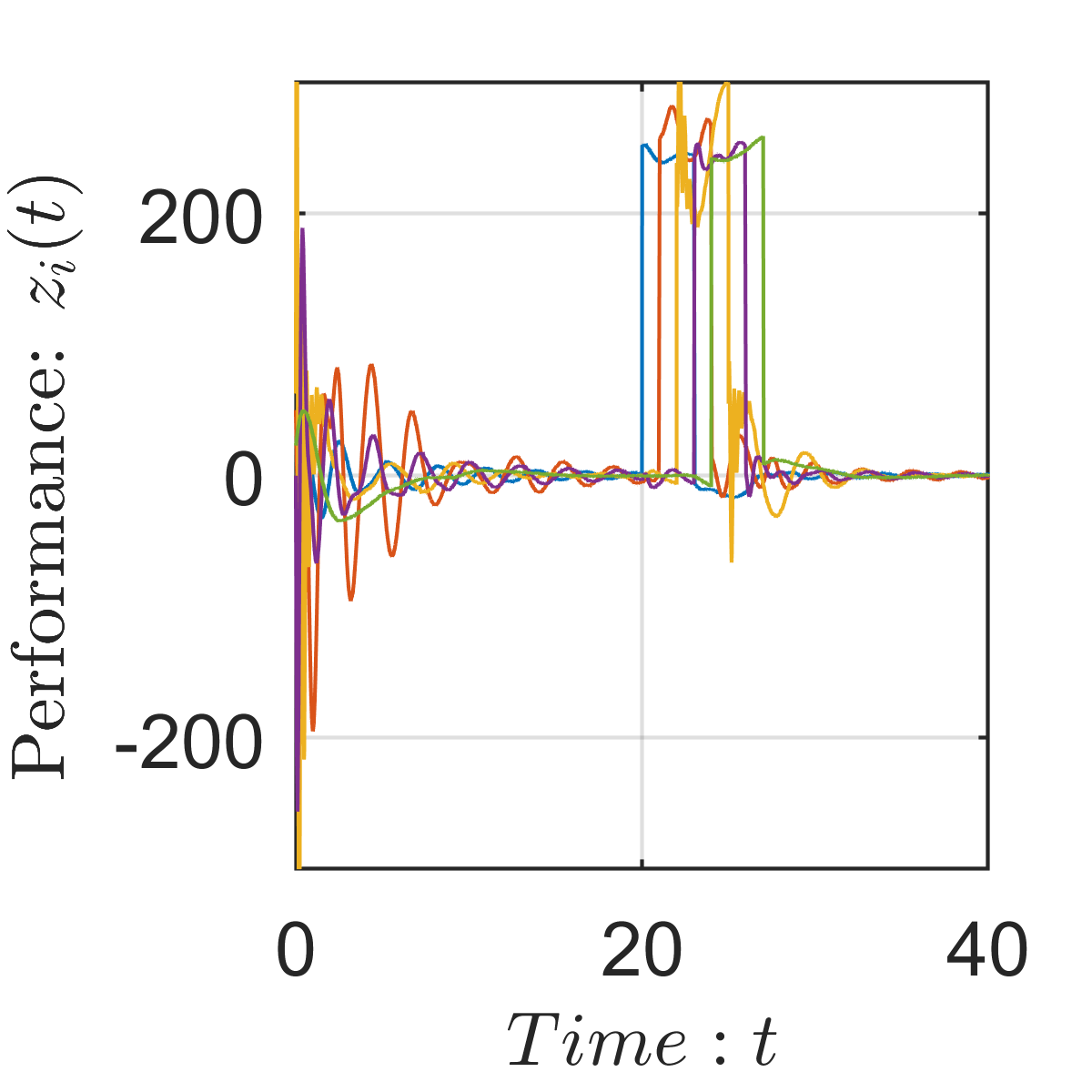}
        \caption{Perform. (under decentrally \\ derived stabilizing DOFC) \\ \centering MAP = 27.26}
        \label{Fig:}
    \end{subfigure}
    \begin{subfigure}[h]{0.48 \columnwidth}
        \centering
        \includegraphics[width=\textwidth]{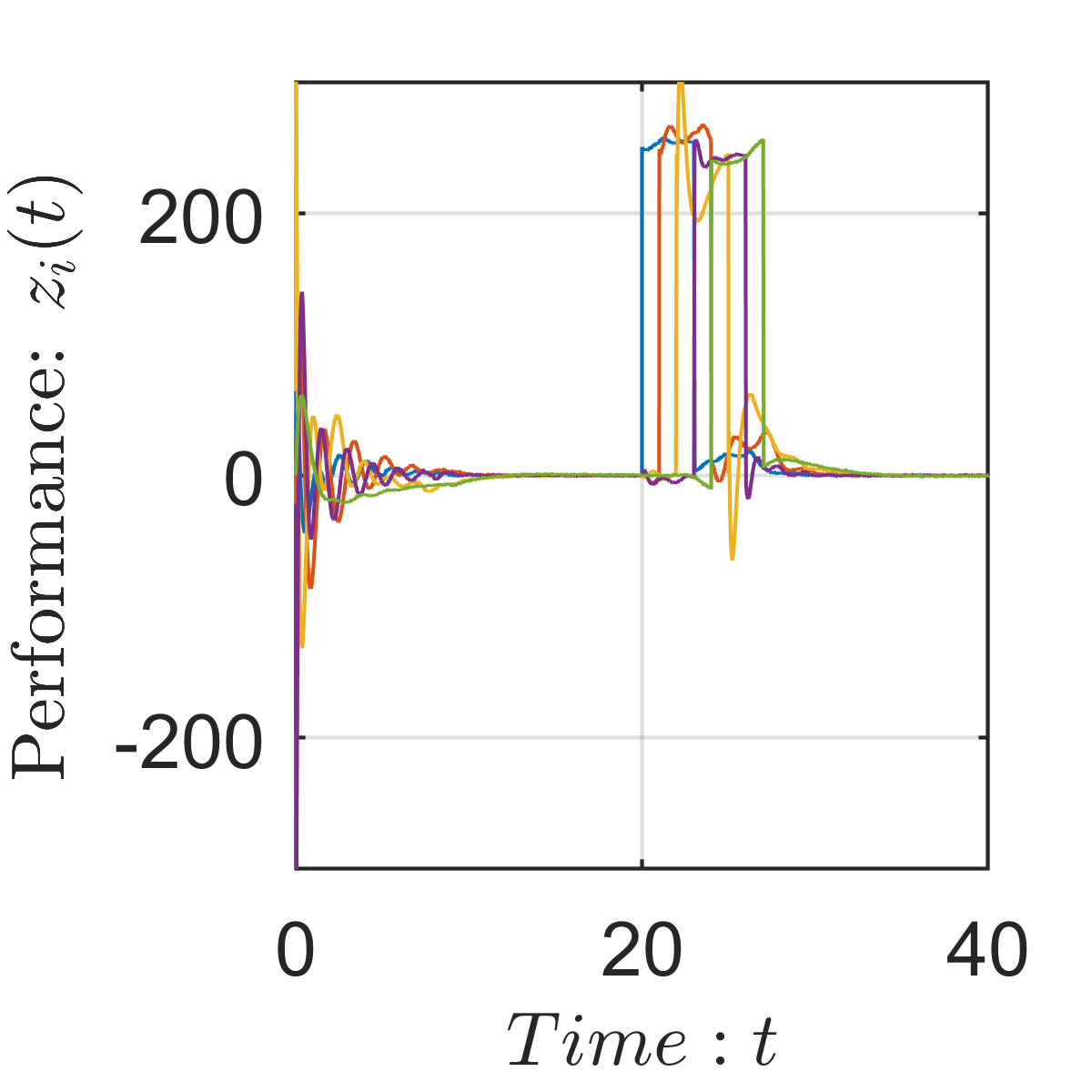}
        \caption{Perform. (under decentrally \\ derived dissipativating DOFC) \\ \centering MAP = 24.05 (Impr.:+11.81\%)}
        \label{Fig:}
    \end{subfigure}
    \caption{\textbf{The effect of dissipativation on disturbance rejection}: Comparison of output $[y_i(t)]_{i\in\N_N}$ and performance metric $[z_i(t)]_{i\in\N_N}$ trajectories obtained under: \textbf{stabilizing (left) vs. dissipativating (right)} controllers - synthesized both centrally (a,b,e,f) and decentrally (c,d,g,h) for DOFC.}
    \label{Fig:SEDisturbanceRejection}
\end{figure}

In this section, we provide a collection of numerical results obtained from a randomly generated networked system, so as to illustrate the applicability of the proposed: 
(1) decentralized analysis and control synthesis techniques (in particular, Theorems. \ref{Th:CTLTIStability}, \ref{Th:StabilizationUnderFSF}, \ref{Th:Observer} and \ref{Th:StabilizationUnderDOF}), and 
(2) $(Q,S,R)$-dissipativity based centralized (or decentralized) distributed Luenberger observer design and DOF controller design techniques, i.e., Props. \ref{Pr:DissipativeObserver} and \ref{Pr:DissipativationUsingDOF} (or Theorems \ref{Th:DissipativeObserver} and \ref{Th:DissipativationUsingDOF}), respectively.

Consider the networked dynamical system $\mathcal{G}_5$ comprised of five subsystems $\{\Sigma_i:i\in\N_5\}$ described in \eqref{Eq:ExampleNetworkedSystem} that are interconnected according to the directed network topology shown in Fig. \ref{Fig:SESetup1} (top-right). Note that, this network topology and each of the subsystem dynamic models have been generated randomly (using random geometric graphs \cite{Dall2002} and MATLAB ``rss($\cdot$)'' command, respectively). In particular, using MATLAB and YALMIP \cite{Lofberg2004}, by implementing the said random networked system generator together with many of the previously proposed centralized and decentralized  analysis and control synthesis techniques (Props. \ref{Pr:CTLTIStability}, \ref{Pr:CTLTIQSRDissipativity}, \ref{Pr:StabilizationUnderFSF}, \ref{Pr:DissipativationUnderFSF}, \ref{Pr:Observer}, \ref{Pr:DissipativeObserver}, \ref{Pr:StabilizationUnderDOF}, \ref{Pr:DissipativationUsingDOF} (centralized) and Theorems \ref{Th:CTLTIStability}, \ref{Th:CTLTIQSRDissipativity}, \ref{Th:StabilizationUnderFSF}, \ref{Th:DissipativationUnderFSF}, \ref{Th:Observer}, \ref{Th:DissipativeObserver}, \ref{Th:StabilizationUnderDOF}, \ref{Th:DissipativationUsingDOF} (decentralized)), we have developed a general software framework for analysis and control synthesis of arbitrary networked dynamical systems of the form \eqref{Eq:CTNSDynamics} (available at \url{https://github.com/shiran27/NetworkedSystemAnalysisAndControl}).  


When executing different proposed decentralized processes for the considered networked system $\mathcal{G}_5$, here we limit ourselves to the default subsystem indexing scheme $1-2-3-4-5$. The decentralized stability analysis process proposed in Th. \ref{Th:CTLTIStability} returns infeasible for $\mathcal{G}_5$ - which implies the possibility of $\mathcal{G}_5$ being unstable. This calls for stabilizing control synthesis, and as shown in Fig. \ref{Fig:SESetup2}, three distributed controller configurations are considered: (1) using local full-state feedback control \eqref{Eq:CTLocalFSFController} (labeled FSFC), (2) using local state observer \eqref{Eq:CTLocalObserver} based local state feedback control (labeled SOFC), and (3) using local dynamic output feedback control \eqref{Eq:CTLocalDOFController} (labeled DOFC). Upon synthesizing such local controllers (either centrally or decentrally, aiming to stabilize or $(Q,S,R)$-dissipativate), we use MATLAB Simulink (see Fig. \ref{Fig:SESetup}) to simulate/assess the closed-loop networked system's behavior under certain input and noise (disturbance) processes shown in Figs. \ref{Fig:SEOpenLoopu} and \ref{Fig:SEOpenLoopw}, respectively. Note also that Figs. \ref{Fig:SEOpenLoopy} and \ref{Fig:SEOpenLoopz} show the unstable open-loop output \eqref{Eq:CTNSDynamics} and performance \eqref{Eq:CTLocalDOFControllerPerf} trajectories of the networked system (i.e., of the subsystems).


\paragraph*{\textbf{Decentralization}} 
Corresponding to the aforementioned three distributed control system configurations: FSFC, SOFC and DOFC, Fig. \ref{Fig:SEDecentralization} shows the observed subsystem output profiles $[y_i(t)]_{i\in\N_N}$ under stabilizing controllers/observers derived: (1) centrally (via Props. \ref{Pr:StabilizationUnderFSF}, \ref{Pr:Observer} and \ref{Pr:StabilizationUnderDOF}, see Figs. \ref{Fig:SEDecentralization}(a,c,e)) and (2) decentrally (via Theorems \ref{Th:StabilizationUnderFSF}, \ref{Th:Observer} and \ref{Th:StabilizationUnderDOF}, Figs. \ref{Fig:SEDecentralization}(b,d,f)). Based on these observations, it is clear that decentrally designed controllers performs similar to their centrally designed counter parts. In fact, interestingly, based on the observed mean absolute output values (MAO, reported in subcaptions in Fig. \ref{Fig:SEDecentralization}), it can even be concluded that decentrally derived controllers result in smoother (less fluctuations) and faster output trajectories. A probable reason behind this observation may be the emphasis that decentralized control synthesis has on individual (characteristic) agent dynamics components (compared to that in centralized control synthesis). Before moving on, note that, in FSFC and SOFC, based on Theorems \ref{Th:StabilizationUnderFSF} and \ref{Th:Observer}, decentrally derived stabilizing controller and observer gains are given in \eqref{Eq:StabilizingControllerGains} and \eqref{Eq:StabilizingObserverGains}, respectively.

\paragraph*{\textbf{Dissipativation}}
In the sequel, by dissipativation, we simply refer to the $(Q,S,R)$-dissipativation with $Q=-0.2 \I, S=\frac{1}{2}\I, R=-0.2 \I$ (i.e., based on Rm. \ref{Rm:QSRDissipativityVariants}, strict passivation with input feedforward and output feedback passivity indices as $\nu = 0.2$ and $\rho=0.2$, respectively). As pointed out earlier, unlike stabilizing control synthesis, dissipativating control synthesis takes into account the disturbances $w_i(t), i\in\N_N$ as well as underlying/interested performance metrics $z_i(t)$ (e.g., see \eqref{Eq:CTLocalObserverPerf} and \eqref{Eq:CTLocalDOFControllerPerf}) while also ensuring stability. Therefore, it is reasonable to expect (hypothesize) dissipativating controllers/observers to provide: (1) better (lower) performance metric trajectories and (2) better robustness to disturbances (smoother output and performance profiles), compared to stabilizing controllers/observers.

In Fig. \ref{Fig:SEDissipativation}, limiting to distributed control system configurations: SOFC and DOFC, we test the aforementioned first hypothesis by comparing performance metric trajectories observed under: (1) a stabilizing controller/observer (derived centrally using Props. \ref{Pr:Observer} and \ref{Pr:StabilizationUnderDOF} or decentrally using Theorems  \ref{Th:Observer} and  \ref{Th:StabilizationUnderDOF}), and (2) a dissipativating controller/observer (derived centrally using Props. \ref{Pr:DissipativeObserver} and \ref{Pr:DissipativationUsingDOF} or decentrally using Theorems  \ref{Th:DissipativeObserver} and \ref{Th:DissipativationUsingDOF}). The fact that dissipativating controllers/observers provide superior performance metric trajectories than stabilizing controllers/observers is evident from the reported observations in Fig. \ref{Fig:SEDissipativation} - particularly from the provided mean absolute performance (MAP) values in the subcaptions. 
We point out that deriving dissipative distributed full-state feedback controllers (i.e., FSFC), both centrally and decentrally, turned out to be infeasible for the considered networked system. 
We also highlight that performance improvements due to dissipativation is more prominent when used for distributed dynamic output feedback control (i.e., DOFC) rather than distributed state observer based feedback control (i.e., SOFC).

In Fig. \ref{Fig:SEDisturbanceRejection}, limiting to the distributed control system configuration: DOFC, we test the previously mentioned second hypothesis: dissipativating controllers have superior disturbance rejection qualities as opposed to stabilizing controllers. For this purpose, we first superimpose each subsystem disturbance signal $w_i(t), i\in\N_N$ shown in Fig. \ref{Fig:SEOpenLoopw} with a single square pulse (of width $3\, s$) occurring between $t \in[20,27]$ as indicated in Fig. \ref{Fig:SESetup1} (the exact form of this pulse is indicative from Figs. \ref{Fig:SEDisturbanceRejection}(e,f,g,h)). The resulting output and performance trajectories obtained under centrally and decentrally derived: (1) stabilizing and (2) dissipativating DOF controllers are shown in Fig. \ref{Fig:SEDisturbanceRejection}. According to these observations (also using the reported MAO and MAP values in Fig. \ref{Fig:SEDisturbanceRejection}), it is clear that dissipativating controllers, compared to stabilizing controllers, render smoother and faster output and performance trajectories in the face of significant disturbances.


\section{Conclusion}
\label{Sec:Conclusion}
Starting from reviewing existing and new LMI-based control solutions for LTI systems, we presented several decentralized analysis and control synthesis techniques to verify and ensure properties like stability and dissipativity of large-scale networked systems. We considered a substantially more general problem setup than state of the art and developed decentralized processes covering a broader range of properties of interest. The synthesized control laws are distributed, and the proposed analysis and control synthesis processes themselves are decentralized, compositional and resilient to subsystem removals. We also have shown that optimizing the indexing scheme used in such distributed processes can substantially reduce the required information-sharing sessions between subsystems and, in some cases, even make the overall process distributed. Moreover, we have derived novel centralized LMI-based solutions for dissipative local observer design and dissipative local dynamic output feedback controller design problems along with their decentralized counterparts. Subsequently, we specialized all the derived results for discrete-time networked systems and provided several simulation examples to demonstrate the proposed novel decentralized analysis and control synthesis processes and dissipativity-based results. Future work aims to study the effect of erroneous and failed information sharing sessions among subsystems and develop robust distributed analysis and control synthesis approaches for such scenarios.



\bibliographystyle{IEEEtran}
\bibliography{References}

\end{document}